\documentclass{ptephy_v1}

\usepackage{times}

\usepackage[colorlinks=true,linktoc=page,citecolor=red,linkcolor=blue]{hyperref}	
\usepackage{graphicx}
\graphicspath{{./figs/}}  
\usepackage{amsmath,amssymb,amsfonts,amsthm,amscd}
\usepackage{bm,braket,array}
\usepackage{multirow}
\usepackage[all]{xy}
\usepackage{slashed}

\newtheorem{thm}{Theorem}[section]
\newtheorem{prop}[thm]{Proposition}
\newtheorem{lem}[thm]{Lemma}
\newtheorem{cor}[thm]{Corollary}

\theoremstyle{definition}

\theoremstyle{remark}

\def\im{\mbox{Im }}
\def\mod{{\rm\ mod\ }}

\def\tr{{\rm tr\,}}

\def\p{\partial}

\def\wt{\widetilde}
\def\ker{\mbox{Ker }}

\def\ua{\uparrow}
\def\da{\downarrow}

\newcommand{\C}{\mathbb{C}}

\newcommand{\R}{\mathbb{R}}
\newcommand{\Z}{\mathbb{Z}}

\def\widebar{\accentset{{\cc@style\underline{\mskip10mu}}}} 
\def\wideubar{\underaccent{{\cc@style\underline{\mskip10mu}}}} 

\begin{document}

\title{Generalized homology and Atiyah-Hirzebruch spectral sequence in crystalline symmetry protected topological phenomena}


\author{Ken Shiozaki}
\affil{Yukawa Institute for Theoretical Physics, Kyoto University, Kyoto 606-8502, Japan \email{ken.shiozaki@yukawa.kyoto-u.ac.jp}}
\author{Charles Zhaoxi Xiong}
\affil{Department of Physics, Harvard University, Cambridge, Massachusetts 02138, USA}
\author{Kiyonori Gomi}
\affil{Department of Mathematics, Tokyo
Institute of Technology,2-12-1 Ookayama, Meguro-ku, Tokyo, 152-8551, Japan}
%
%
%

\begin{abstract}%
We propose that symmetry-protected topological (SPT) phases with crystalline symmetry are formulated by an equivariant generalized homology $h^G_n(X)$ over a real space manifold $X$ with $G$ a crystalline symmetry group. 
The Atiyah-Hirzebruch spectral sequence unifies various notions in crystalline SPT phases, such as the layer construction, higher-order SPT phases, and Lieb-Schultz-Mattis-type theorems. 
This formulation is applicable to not only free fermionic systems but also interacting systems with arbitrary on-site and crystal symmetries.
\end{abstract}


\maketitle

\setcounter{tocdepth}{2}

\section{Introduction}
Symmetry-protected topological (SPT) phases are topologically distinct phases of gapped, invertible states with symmetry~\cite{SPT_origin, Wen_Definition, Cirac,  CGLW13, LG12, PTBO10, CGW11, LV12, GW14, VS13, SRFL08, Kit09, FM13, FK13, WPS14, Kap14, KTTW15, Freed_invertible_iTQFT, FH16, Wit15, 3dBTScBurnell,2dFermionGExtension}.
A quantum state is said to be invertible when it is realized as the unique gapped ground state of a Hamiltonian constructed of local terms for any system size and closed space manifold. 
Two invertible states are said to belong to the same SPT phase if there exists a continuous path of Hamiltonians interpolating between them that preserves both the symmetry and the energy gap~\cite{Wen_Definition}.
For SPT phases protected solely by onsite symmetry, the classification scheme is well-developed. 
The group (super) cohomology theory produces explicit lattice Hamiltonians and topological actions of nonlinear $\sigma$ models in discretized Euclidean spacetime~\cite{CGLW13,GW14}, so that gauging the onsite symmetry yields twisted discrete pure gauge theories~\cite{DW90,LG12}. 
The cobordism group \cite{Kap14,KTTW15,Kapustin_equivariant,Yonekura18} and other invertible field theory invariants \cite{Freed_invertible_iTQFT,FH16} give classifications of low-energy effective response actions of SPT phases.
A physical feature of nontrivial SPT phases is the quantum anomaly in boundary states: in the presence of a real-space boundary, the topological response action is not invariant under gauge transformations. 
This is a manifestation of the \'{}t Hooft anomaly, which can only be canceled by the low-energy degrees of freedom living on the boundary of an SPT phase with the opposite anomaly.

Soon after SPT phases were proposed, it was realized that crystalline symmetries \cite{TeoFuKane08,Fu11} can also serve to protect nontrivial phases of invertible systems, giving rise to the notation of crystalline SPT phases~\cite{CGW11, Cirac, Wen_sgSPT_1d, Hsieh_sgSPT, Hsieh_sgSPT_2, You_sgSPT, SPt, Cho_sgSPT, Yoshida_sgSPT, Hermele_torsor, Jiang_sgSPT, TE18}.
Two systematic approaches to classifying crystalline SPT phases have been proposed. 
Song et al. constructed crystalline SPT phases by placing lower-dimensional invertible states at high-symmetry points and demonstrated how classification could be obtained by considering trivialization of invertible states on lower-dimensional cells by invertible states in adjacent higher-dimensional cells~\cite{Hermele_torsor}. 
Later, this idea was applied to the classification of Lieb-Schultz-Mattis (LSM) theorems through bulk-boundary correspondence~\cite{CGW11,PWJZ17} and a comprehensive classification of 3-dimensional bosonic SPT phases with space group symmetries~\cite{HSHH17}. 
Thorngren and Else~\cite{TE18} proposed a useful picture for interpreting crystalline symmetries as gauge fields. 
They contended if the spatial scale of a crystalline symmetry (such as the lattice vector) is large compared to the correlation length, that the crystalline symmetry will behave in a topological quantum field theory (TQFT) as if it is an onsite symmetry, hence allowing one to recast the classification of crystalline SPT phases as a classification of onsite SPT phases.

The notion of higher-order SPT phases is of importance to crystalline systems and highlights the difference between crystalline SPT phases and SPT phases protected only by onsite symmetries~\cite{BenalcazarScience,BenalcazarPRB,FF17,SchindlerHigher,SongFangFangHigher17,
LangbehnHigher17,FangMapping,KhalafPRX18AII,KhalafPRB18Higher,LukaPRB18SecondOrder,TB18}. 
An $n$-th order SPT phase in $d$ space dimensions is an SPT phase that exhibits an anomalous surface state localized on a $(d-n)$-dimensional spatial edge on the boundary. 
Historically, higher-order SPT phases were first discovered in free fermionic systems and interpreted in known phases from the perspective of layer construction. 
Huang et al.~\cite{HSHH17} and Trifunovic and Brouwer~\cite{TB18} found a fundamental structure behind crystalline SPT phases: there is a filtration $0 \subset F_0 h_n^G \subset \cdots \subset F_d h_n^G = h_n^G(X)$ for the classification of crystalline SPT phases (the notation here is introduced in Sec.~\ref{sec:ahss}) with respect to the space dimension for which the layer construction is defined so that the quotient $F_p h^G_n / F_{p-1} h^G_n$ is the classification of $(d-p+1)$-dimensional higher-order SPT phases. 
In Sec.~\ref{sec:ahss}, we show that the filtration structure holds true for SPT phases with many-body interactions and demonstrate how it may be properly formulated.

Kitaev~\cite{Kitaev_Stony_Brook_2011_invertible_1, Kit13, Kit15} pointed out that the physical properties of invertible states imply that the ``spaces of invertible states" of different dimensions form an $\Omega$-spectrum in the sense of generalized cohomology, which gives us a general mathematical framework for describing SPT phases~\cite{Xio17,CharlesAris18,GT17}. 
The nature of the microscopic degrees of freedom (i.e., bosons versus fermions) is encoded in the $\Omega$-spectrum. That is, there is one $\Omega$-spectrum for bosonic SPT phases and another $\Omega$-spectrum for fermionic SPT phases.
Kitaev proposed that the classification of $d$-dimensional SPT phases with onsite $G_i$ symmetry is given by the generalized cohomology $h^d(BG_i)$, where $h^*(-)$ is the generalized cohomology theory associated with the $\Omega$-spectrum and $BG_i$ is the classifying space of $G_i$.

More generally, one can show, by the same argument, that the spaces of invertible states form an $\Omega$-spectrum even if we restrict attention to subspaces that respect a given onsite symmetry (possibly anti-unitary)~\cite{CharlesAris18,GT17}. 
Taking this as our fundamental assumption, we propose that the classification of SPT phases over a real space $X$ protected by crystalline symmetry $G$ is given by $h^G_n(X)$, where $h^G_n$ is the equivariant generalized {\it homology} theory defined the said $\Omega$-spectrum, and the symmetry group $G$ acts both on the real space $X$ and the $\Omega$-spectrum. 
To support this proposal, we will apply the Atiyah-Hirzeburch spectral sequence (AHSS) to the generalized homology theory $h^G_n$ and show that the various terms and differentials in the spectral sequence have concrete physical meanings. In particular, we will see that the technology developed by Song et al.~\cite{Hermele_torsor} is nothing but the first differential in the AHSS.

The AHSS is one of the many spectral sequences commonly used in the mathematical literature to compute generalized (co)homology groups $h_*(X)$ (resp.\ $h^*(X)$).~\cite{AH}
The first step in applying the AHSS is to collect the local topological data for each cell in a cell-decomposition of the space $X$. 
Then, we compare the local data in cells of different dimensions and glue them together appropriately -- these procedures are known as ``differentials" in the language of spectral sequence -- to obtain global information.
Iterating differentials, we finally get the so-called $E^{\infty}$-page (resp.\ $E_{\infty}$-page), which approximates the generalized (co)homology theory $h^G_*(X)$ (resp.\ $h_G^*(X)$) in the manner of a filtration. 
The filtration is effectively a set of short exact sequences that need to be solved. 
Solving these short exact sequences, we will then be able to determine the generalized (co)homology group $h^G_*(X)$ (resp.\ $h_G^*(X)$). 
We emphasize, for crystalline SPT phases, that differentials and group extensions in the AHSS have well-defined physical meanings even in higher orders, lending themselves to full determination using our understanding of the physical properties of the relevant phases.

The organization of this paper is as follows. 
Sec.~\ref{sec:spectrum} is devoted to introducing the viewpoint of the $\Omega$-spectrum, which Kitaev pointed out, and its connection to the adiabatic pump and an emergent invertible state trapped on a texture. 
In Sec.~\ref{sec:genehomo}, we define the equivariant generalized homology for a given $\Omega$-spectrum of invertible states. 
Some generic properties of a generalized homology are given. 
We formulate the AHSS in Sec.~\ref{sec:ahss} in detail. 
We collect physical implications of the AHSS in Sec.~\ref{sec:phys_ahss}. 
The connection to the higher-order SPT phases and various LSM-type theorems are described there. 
Sec.~\ref{sec:interacting} and \ref{sec:free} present various examples of the computation of the AHSS. 
Sec.~\ref{sec:interacting} is for interacting crystalline SPT phases. 
For free fermions, the generalized homology is recast as the $K$-homology. 
We describe the AHSS for free fermions in Sec.~\ref{sec:free}. 
We summarize this paper and suggest future directions in Sec.~\ref{sec:conc}. 

\noindent
{\it Notation---}
In this paper, ``Inv$^d$" and ``SPT$^d$" mean ``$d$-dimensional invertible" and ``$d$-dimensional SPT", respectively. 
Unless stated otherwise, dimension always refers to the spatial dimension.

\section{Invertible states as a spectrum}
\label{sec:spectrum}
The essence of the generalized cohomology approach for SPT phases by Kitaev is the following homotopy equivalence relation among invertible states with different space dimensions ~\cite{Kit13}
\begin{align}
F_d \cong \Omega F_{d+1}. 
\label{eq:spectra}
\end{align}
Here, $F_d$ is the based topological space consisting of Inv$^d$ states protected by onsite symmetry, where the base point is the trivial tensor product state denoted by $\ket{1}$. 
$\Omega F_{d+1} = \left\{\ell: S^1 \to F_{d+1} | \ell(0)=\ell(1)=\ket{\rm 1} \right\}$ is the based loop space of $F_{d+1}$, the Inv$^{d+1}$ states.
The classification of SPT$^d$ phases is given by disconnected parts of $F_d$, the generalized cohomology $h^d(pt) = [pt,F_d]$.

\subsection{On the homotopy equivalence $F_d \cong \Omega F_{d+1}$}
\label{sec:1d_pump}
\begin{figure}[!]
\begin{center}
\includegraphics[width=\linewidth, trim=0cm 2cm 0cm 2cm]{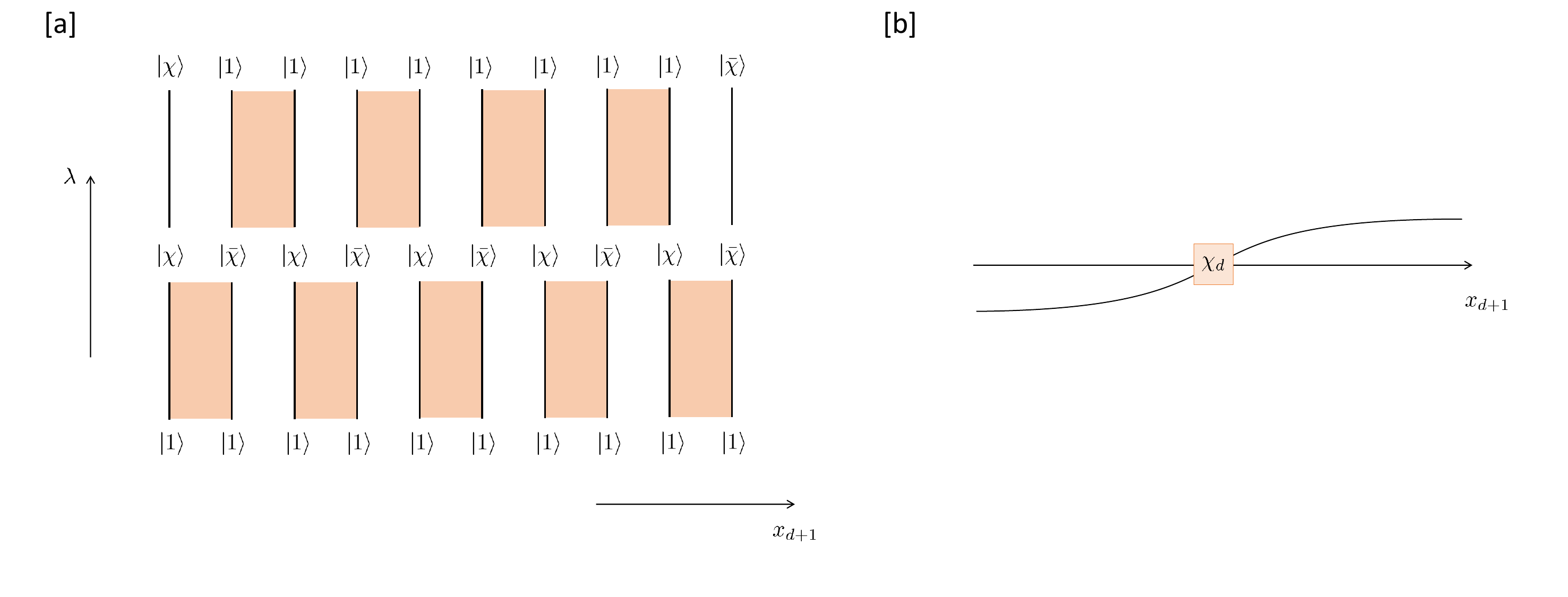}
\end{center}
\caption{[a] An adiabatic cycle in Inv$^{d+1}$ states labeled by an Inv$^d$ state $\ket{\chi}$.
[b] An Inv$^d$ state localized as the texture of Inv$^{d+1}$ states. 
}
\label{fig:spectra}
\end{figure}
The relation (\ref{eq:spectra}) states that an adiabatic cycle in Inv$^{d+1}$ states which begins and ends at the trivial state is uniquely labeled by an Inv$^d$ state (up to homotopy equivalence). 
Especially, the topological classification of such adiabatic processes is given by $[S^1,F_{d+1}]_* = [pt,\Omega F_{d+1}] = [pt,F_d] = h^d(pt)$, the classification of SPT$^d$ phases.
Kitaev provided a canonical construction of the map $f: F_d \to \Omega F_{d+1}$~\cite{Kit13}, which is best understood from Fig.~\ref{fig:spectra}[a] and described below. 
Let $\ket{\chi} \in F_d$ be an Inv$^d$ state and $\ket{\bar \chi} \in F_d$ be its conjugate so that the tensor product of them is adiabatically equivalent to a trivial state $\ket{\chi} \otimes \ket{\bar \chi} \sim \ket{1} \otimes \ket{1}$. 
The existence of the ``inverse" in this sense is a characteristic of invertible states.  
At the initial time $\lambda = 0$, we consider the layer of trivial states $\Ket{\lambda=0} = \bigotimes_{j \in \Z} \ket{1}_{j}$ along the $x_{d+1}$-direction. 
During the adiabatic time evolution in a half period, for a given Inv$^d$ state $\ket{\chi} \in F_d$, we take the adiabatic deformation $\ket{1}_{2j-1} \otimes \ket{1}_{2j} \sim \ket{\chi}_{2j-1} \otimes \ket{\bar \chi}_{2j}$ for adjacent two sites $2j$ and $2j+1$ to get the Inv$^{d+1}$ state 
\begin{align}
\Ket{\lambda = \frac{1}{2}}
= \bigotimes_{j \in \Z} \left( \ket{\chi}_{2j-1} \otimes \ket{\bar \chi}_{2j} \right). 
\end{align}
Next, we take the inverse adiabatic deformation $\ket{\bar \chi}_{2j} \otimes \ket{\chi}_{2j+1} \sim \ket{1}_{2j} \otimes \ket{1}_{2j+1}$ with switching bipartite sublattices. 
The resulting Inv$^{d+1}$ state goes back to the trivial state $\Ket{\lambda = 1} = \Ket{\lambda=0}$, which achieves an adiabatic cycle labeled by an Inv$^d$ state $\ket{\chi}$ which begins at the trivial Inv$^{d+1}$ state. 
For an open system in the $x_{d+1}$-direction, an adiabatic cycle pumps the Inv$^d$ state $\ket{\chi}$ from the right to the left boundary (See Fig.~\ref{fig:spectra} [a]). 

The condition of adiabatic cycles so that the initial state is the trivial Inv$^{d+1}$ state is crucial to have a canonical construction of $F_d \to \Omega F_{d+1}$. 
In fact, for the free loop space ${\cal L} F_{d+1} = {\rm Map}(S^1,F_{d+1})$, the initial state may be topologically nontrivial, and since topologically nontrivial Inv$^{d+1}$ states can not be decomposed into the tensor product of Inv$^d$ states on sites $j \in \Z$ we do not expect a canonical construction of adiabatic cycles. 

To verify the homotopy equivalence (\ref{eq:spectra}), it is required an inverse map $g: \Omega F_{d+1} \to F_d$ such that $g \circ f$ and $f \circ g$ are homotopically equivalent to the identity maps. 
A canonical construction of the inverse map $g$ is not known yet but is understood as an Inv$^d$ state trapped on a texture of a family of Inv$^{d+1}$ states in between a trivial Inv$^{d+1}$ state. 
Let $H(\lambda \in [0,1])$ be an adiabatic Hamiltonian realizing the adiabatic cycle $\Ket{\lambda \in [0,1]}$. 
With this, we have a semiclassical Hamiltonian $H(\lambda = x_{d+1})$ that slowly varies compared with the correlation length along the $x_{d+1}$-direction.
Nontrivial cycle labeled by the Inv$^d$ state $\ket{\chi}$ implies the existence of the localized Inv$^d$ state $\ket{\chi}$ at the texture represented by the Hamiltonian $H(x_{d+1})$ (See Fig.~\ref{fig:spectra}[b]). 
A simple example is the following tensor product state 
\begin{align}
\cdots 1111 \chi \bar \chi \chi \bar \chi \cdots \chi \bar \chi \chi \bar \chi \chi \cdots \bar \chi \chi \bar \chi \chi 1111 \cdots, 
\label{eq:1d_texture}
\end{align}
where there exists a single Inv$^d$ state $\ket{\chi}$ per a texture.

\subsection{On the homotopy equivalence $F_d \cong \Omega^2 F_{d+2}$}
\begin{figure}[!]
\begin{center}
\includegraphics[width=0.8\linewidth, trim=0cm 0cm 0cm 0cm]{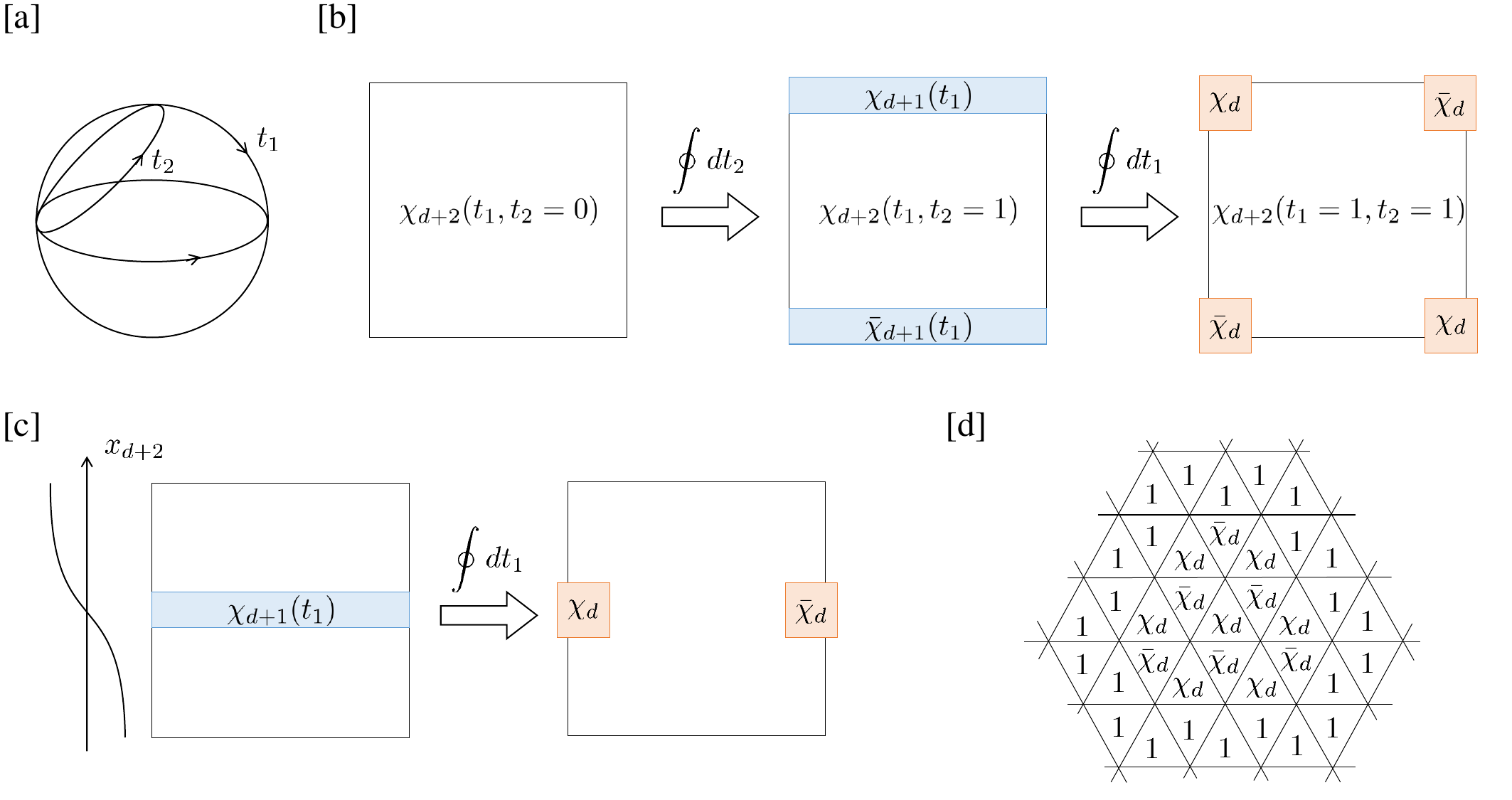}
\end{center}
\caption{
[a] Adiabatic parameters $(t_1,t_2)$ living in the 2-sphere. 
[b] The subsequent adiabatic pumps for the $t_2$ and $t_1$ time directions. 
[c] Adiabatic pump of the $t_1$ direction with a texture. 
[d] A texture-induced Inv$^d$ state.
}
\label{fig:spectra2}
\end{figure}

It would be instructive to see the physical meaning of the iterated based loop space 
\begin{align}
F_d \cong \Omega F_{d+1} \cong \Omega^2 F_{d+2} = \{ \ell: S^2_* \to F_{d+2} | \ell(*) = \ket{1} \}.
\end{align}
An element of $\Omega^2 F_{d+2}$ represents a two-parameter adiabatic cycle. 
By applying the canonical construction of $F_d \to \Omega F_{d+1}$ twice, we have that for $F_d \cong \Omega^2 F_{d+2}$. 
For a given Inv$^d$ state $\ket{\chi_d}$, we have an adiabatic cycle $\ket{\chi_{d+1}(t_1)}$ moving in topologically trivial Inv$^{d+1}$ states so that a one period of adiabatic cycle pumps the state $\ket{\chi_d}$. 
Also, at each time slice $t_1$, there is an adiabatic cycle $\ket{\chi_{d+2}(t_1,t_2)}$ moving in topologically trivial Inv$^{d+2}$ states so that a one period of the second parameter $t_2$ gives a pump of the state $\ket{\chi_{d+1}(t_1)}$. 
The parameter space $(t_1,t_2)$ becomes the 2-sphere (see Fig.~\ref{fig:spectra2} [a]) compactified to the trivial Inv$^{d+1}$ state. 
As in the case of adiabatic pumps with a single time, one can see that on a rectangular system, subsequent adiabatic cycles of $t_1$ and $t_2$ give the pump of the Inv$^d$ state $\ket{\chi_d}$: 
Let us consider topologically trivial Inv$^{d+2}$ states $\ket{\chi_{d+2}(t_1,t_2)}$ parameterized by $S^2$ on a rectangular open system along the $x_{d+1}$- and $x_{d+2}$-directions as shown in Fig.~\ref{fig:spectra2} [b]. 
First we take an adiabatic cycle for $t_2$, resulting in the topologically trivial Inv$^{d+1}$ states $\ket{\chi_{d+1}(t_1)}$ and its conjugate $\ket{\bar \chi_{d+1}(t_1)}$ on the boundary. 
Then, an adiabatic cycle of $t_1$ gives the Inv$^d$ states $\ket{\chi_d}$ and $\ket{\bar \chi_d}$ at the four corners. 

Adiabatic parameters can be replaced by semiclassical variables representing a texture in real space.
When $t_2$ is replaced by the semiclassical parameter slowly depending on the $x_{d+2}$-direction, an adiabatic deformation at the time $t_1$ corresponds to the topologically trivial Inv$^{d+1}$ state $\ket{\chi_{d+1}(t_1)}$ localized at the kink along the $x_{d+2}$-direction. 
The adiabatic cycle of the time $t_1$ gives the localized Inv$^{d}$ states $\ket{\chi_d}$ and $\ket{\bar \chi_d}$ on the boundary of the kink (See Fig.~\ref{fig:spectra2} [c]). 
If both the adiabatic parameters $t_1$ and $t_2$ are replaced by the texture variables, the semiclassical state $\Ket{\chi_{d+2}(\bm{n}(x_{d+1},x_{d+2}))}$, $\bm{n} \in S^2$, represents a skyrmion induced Inv$^d$ state $\ket{\chi_d}$. 
In a similar way to (\ref{eq:1d_texture}), for a bipartite 2-dimensional lattice, a simple realization is given by putting an Inv$^d$ state $\ket{\chi_d}$ at the center and wrapping the center with trivial pairs $\ket{\chi_d} \otimes \ket{\bar \chi_d}$ (See Fig.~\ref{fig:spectra2} [d]).
The Inv$^d$ state $\ket{\chi_d}$ trapped at a single skyrmion-like texture is stable under a perturbation leaving the system gapped.

%

%
%
%

\section{Crystalline SPT phases and Generalized homology}
\label{sec:genehomo}
Let $\{F_d\}_{d \in \Z}$ be the $\Omega$-spectrum so that $F_d$ is the space of Inv$^d$ states with a given onsite symmetry.
Let $X$ be a real space manifold. 
$X$ is typically the infinite Euclidean space $X = \R^d$ in the context of SPT phases, but $X$ can be an arbitrary real space manifold. 
Let $G$ be a symmetry group that acts on the real space $X$ and also the Inv$^{d}$ states, that is, the spectrum $F_d$. 
For a pair $(X,Y)$ of $G$-equivariant real spaces $X$ and $Y$ with $Y \subset X$, the $G$-equivariant generalized homology theory $h^G_n(X,Y)$ is mathematically defined by \cite{Whi,KonoTamaki}
\begin{align}	
h^G_{n}(X,Y) := \underset{k \to \infty}{\rm colim} \left[ S^{n+k}, (X/Y) \wedge  F_k \right]_G, 
\end{align}
where $G$ trivially acts on the sphere $S^{n+k}$. 
When $Y = \emptyset$, we simply write $h^G_n(X) = h^G_n(X,\emptyset)$. 

The group $h^G_n(X, Y)$ enjoys the axioms of the equivariant generalized homology theory. They are dual to those of equivariant generalized cohomology theory described for example in \cite{Bre}, and are shown by generalizing~\cite{Whi}. 
\begin{itemize}
\item
(homotopy)
If $G$-equivariant maps $f, f' : X \to X'$ are $G$-equivariantly homotopic by a homotopy which carries $Y \subset X$ into $Y' \subset X'$ all the way, then the induced maps $f_*, f'_* : h^G_n(X, Y) \to h^G_n(X', Y')$ are the same: $f_* = f'_*$.
\item
(excision)
For $A, B \subset X$, the inclusion $A \to A \cup B$ induces an isomorphism $h^G_n(A, A \cap B) \to h^G_n(A \cup B, B)$.
\item
(exactness)
For $Y \subset X$, there is a long exact sequence $\cdots \to h^G_n(Y) \to h^G_n(X) \to h^G_n(X, Y) \to h^G_{n-1}(Y) \to \cdots$.
\item
(additivity)
For $Y_\lambda \subset X_\lambda$ parameterized by a set $\Lambda = \{ \lambda \}$, the inclusions $X_\lambda \to \sqcup_\lambda X_\lambda$ induces an isomorphism $\prod_\lambda h^G_n(X_\lambda, Y_\lambda) \to h^G_n(\sqcup_\lambda X_\lambda, \sqcup_\lambda Y_\lambda)$.
\end{itemize}

In the above, the spaces and their subspaces are assumed to be $G$-CW complexes and their subcomplexes (see \cite{May} for the definition).

By design, the counterpart of the dimension axiom is $h^G_n(pt) = h_G^{-n}(pt) = \pi_0(F^G_{-n})$, where $F^G_{-n} \subset F_{-n}$ consists of $G$-fixed points. 
Here we assumed that $F_d$ is an equivariant $\Omega$-spectrum, which would be the general assumption in the context.
From the axiom, we can derive various exact sequences. For instance, we have the Mayer-Vietoris exact sequence $\cdots \to h^G_{n+1}(X) \to h^G_n(A \cap B) \to h^G_n(A \sqcup B) \to h^G_n(X) \to \cdots$ for a cover $X = A \cup B$. We also have $h^G_n(X \times D^d, X \times \partial D^d) \cong h^G_{n - d}(X)$, where $D^d$ is the $d$-dimensional disk endowed with trivial $G$-action.


In the sequel, we will use the following property of the generalized homology, which is called ``rolling and unrolling" in \cite{TE18}:
\begin{itemize}
\item
For any normal subgroup $H \subset G$ which acts on $X$ freely and the $\Omega$-spectrum trivially, we have $h^G_n(X) \cong h^{G/H}_n(X/H)$.
\end{itemize}
Note that this is independent of the axioms of generalized homology theory.

Let us see the physical meaning of the integer grading $n \in \Z$. 
When $X$ is the $d$-dimensional disc $D^d$, $Y$ is its boundary $\p D^d$, and the crystalline symmetry group $G$ is trivial $G = \{e\}$, the Poincar\'e-Lefschetz duality implies that $h_n(D^d, \p D^d) \cong h^{d-n}(D^d) \cong h^{d-n}(pt) = [pt,F_{d-n}]$. 
That is, the generalized homology $h_n(D^d, \p D^d)$ gives the classification of SPT$^{d-n}$ phases. 
For this reason, we call the integer degree $n \in \Z$ {\it the degree of SPT phenomena}. 
Similarly, with crystalline symmetry $G$, the generalized homology $h^G_n(X,Y)$ is understood as the classification of degree $n$ SPT phenomena over $X$ up to degree $(n-1)$ SPT phenomena over $Y$. 
For examples, 
\begin{itemize}
\item
$h^G_0(X,Y)$ is the classification of SPT phases over $X$ which can have anomalies over $Y$. 
\item
$h^G_1(X,Y)$ is the classification of adiabatic cycles over $X$ which can change SPT phases nonadiabatically over $Y$. 
\item
$h^G_{-1}(X,Y)$ is the classification of anomalies over $X$ which can have ``sources and sinks of anomalies" over $Y$. 
\end{itemize}

In the rest of this section, we interpret the exactness axiom and Mayer-Vietoris sequences from the viewpoint of SPT physics. 

\subsection{Exactness}
Let $X$ and $Y$, $Y \subset X$, be a pair of $G$-symmetric real spaces. 
Associated with the inclusions $f: Y \to X$ and $g: (X, \emptyset) \to (X,Y)$, 
we have the long exact sequence 
\begin{align}
\cdots 
\xrightarrow{\p^{n+1}}
h^{G|_{Y}}_{n}(Y) 
\xrightarrow{f_*^n}
h^{G}_n(X) 
\xrightarrow{g_*^n}
h^G_n(X,Y) 
\xrightarrow{\p^n}
h^{G}_{n-1}(Y) 
\xrightarrow{f_*^{n-1}}
\cdots
\end{align}
The existence of this sequence is an axiom of the generalized homology. 
Let us illustrate the definitions of homomorphisms $f_*^n, g_*^n$ and $\p^n$ for SPT phases, i.e., \ $n=0$.  

The homomorphism $f_*^0$ is defined by embedding an SPT phase over $Y$ into $X$: 
$$
\begin{array}{c}
\includegraphics[width=0.8\linewidth, trim=0cm 10.5cm 0cm 0cm]{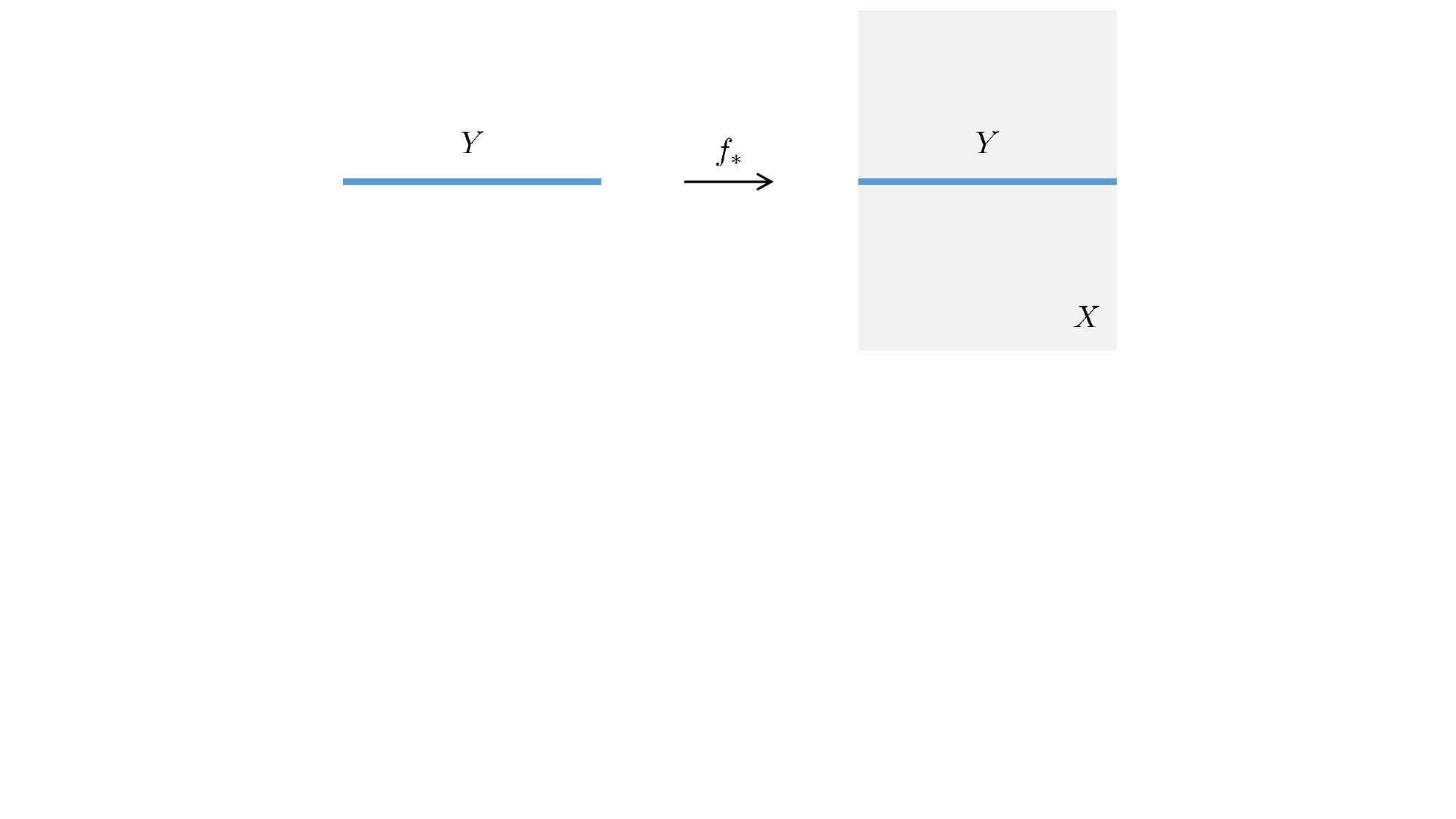}
\end{array}
$$
The homomorphism $f_*^{-1}$ is defined similarly    . 

The homomorphism $g_*^0$ is defined as follows. 
For a given SPT phase $x$ over $X$, cutting out $Y$ from $X$ leads to anomalies localized on $Y$. 
Then, the SPT phase $x$ over $X \backslash Y$ which can have anomalies on $Y$ defines an element of $g^0_*(x) \in h^G_0(X,Y)$:
$$
\begin{array}{c}
\includegraphics[width=0.8\linewidth, trim=0cm 10.5cm 0cm 0cm]{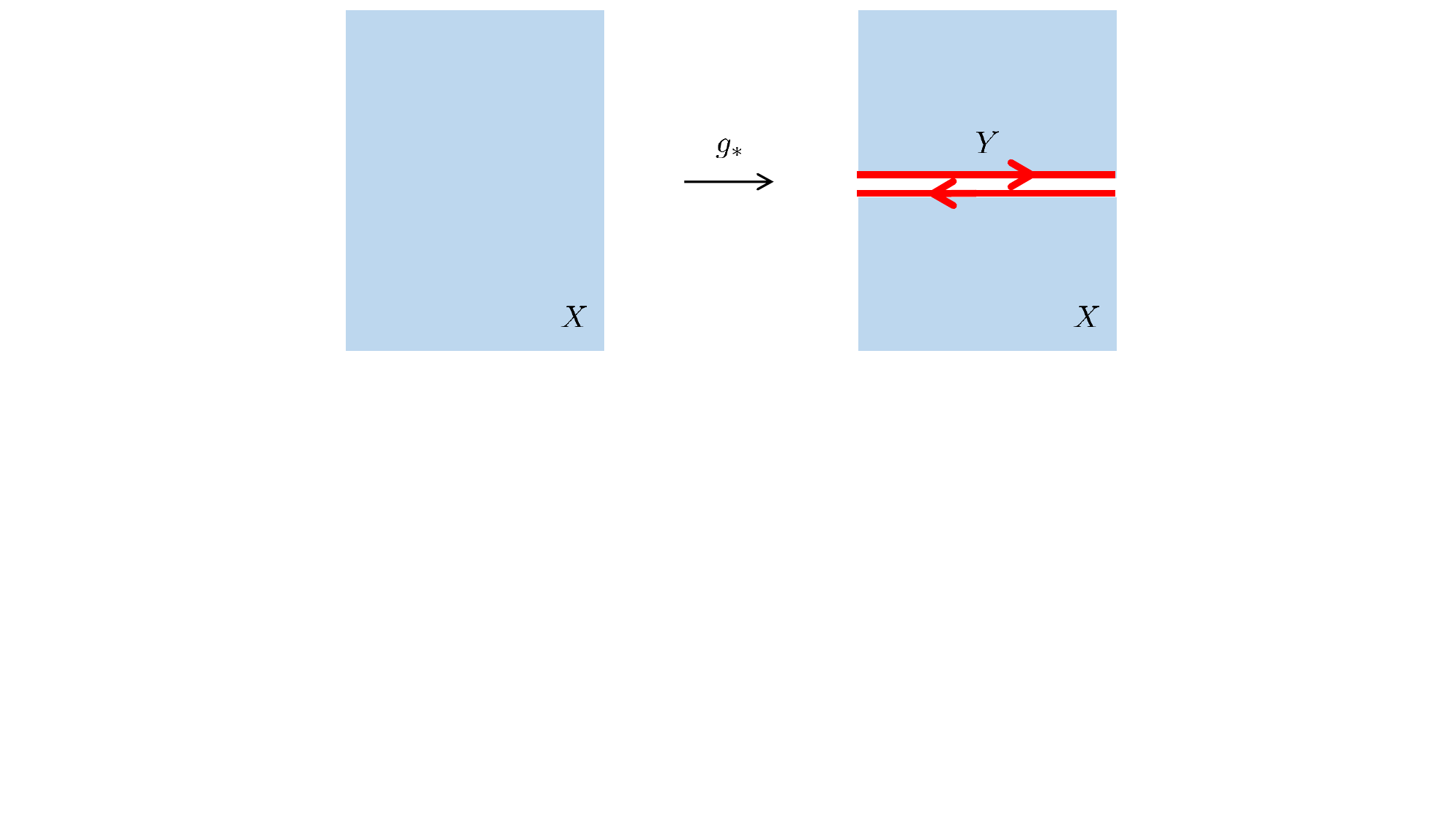}
\end{array}
$$

For an SPT phase $x \in h^G_0(X,Y)$ on $X$ which can be anomalous on $Y$, the boundary map $\p^0(x) \in h^G_{-1}(Y)$ is defined as the anomaly of $x$ on $Y$:
$$
\begin{array}{c}
\includegraphics[width=0.8\linewidth, trim=0cm 10.5cm 0cm 0cm]{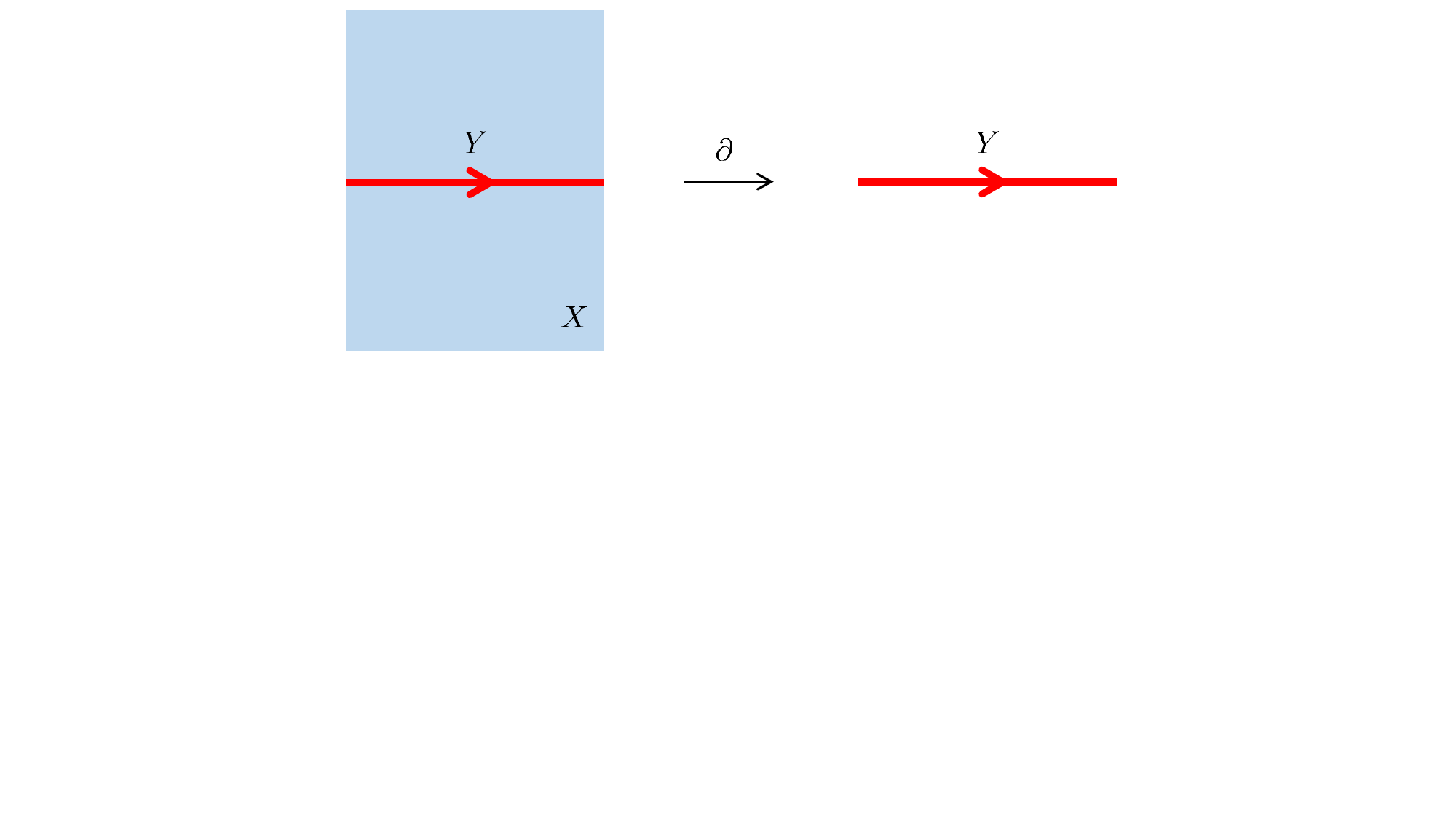}
\end{array}
$$

For the exactness of the sequence, $\im \subset \ker$ for $n=0$ is readily confirmed. 
An SPT phase on $Y$ has no anomaly on $Y$, meaning $\im f^0_* \subset \ker g_*^0$. 
In the same way, an SPT phase on $X$ has no anomaly, leading to $\im g_*^0 \subset \ker \p^0$. 
Since an anomaly $\p^0(x) \in h^G_{-1}(Y)$ given by the boundary map $\p^0$ from $(X,Y)$ should be represented as the boundary of an SPT phase over $X \backslash Y$, $\p^0(x)$ is nonanomalous as an anomaly on $X$. 
(Note that, by definition, an SPT phase $x \in h^G_0(M)$ on $M$ has no anomaly even if $M$ has a boundary $\p M \neq \emptyset$.)

\subsubsection{Boundary map $\p^0$: the bulk-boundary correspondence}
When $X$ has a boundary $\p X$, the boundary map $\p^0: h^G_0(X,\p X) \to h^G_{-1}(\p X)$ is what we call the bulk-boundary correspondence.
With crystalline symmetry, not every SPT phase on $X$ implies a boundary anomaly on $\p X$. 
Conversely, not every anomaly on $\p X$ implies a bulk SPT phase on $X$. 
The precise relationship among bulk SPT phases and boundary anomalies is that for an SPT phase $x \in h^G_0(X,\p X)$ in bulk, the boundary anomaly of the SPT phase $x$ is given by the element $x \mapsto \p^0(x) \in h^G_{-1}(\p X)$.

\subsection{Mayer-Vietoris sequence}
Let us introduce a $G$-symmetric cover $X = U \cup V$ for a real space $X$ where each $U$ and $V$ is $G$-symmetric.~\footnote{
We say that a space $X$ is $G$-symmetric if $g(x) \in X$ for any point $x \in X$ and any $g \in G$. 
} 
Associated with the sequence of inclusions 
\begin{align}
\begin{CD}
\begin{array}{c}
\includegraphics[width=0.1\linewidth, trim=0cm 0cm 0cm 0cm]{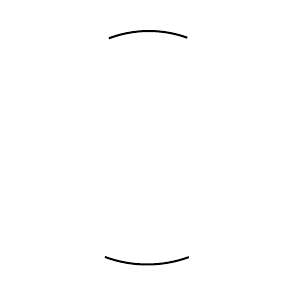}
\end{array}@>>>
\begin{array}{c}
\includegraphics[width=0.1\linewidth, trim=0cm 0cm 0cm 0cm]{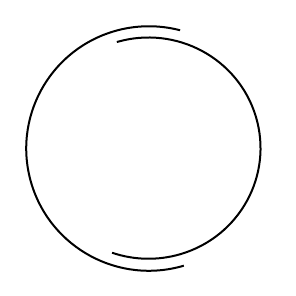}
\end{array} 
@>>>
\begin{array}{c}
\includegraphics[width=0.1\linewidth, trim=0cm 0cm 0cm 0cm]{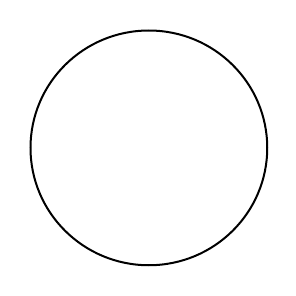}
\end{array}  \\
U \cap V
@>i>>
U \sqcup V 
@>j>>
X = U \cup V 
\end{CD}
\end{align}
we have the long exact sequence called the Mayer-Vietoris sequence of generalized homology~\cite{Hat}
\begin{align}
\cdots 
\xrightarrow{\p^{n+1}}
h^{G}_{n}(U \cap V) 
\xrightarrow{i_*^n}
h^{G}_n(U) \oplus h^{G}_n(V) 
\xrightarrow{j_*^n}
h^G_n(X) 
\xrightarrow{\p^n}
h^{G}_{n-1}(U \cap V) 
\xrightarrow{i_*^{n-1}}
\cdots. 
\label{eq:mv}
\end{align}
We illustrate the homomorphisms $i_*^n,j_*^n$ and $\p^n$ in (\ref{eq:mv}) for SPT phases ($n=0$, says). 
Other homomorphisms with $n \neq 0$ are understood similarly. 

The homomorphism $i_*^0$ is defined as follows. 
First, for a given SPT phase $x \in h^{G}_{0}(U \cap V)$ we adiabatically create a pair of SPT phases $x$ and its conjugate $\bar x = -x$ inside $U \cap V$. 
Next, move the SPT phase $x$ to the interior of $U$ and $(-x)$ to the interior of $V$, which defines the homomorphism $i_*^0: x \mapsto (x|_U,-x|_V) \in h^{G}_0(U) \oplus h^{G}_0(V)$, where $x|_U$ and $x|_V$ represent the SPT phases of $x$ as those in the real spaces $U$ and $V$, respectively. 

Similarly, the homomorphism $i_*^{-1}$ is defined by making a pair of anomalies $x$ and $\bar x = -x$ inside $U \cap V$ belonging to $h^{G}_{-1}(U \cap V)$ and moving those anomalies to $U$ and $V$. 

The homomorphism $j_*^0$ is defined by embedding the SPT phases $x$ over $U$ and $y$ over $V$ into the total real space $X$. 
This defines the homomorphism $j_*^0:(x,y) \mapsto x|_X + y|_X \in h^G_0(X)$. 

The boundary homomorphism $\p^0$ is defined by the anomaly of SPT phases. 
For an SPT phase $x$ over $X$, cutting the real space $X$ to the two pieces $U$ and $V$ yields the anomaly localized on the boundaries $\p U$ that belongs to the generalized homology $h^{G}_{-1}(U \cap V)$. 
This defines the boundary homomorphism $\p^0: x \mapsto x|_{\p U} \in h^{G}_{-1}(U \cap V)$. 

Some parts of the exactness of the Mayer-Vietoris sequences are readily confirmed. 
It holds that $\im i_*^0 \subset \ker j_*^0$, since an adiabatically created pair of SPT phases $x$ and $\bar x = -x$ over $U \cap V$ is a trivial SPT phase over $X$. 
$\im j_*^0 \subset \ker \p^0$ holds true since an SPT phase over $U$ has no anomaly unless cutting $U$, and so is an SPT phase over $V$. 
Similarly, we find that $\im \p^0 \subset \ker i_*^{-1}$ holds since the anomaly over $U \cap V$ created by the boundary homomorphism $\p^0$ becomes non-anomalous over $U$ and $V$. 
We should note that to prove the exactness, we should further show the inverse inclusions.

\section{Atiyah-Hirzebruch spectral sequence}
\label{sec:ahss}
A useful tool to compute a generalized (co)homology $h^G_n(X,Y)$ is the AHSS. 
For SPT phases, $X$ is a real space manifold with a finite dimension. 
In such cases the $E^r$-page of the spectral sequences converges at $E^{\infty} = E^{d+1}$, where $d$ is the dimension of $X$. 
In this section, we illustrate the AHSS for crystalline SPT phenomena. 
We especially describe how to compute the differentials $d^r$ of the AHSS from the viewpoint of physics and how to solve the short exact sequences by $E^{\infty}$-page to get the generalized homology. 
It turns out that the AHSS is the mathematical structure behind prior works \cite{Hermele_torsor, PWJZ17, HSHH17, TB18}.

\subsection{Cell decomposition and $E^1$-page}
\label{sec:e1}
Let $X$ be a $d$-dimensional space manifold and $G$ be a crystalline symmetry group or point group acting on the space $X$. 
The first step of the AHSS is to introduce a $G$-symmetric filtration of $X$, 
\begin{align}
X_0 \subset X_1 \subset \cdots \subset X_d = X, 
\end{align}
where $p$-dimensional subspace $X_p$ is called $p$-skeleton. 
If we have a $G$-symmetric cell decomposition of $X$, associated with its cell decomposition, the $p$-skeleton $X_p$ is given inductively by 
\begin{align}
X_0 = \{0\mbox{-cells}\}, \qquad 
X_p = X_{p-1} \cup \{p\mbox{-cells}\}. 
\end{align}
Here, $p$-dimensional open cells composing the cell decomposition are called $p$-cells. 
It should be remarked in the construction of AHSS that the spectral sequence itself is constructed from the filtration of a space that may not be associated with a cell decomposition, and a cell structure is only necessary for a useful expression of the spectral sequence.

In the rest of this section, we assume a $G$-symmetric cell decomposition so that the little group $G_{D^p_j} \subset G$ for a $p$-cell $D^p_j$ does not change the real space position in $D^p_j$, i.e., the group $G_{D^p_j}$ behaves as an onsite symmetry inside the $p$-cell $D^p_j$. 

The $E^1$-page of the AHSS is defined as
\begin{align}
E^1_{p,-q} 
= h_{p-q}^G(X_p \cup Y, X_{p-1} \cup Y).
\end{align}
By using the cell decomposition of $X$ and $Y$, we can express $E^1_{p,-q}$ as
\begin{equation}\begin{split}
E^1_{p,-q} 
= \prod_{j\in I_p} h^{G_{D^p_j}}_{p-q}(D^p_j,\p D^p_j)
= \prod_{j \in I_p} h^{q}_{G_{D^p_j}}(pt)
= ``\mbox{SPT$^{q}$ phases on $p$-cells}", 
\end{split}\end{equation}
where $j$ runs the set (denoted by $I_p$) of inequivalent $p$-cells of $X$ that are not in $Y$, and $G_{D^p_j}$ is the little group that fixing the $p$-cell $D^p_j$. 
Note that for each $p$-cell $D^p_j$, the group $h^q_{G_{D^p_j}}(pt)$ is dictated solely by the subgroup $G_{D^p_j}$.
The $E^1$-page expresses the collection of ``local data of SPT phases", which we write in the following table 
$$
\begin{array}{c|cccccccc}
\vdots & \vdots & \vdots & \vdots & \vdots & \\
q=0 & {\rm SPT}^{0} & {\rm SPT}^{0} & {\rm SPT}^{0} & {\rm SPT}^{0} & \cdots \\
q=1 & {\rm SPT}^{1} & {\rm SPT}^{1} & {\rm SPT}^{1} & {\rm SPT}^{1} & \cdots \\
q=2 & {\rm SPT}^{2} & {\rm SPT}^{2} & {\rm SPT}^{2} & {\rm SPT}^{2} & \cdots \\
q=3 & {\rm SPT}^{3} & {\rm SPT}^{3} & {\rm SPT}^{3} & {\rm SPT}^{3} & \cdots \\
\vdots & \vdots & \vdots & \vdots & \vdots &\\
\hline 
E^1_{p,-q} & p=0 & p=1 & p=2 & p=3 & \cdots \\
\end{array}
$$

\subsection{First differential and $E^2$-page}
\label{sec:d1}
Because the topological information contained in the $E^1$-page is spatially local, these data should be glued in between different cells properly. 
To do so, we define the first differential 
\begin{align}
d^1_{p,-q}: E^1_{p,-q} \to E^1_{p-1,-q}, 
\end{align}
that is, 
\begin{align}
d^1_{p,-q}: \mbox{``SPT$^q$ phases on $p$-cells" $\to$ ``SPT$^q$ phases on $(p-1)$-cells"}, 
\end{align}
as trivializing SPT$^{q}$ phases on $(p-1)$-cells by pair-creation of SPT$^{q}$ phases on adjacent $p$-cells~\cite{Hermele_torsor}. 
Let's see some first differentials in low-space dimensions. 
The physical meaning of $d^1_{p,-q}$ depends on the degree of SPT phenomena $n=p-q$. 

\begin{itemize}
\item
{\it $d^1_{1,0}$: ``SPT$^0$ phases on 1-cells" $\to$ ``SPT$^0$ phases on 0-cells"}

An SPT$^0$ phase $\ket{\chi}$ in a 1-cell $D^1_j$ is identified with the adiabatic pump labeled by the SPT$^0$ state $\ket{\chi}$. 
SPT$^0$ phases on $0$-cells adjacent to the $1$-cell may be trivialized by the pair-creation of the SPT$^0$ state $\ket{\chi}$ and its conjugate $\ket{\bar \chi}$: 
$$
\begin{array}{c}
\includegraphics[width=0.8\linewidth, trim=0cm 15cm 0cm 0cm]{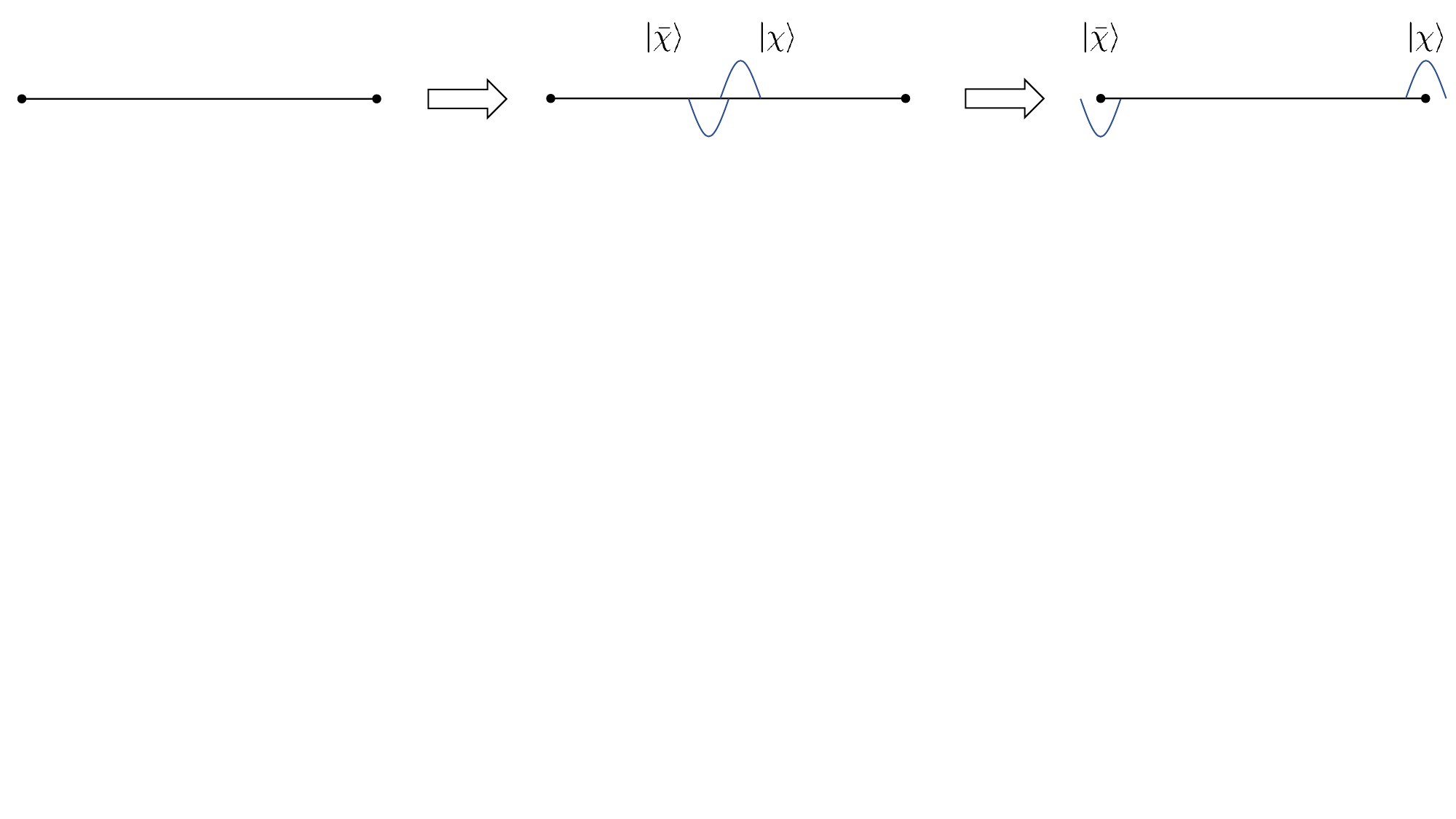}
\end{array}
$$
The point is that this trivialization is doable {\it within a unique gapped ground state}. 
This yields the equivalence relation $E^1_{0,0}/\im (d^1_{1,0})$. 

Note that $d^1_{1,0}$ is also identified with the consistency condition $\ker (d^1_{1,0}) \subset E^1_{1,0}$ for adiabatic cycles inside $1$-cells. 
This is because if there remains an edge SPT$^0$ state at 0-cells per an adiabatic cycle, 
such an adiabatic cycle produces an SPT$^0$ state at the $0$-cell, meaning the resulting invertible state adiabatically differs from the initial state. 

\item
{\it $d^1_{2,-1}$: ``SPT$^1$ phases on 2-cells" $\to$ ``SPT$^1$ phases on 1-cells"}

The homomorphism $d^1_{2,-1}$ represents how SPT$^1$ phases on $2$-cells trivialize SPT$^1$ phases on adjacent $2$-cells. 
This trivialization exists since inside a $2$-cell one can make an SPT$^1$ phase localized on a circle within a unique gapped ground state, which results in the equivalence relation $E^1_{1,-1}/\im (d^1_{2,-1})$: 
$$
\begin{array}{c}
\includegraphics[width=0.8\linewidth, trim=0cm 14cm 0cm 0cm]{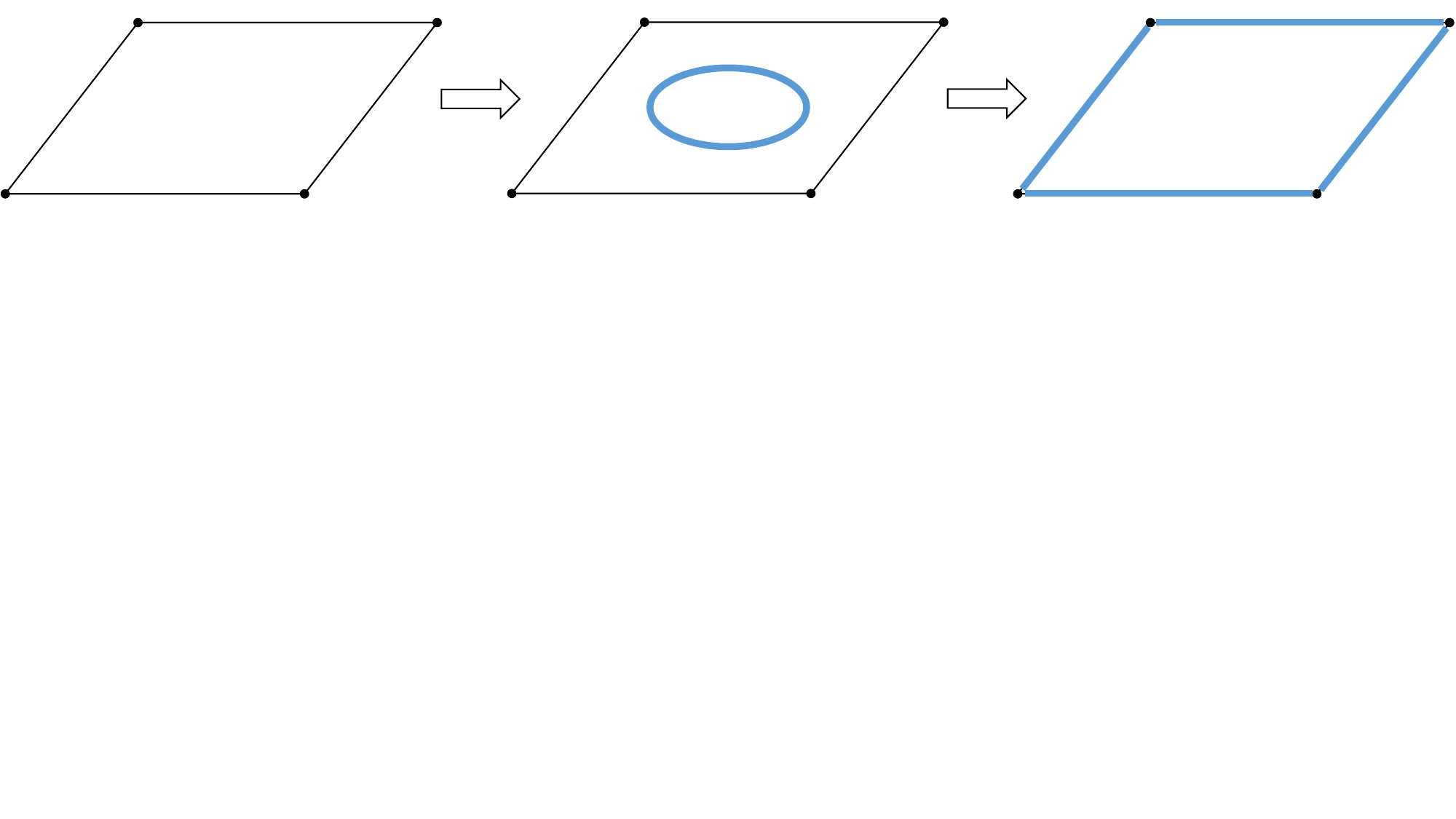}
\end{array}
$$
At the same time, $d^1_{2,-1}$ yields the consistency condition, $\ker (d^1_{2,-1}) \subset E^1_{2,-1}$, to have a nontrivial adiabatic cycle to produce an SPT$^1$ phase in $2$-cells, otherwise there remains a nontrivial SPT$^1$ phase somewhere after a period of adiabatic cycle. 

There is a strong constraint on SPT$^1$ phases created by an adiabatic cycle of $d^1_{2,-1}$. 
A pumped SPT$^1$ phase over a circle should have no flux inside the circle, i.e., the pumped SPT$^1$ phase is a disc state created from the vacuum through the imaginary time path-integral: 
$$
\begin{array}{c}
\includegraphics[width=0.8\linewidth, trim=0cm 14cm 0cm 0cm]{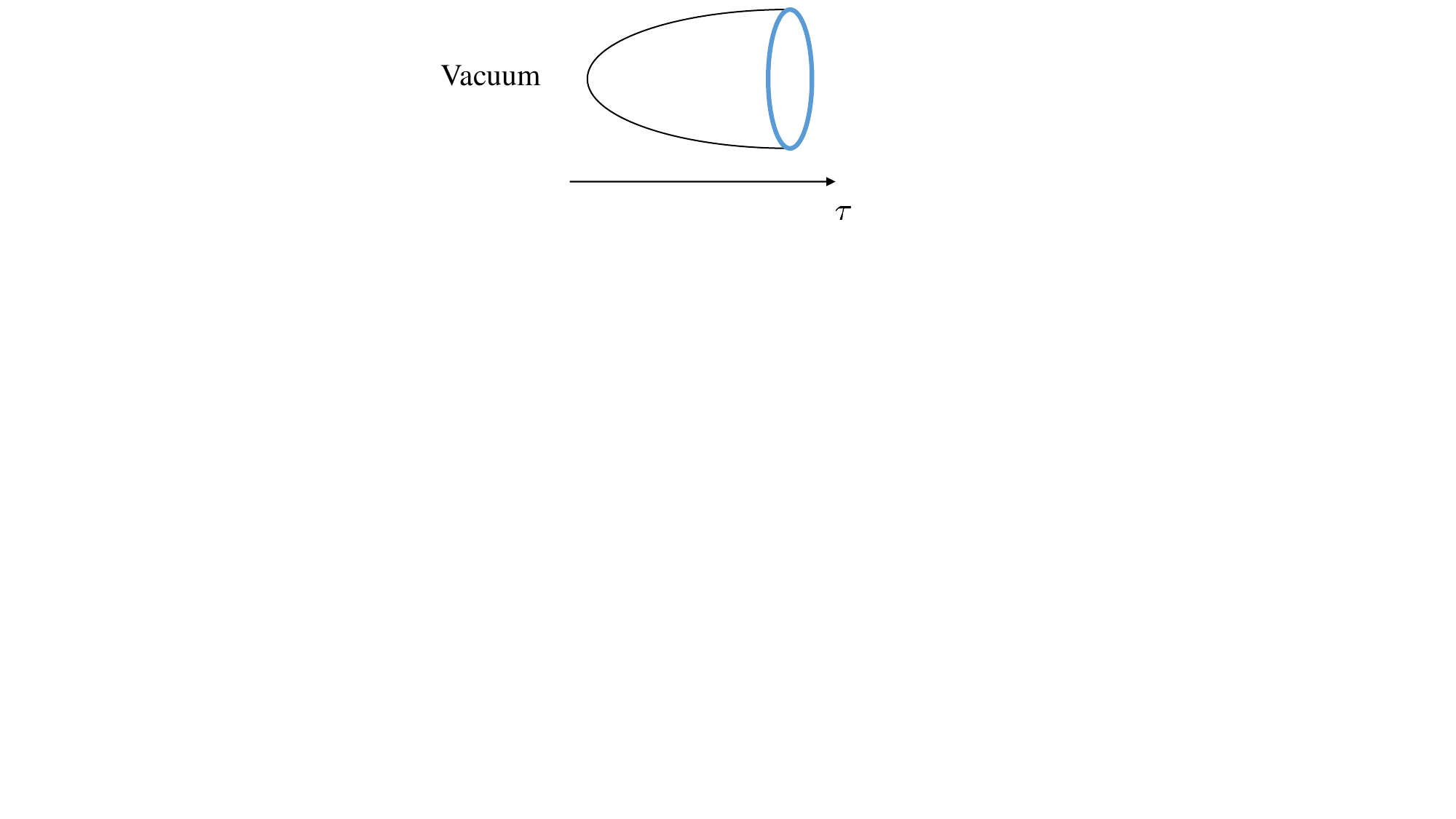}
\end{array}
$$
Put differently, the twisted boundary condition of the pumped SPT$^1$ state is trivial. 

\item
{\it $d^1_{1,-1}$: ``SPT$^1$ phases on 1-cells" $\to$ ``SPT$^1$ phases on 0-cells"}

An SPT$^1$ phase over a 0-cell represents an anomalous $0d$ edge state realized on the $0$-cell. 
Not every Hilbert space belonging to edges of SPT$^1$ phases over $0$-cells contained in $E^1_{0,-1}$ is truly anomalous over the whole space $X$, since some of which become anomaly free. 
If a $1$-cell admits nontrivial SPT$^1$ state, there exists the nonanomalous process to produce left and right anomalous edge states $\gamma_L$ and $\gamma_R$, and they can trivialize the anomalous edges states at the adjacent $0$-cell: 
\begin{align}
\begin{array}{c}
\includegraphics[width=0.8\linewidth, trim=0cm 15cm 0cm 0cm]{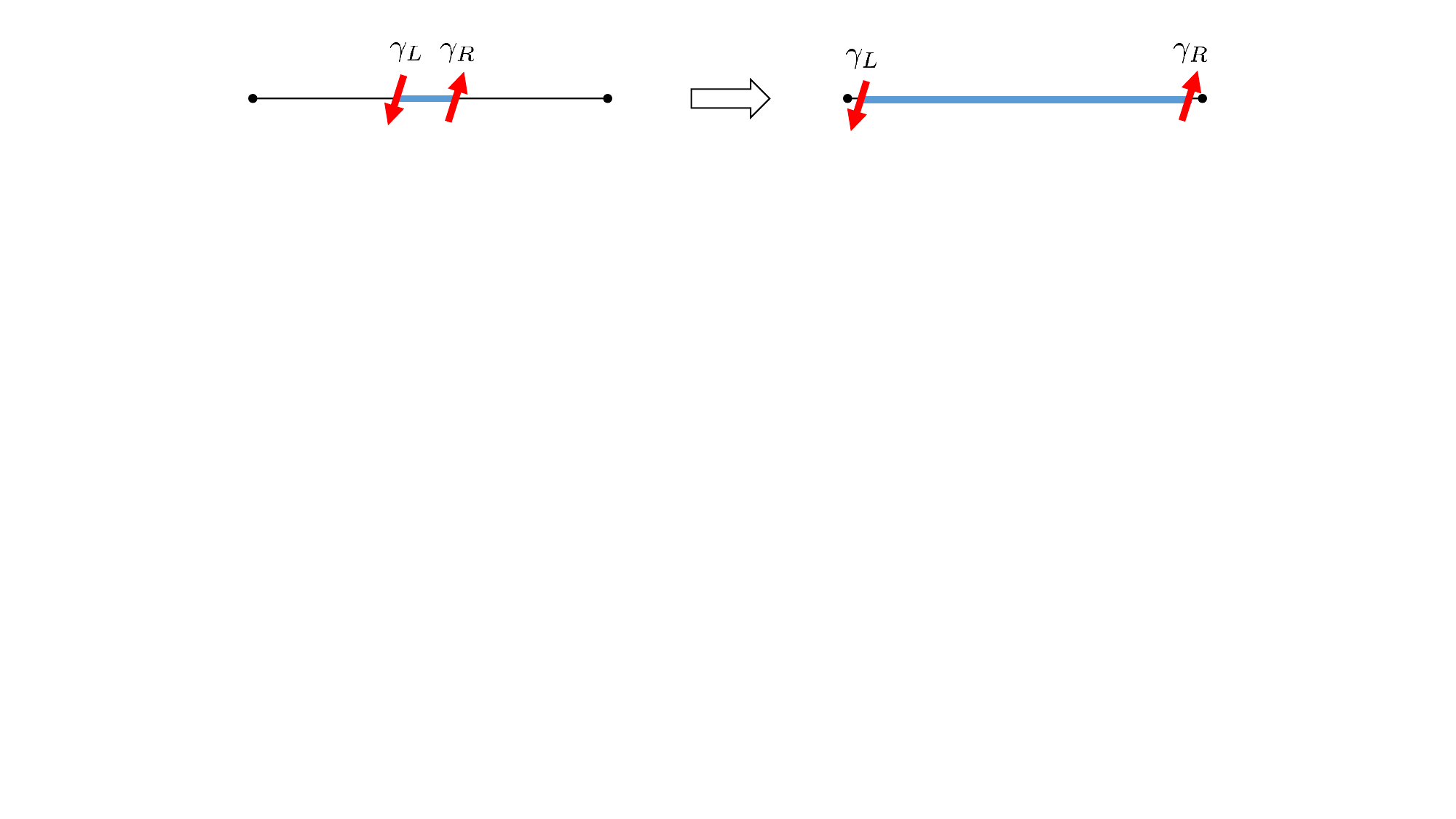}
\end{array}
\label{fig:d111}
\end{align}
This defines the trivialization $E^1_{0,-1} / \im (d^1_{1,-1})$. 
It should be noted that this process is anomaly free, i.e., it does not change the anomaly of the Hilbert space.

Moreover, the homomorphism $d^1_{1,-1}$ is also considered the consistency condition to put SPT$^1$ phases on $1$-cells. 
An anomaly-free combination of SPT$^1$ phases over $1$-cells should be in the subgroup $\ker (d^1_{1,-1}) \subset E^1_{1,-1}$. 
Otherwise, an anomalous edge state remains somewhere, which means the resulting Hilbert space can not carry a unique gapped ground state.

\item
{\it $d^1_{2,-2}$: ``SPT$^2$ phases on 2-cells" $\to$ ``SPT$^2$ phases on 1-cells"}

The abelian group $E^1_{1,-2}$ represents anomalous $1d$ edge states on $1$-cells. 
The homomorphism $d^1_{2,-2}$ indicates how these anomalous edge states are canceled out with the edge anomaly of SPT$^2$ phases adiabatically created in 2-cells: 
$$
\begin{array}{c}
\includegraphics[width=0.8\linewidth, trim=0cm 14cm 0cm 0cm]{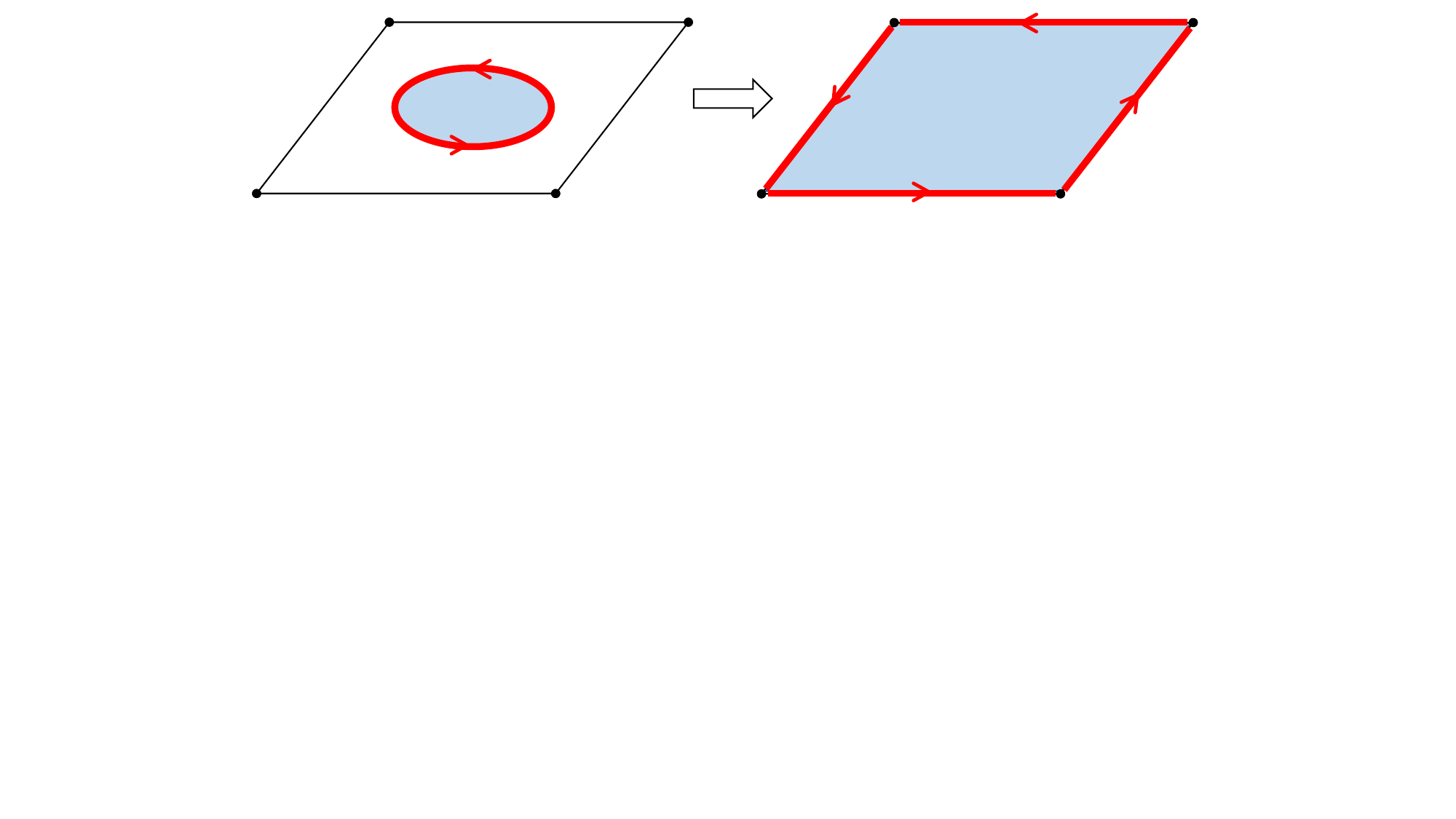}
\end{array}
$$
This yields the equivalence relation $E^1_{1,-2}/\im (d^1_{2,-2})$. 
The point is that this process is anomaly free: 
The resulting state after trivialization acquires the tensor product of a nontrivial SPT phase, but it does not matter whether the Hilbert space in question is anomalous or not.  

The homomorphism $d^1_{2,-2}$ imposes the consistency condition on the set $E^1_{2,-2}$ of SPT$^2$ phases locally defined on 2-cells.
An element in $E^1_{2,-2}$ must be trivial under the homomorphism $d^1_{2,-2}$, otherwise there remains an anomalous edge state somewhere. 
This gives the consistency condition $\ker (d^1_{2,-2}) \subset E^1_{2,-2}$. 

It is important that an SPT$^2$ phase over a disc inside a $2$-cell created by $d^1_{2,-2}$ should have no flux inside the disc. 
Otherwise, an SPT$^2$ phase with a flux may have an anomalous edge state localized at the flux, which breaks the anomaly-free condition of $d^1_{2,-2}$. 

\item
{\it $d^1_{1,-2}$: ``SPT$^2$ phases on 1-cells" $\to$ ``SPT$^2$ phases on 0-cells"}

The abelian group $E^1_{0,-2} $ represents SPT$^2$ phases at 0-cells and can be understood as ``sources and sinks" of anomalous $1d$ edge states (the anomaly of anomalies, says). 
The homomorphism $d^1_{1,-2}$ can be viewed as how the source of anomalous $1d$ edge states at $0$-cells is trivialized by anomalous edge states on adjacent $1$-cells, which is expressed as $E^1_{0,-1} / \im (d^1_{1,-2})$. 
Also, the homomorphism $d^1_{1,-2}$ represents the consistency condition for anomalous $1d$ edge states on $1$-cells to extend those states to adjacent $0$-cells, which is expressed as $\ker (d^1_{1,-2}) \subset E^1_{1,-1}$. 
See the following figure.
$$
\begin{array}{c}
\includegraphics[width=0.8\linewidth, trim=0cm 13cm 0cm 0cm]{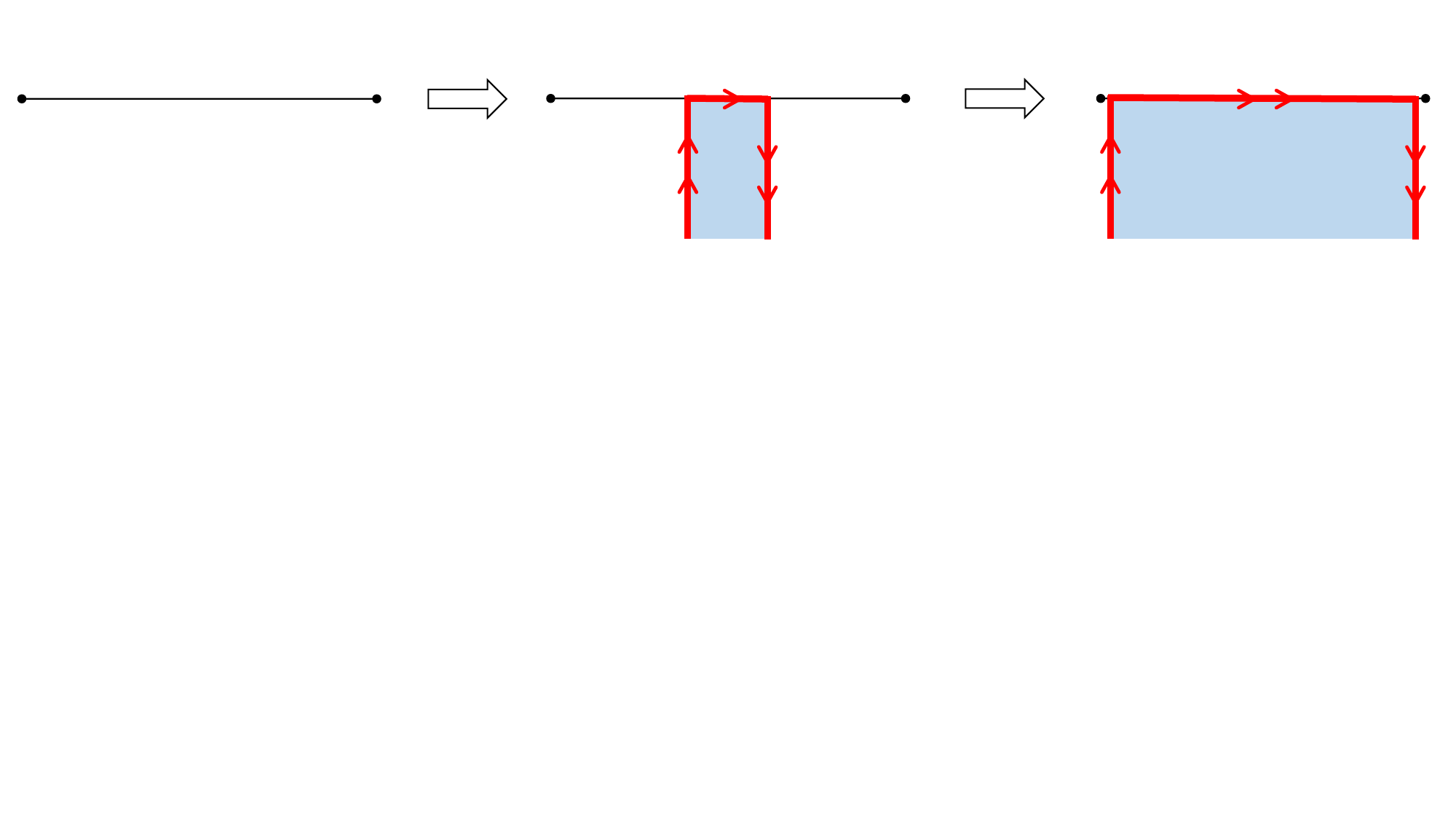}
\end{array}
$$

\end{itemize}

Similarly, we have the physical pictures for all the first differentials $d^1_{p,-q}$. 

The first differential obeys that the boundary of a boundary is trivial 
\begin{align}
d^1_{p-1,-q} \circ d^1_{p,-q} = 0. 
\end{align}
This holds true: an SPT$^{q}$ phase on a $(p-1)$-cell made from an adjacent $p$-cell via the adiabatic cycle is anomaly free. 
Taking the homology of $d^1$ we have the $E^2$-page 
\begin{align}
E^2_{p,-q}:= \ker (d^1_{p,-q}) / \im (d^1_{p-1,-q}). 
\end{align}
The term $E^2_{p,-q}$ has the following physical meaning: 
SPT$^q$ phases in $p$-cells which cannot be trivialized by adjacent $(p+1)$-cells adiabatically and can extend to adjacent $(p-1)$-cells without showing an anomalous state.


\subsection{Higher differentials}
This is not the end of the story. 
In general, there exist $r$-th order trivialization and obstruction, the trivialization of SPT$^{q}$ phases in $p$-cells by $(p-r)$-cells and the obstruction to extend SPT$^{q}$ phases in $p$-cells to $(p+r)$-cells.
In the AHSS, this is expressed by the $r$-th differential 
\begin{align}
d^r_{p,-q}: E^r_{p,-q} \to E^r_{p-r,-q+r-1}, 
\end{align}
i.e.,\ 
$$
d^r_{p,-q}: \mbox{``SPT$^{q}$ phases on $p$-cells" $\to$ ``SPT$^{q-r+1}$ phases on $(p-r)$-cells"}
$$
The $r$-th differential also obeys that 
\begin{align}
d^r_{p-r,-q+r-1} \circ d^r_{p,-q} = 0, 
\end{align}
and its homology defines the $E^{r+1}$-page 
\begin{align}
E^{r+1}_{p,-q}:= \ker (d^r_{p,-q}) / \im (d^r_{p+r,-q-r+1}).
\end{align}
If $X$ is $d$-dimensional, the $E^r$-page converges at $E^{d+1}$-page
\begin{align}
E^1 \Rightarrow 
E^2 \Rightarrow 
\cdots 
\Rightarrow 
E^{d+1} = E^{d+2} = \cdots =: E^{\infty}. 
\end{align}
The converged page is called the limiting page and is denoted by $E^{\infty}$. 
The physical meaning of $E^{\infty}_{p,-q}$ is SPT$^q$ phases in $p$-cells which cannot be trivialized by any adjacent high-dimensional cells adiabatically and can extend to any adjacent low-dimensional cells without anomaly. 

In the context of SPT phases, the higher differentials of the AHSS may be computed in a physical picture. 
We describe the physical definitions of some higher differentials in order. 
\begin{itemize}
\item
{\it $d^2_{2,-1}$: ``SPT$^1$ phases on 2-cells" $\to$ ``SPT$^0$ phases on 0-cells"}

Since $E^2_{2,-1} \subset \ker d^1_{2,-1}$, an adiabatically created SPT$^1$ phase in 2-cells belonging to $E^2_{2,-1}$ can glue together on $1$-cells. 
However, it is nontrivial if its SPT$^1$ phase can collapse at $0$-cells. 
The SPT$^1$ phase of $E^2_{2,-1}$ may have an SPT$^0$ charge around $0$-cells, it behaves as the obstruction to collapse the SPT$^1$ state: 
$$
\begin{array}{c}
\includegraphics[width=0.8\linewidth, trim=0cm 12cm 0cm 0cm]{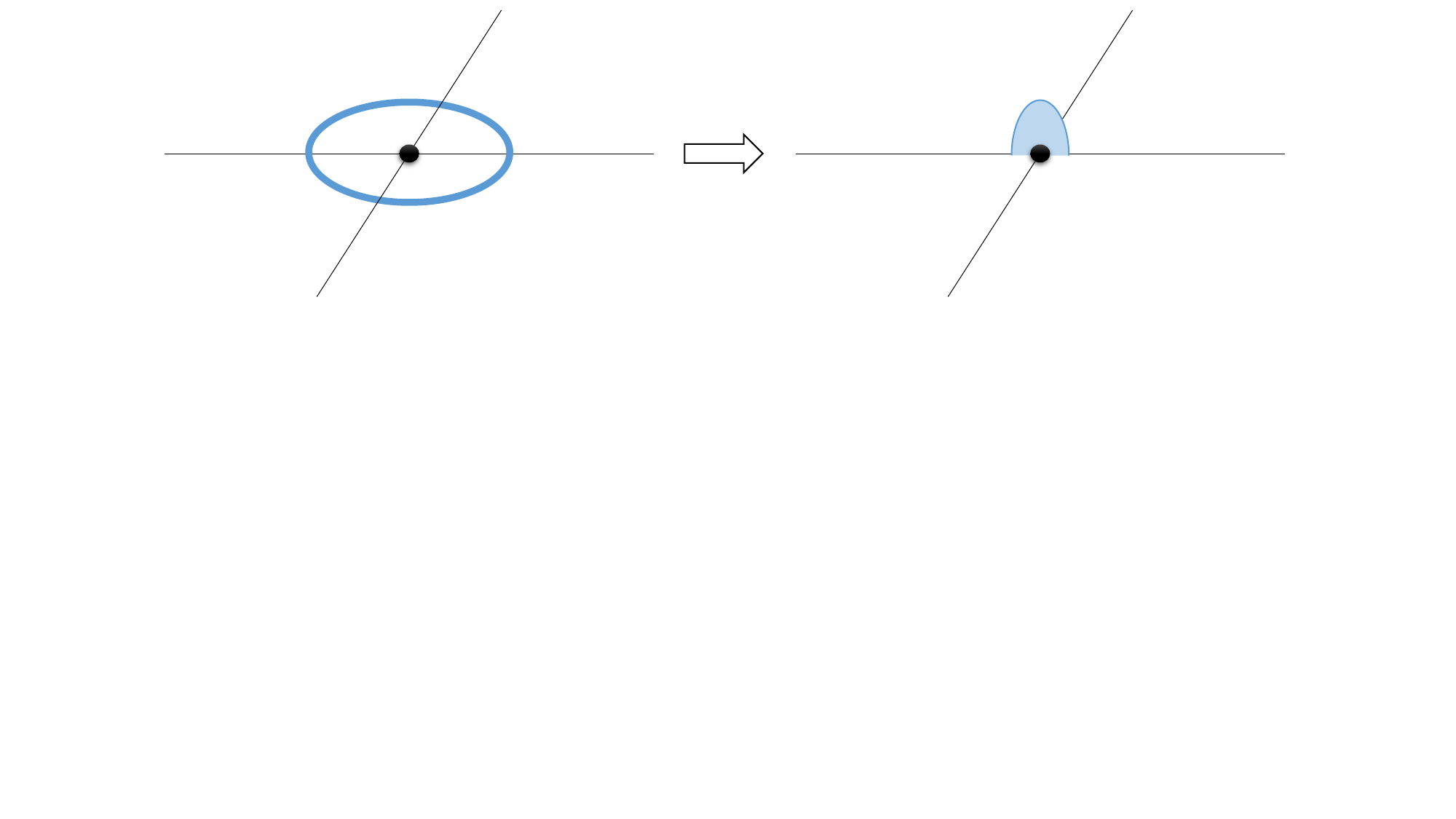}
\end{array}
$$
This defines the homomorphism $d^2_{2,-1}: E^2_{2,-1} \to E^2_{0,0}$. 

In fact, some sort of rotation symmetry around a 0-cell enforces nontrivial symmetry flux at the 0-cell. 
In the presence of such flux, an SPT$^1$ phase may possess an SPT$^0$ phase localized at the flux. 
This phenomenon is well-known in bosonic SPT phases and formulated by the slant product of the group cohomology~\cite{DW90,TantivasadakarnPRB17}. 

See Sec.~\ref{sec:2d_fermion_inversion_i^2=1} for a nontrivial example of $d^2_{2,-1}$. 

\item
{\it $d^2_{2,-2}$: ``SPT$^2$ phases on 2-cells" $\to$ ``SPT$^1$ phases on 0-cells"}

Similarly, the homomorphism $d^2_{2,-2}: E^2_{2,-2} \to E^2_{0,-1}$ measures the obstruction for the collapse of the SPT$^2$ phase in $2$-cells at $0$-cells: 
\begin{align}
\begin{array}{c}
\includegraphics[width=0.8\linewidth, trim=0cm 12cm 0cm 0cm]{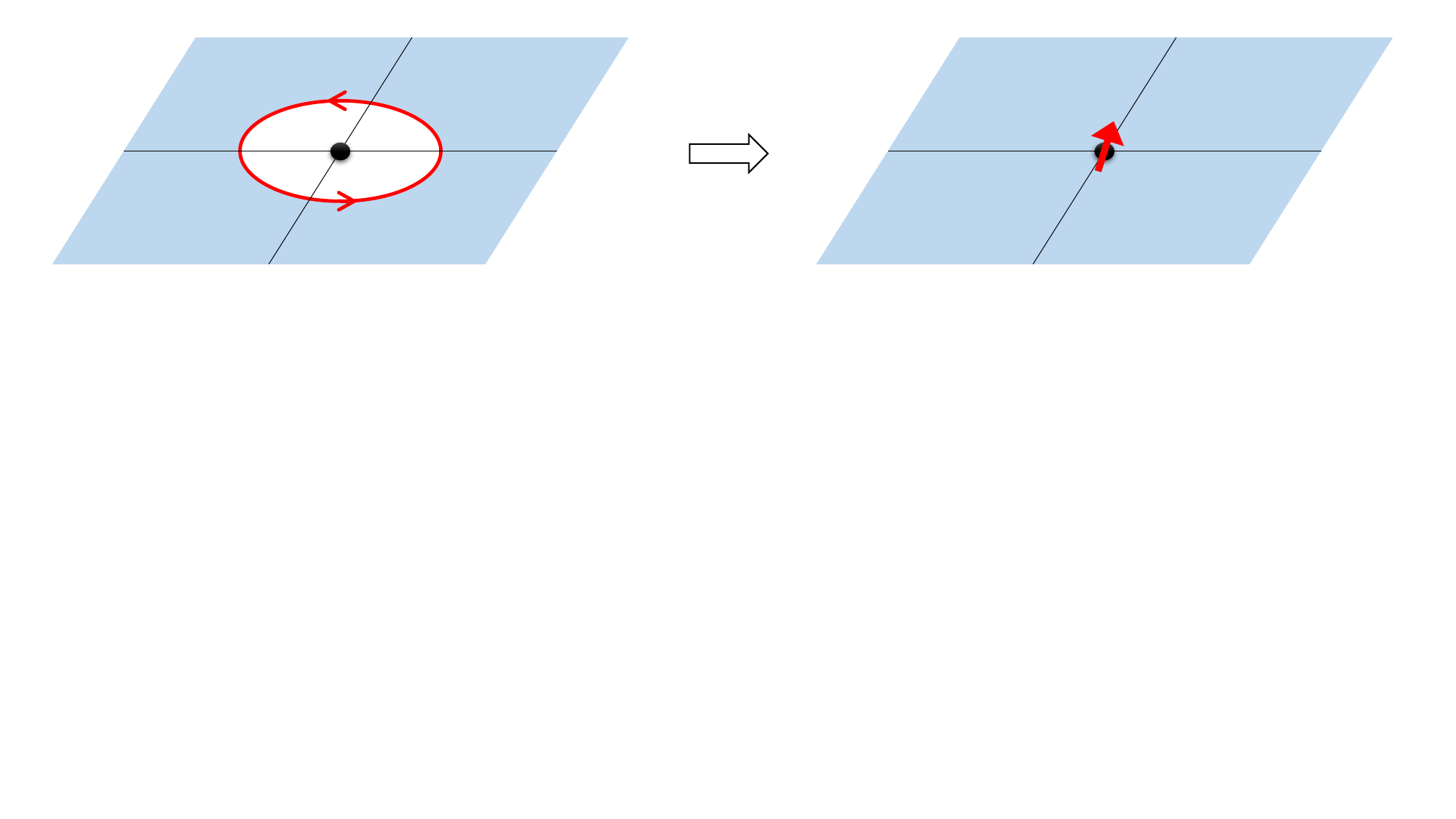}
\end{array}
\label{fig:d222}
\end{align}
There may exist an anomalous SPT$^1$ edge state localized at the hole of the SPT$^2$ phase. 
It behaves as the obstruction to shrink the SPT$^2$ phase at the 0-cell, which defines the homomorphism $d^2_{2,-2}$. 

See Secs.~\ref{sec:2d_fermion_inversion_i^2=1}, \ref{sec:2d boson with Z_4 two-fold rotation symmetry} and \ref{sec:Fermion with magnetic translation symmetry} for nontrivial examples of $d^2_{2,-1}$. 

\item
{\it $d^2_{3,-2}$: ``SPT$^2$ phases on 3-cells" $\to$ ``SPT$^1$ phases on 1-cells"}

Suppose we have an SPT$^2$ phase on a cylinder $S^1 \times \R$ wrapping a $1$-cell. 
In the presence of a rotation symmetry enforcing a symmetry flux along this 1-cell, the dimensional reduction of the SPT$^2$ phase to the 1-cell may lead to a nontrivial SPT$^1$ phase along the 1-cell: 
$$
\begin{array}{c}
\includegraphics[width=0.8\linewidth, trim=0cm 13cm 0cm 0cm]{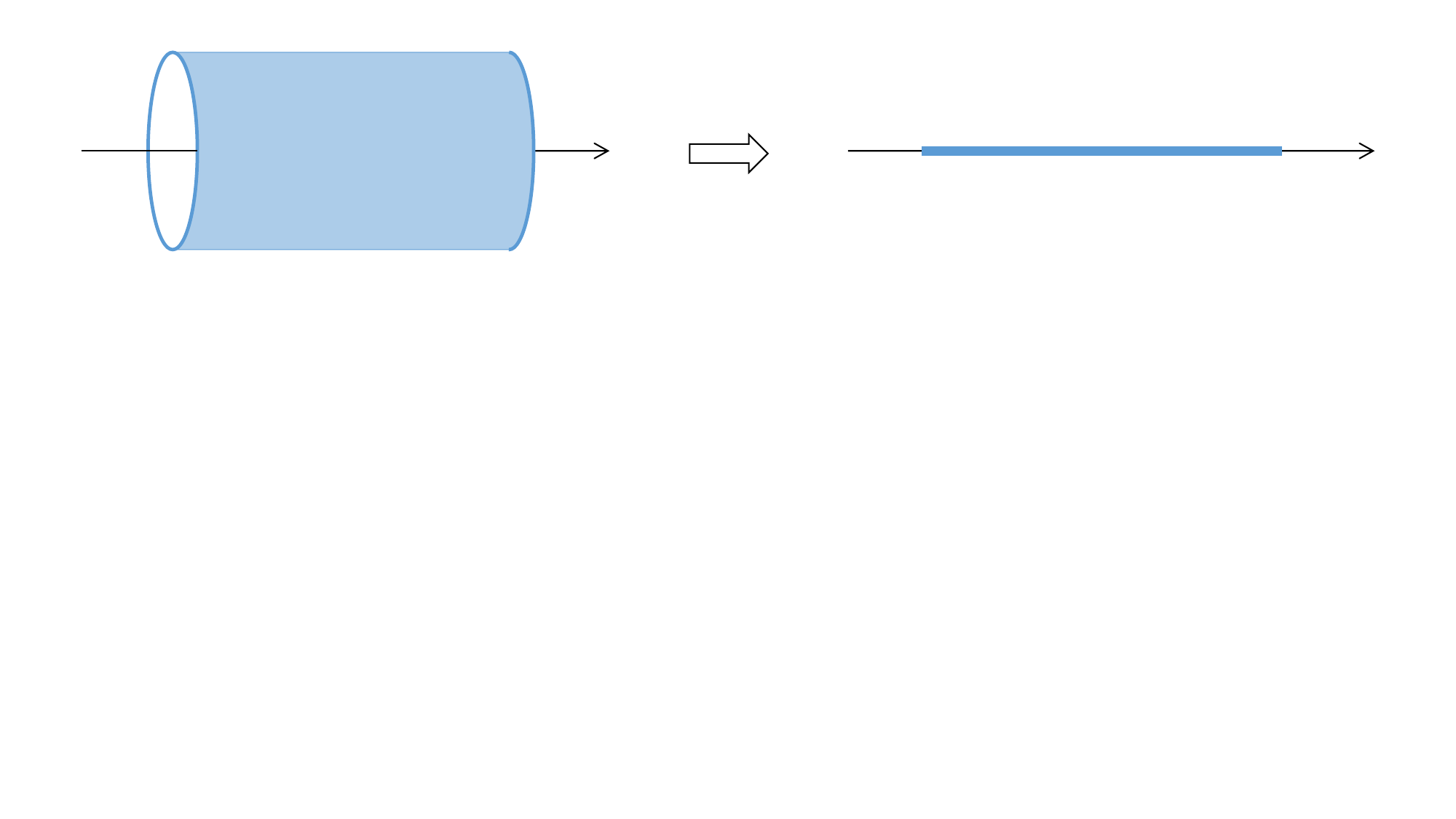}
\end{array}
$$
This defines the homomorphism $d^2_{3,-2}: E^2_{3,-2} \to E^2_{1,-1}$.

\item
{\it $d^2_{3,-3}$: ``SPT$^3$ phases on 3-cells" $\to$ ``SPT$^2$ phases on 1-cells"}

Similarly, the homomorphism $d^2_{3,-3}: E^2_{3,-3} \to E^2_{1,-2}$ measures the obstruction to collapsing SPT$^3$ states in $3$-cells on $1$-cells. 
In the presence of the flux enforced by the symmetry, an SPT$^3$ phase may have anomalous SPT$^1$ edge state localized at the $1$-cell, which defines the homomorphism $d^2_{3,-3}$: 
$$
\begin{array}{c}
\includegraphics[width=0.8\linewidth, trim=0cm 12cm 0cm 0cm]{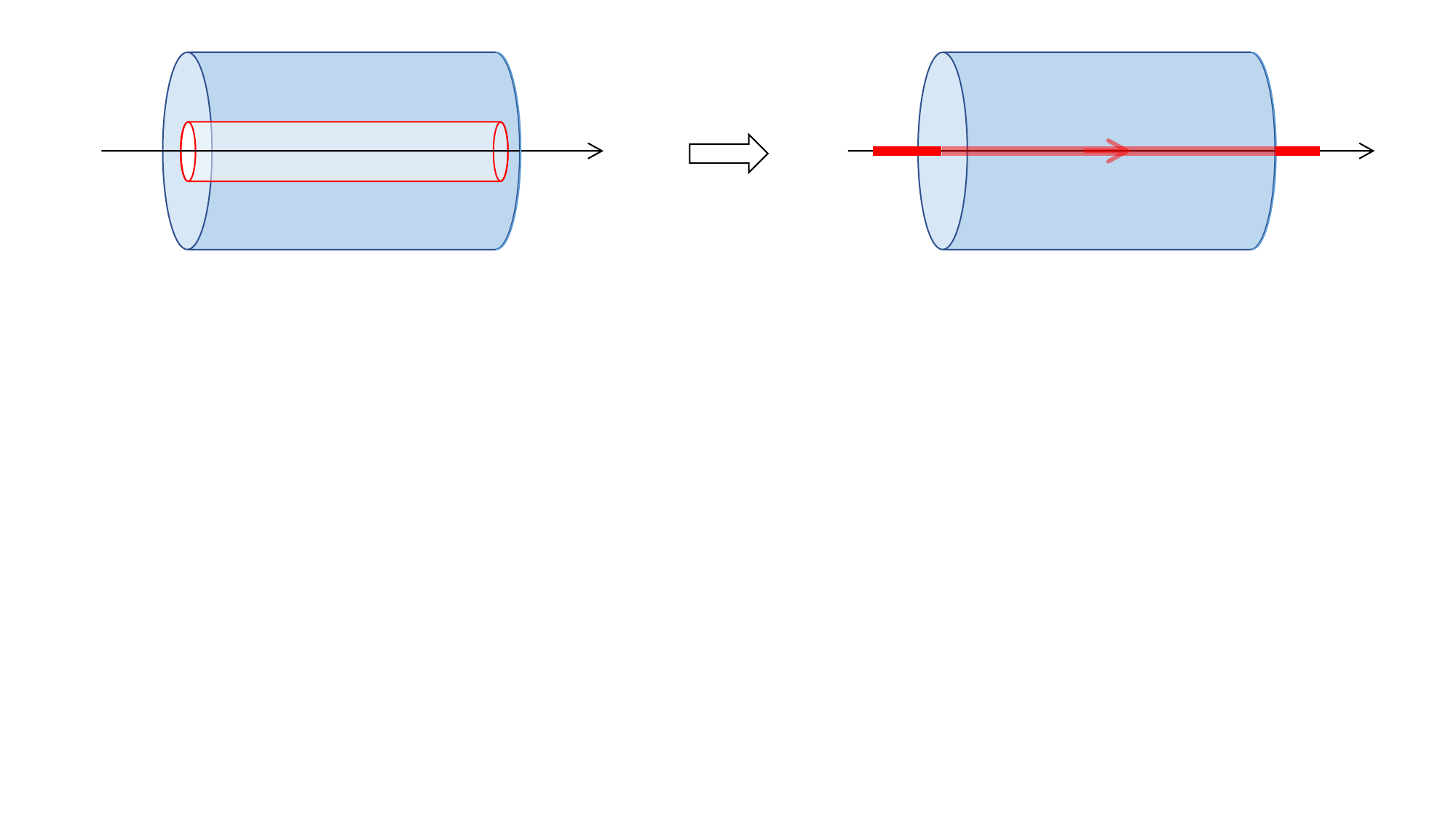}
\end{array}
$$

\item
{\it $d^3_{3,-2}$: ``SPT$^2$ phases on 3-cells" $\to$ ``SPT$^0$ phases on 0-cells"}

The definition of the homomorphism $d^3_{3,-2}:E^3_{3,-2} \to E^3_{0,0}$ is analogous to $d^2_{2,-1}$. 
$E^3_{3,-2} \subset \ker d^2_{3,-2} \subset \ker d^1_{3,-2}$ implies that an SPT$^2$ phase created adiabatically in a $3$-cell glues together except for adjacent $0$-cells. 
The homomorphism $d^3_{3,-2}$ measures the obstruction to collapsing the SPT$^2$ phase at the 0-cells: 
$$
\begin{array}{c}
\includegraphics[width=0.8\linewidth, trim=0cm 11cm 0cm 0cm]{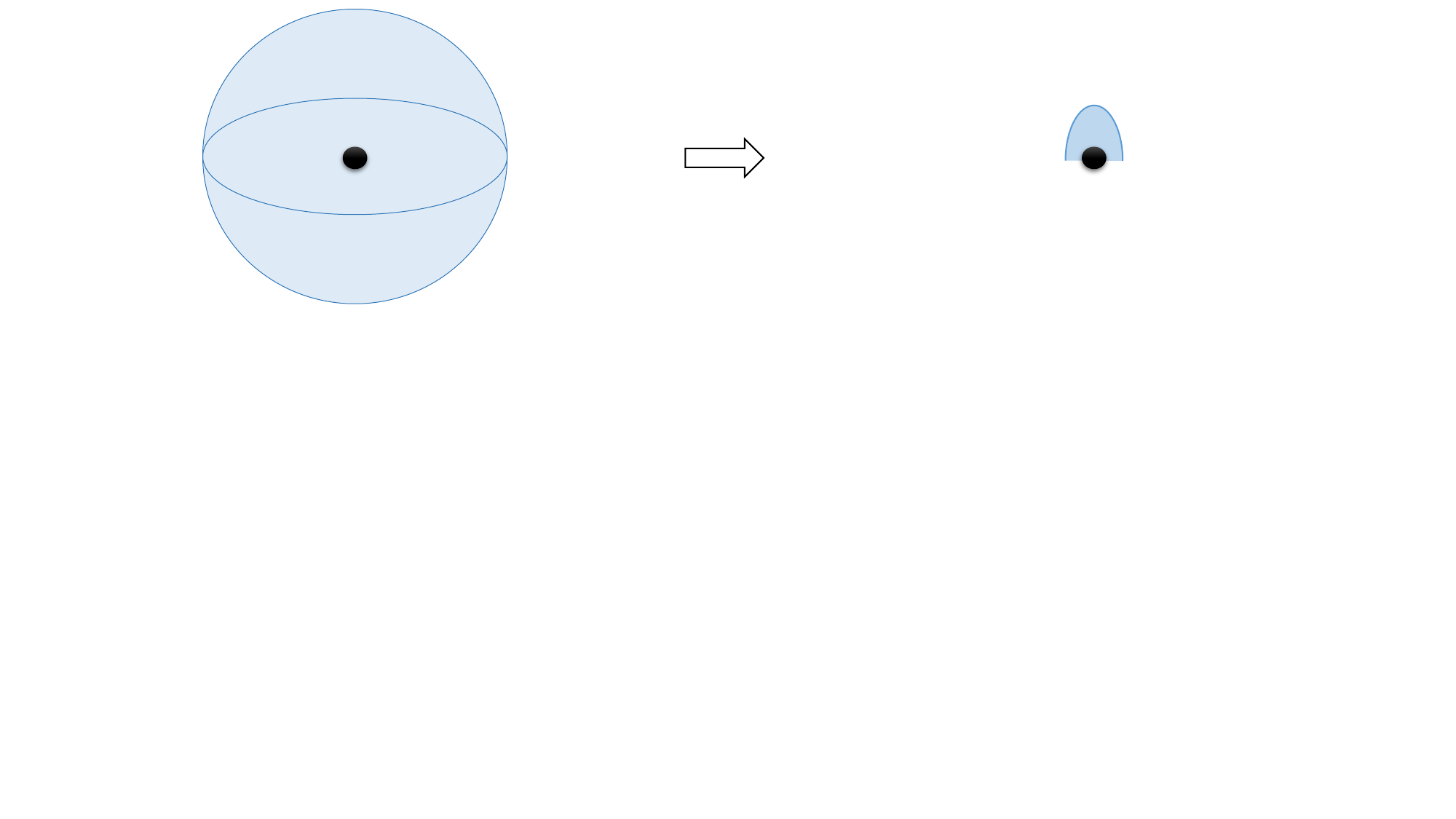}
\end{array}
$$
In general, a sort of inversion symmetry around the 0-cell may enforce a monopole flux inside the 2-sphere enclosing the 0-cell. 
In the presence of such a monopole flux, the SPT$^2$ state defined on the 2-sphere may have a nontrivial SPT$^0$ charge, which defines the homomorphism $d^3_{3,-2}$. 
See Sec.~\ref{sec:3d complex AZ classes with time-reversal inversion symmetry} for an example of nontrivial $d^3_{3,-2}$.

\item
{\it $d^3_{3,-3}$: ``SPT$^3$ phases on 3-cells" $\to$ ``SPT$^1$ phases on 0-cells"}

Similarly, the homomorphism $d^3_{3,-3}:E^3_{3,-3} \to E^3_{0,-1}$ is defined as the obstruction to collapsing an SPT$^3$ phase in a $3$-cell at $0$-cells: 
$$
\begin{array}{c}
\includegraphics[width=0.8\linewidth, trim=0cm 12cm 0cm 0cm]{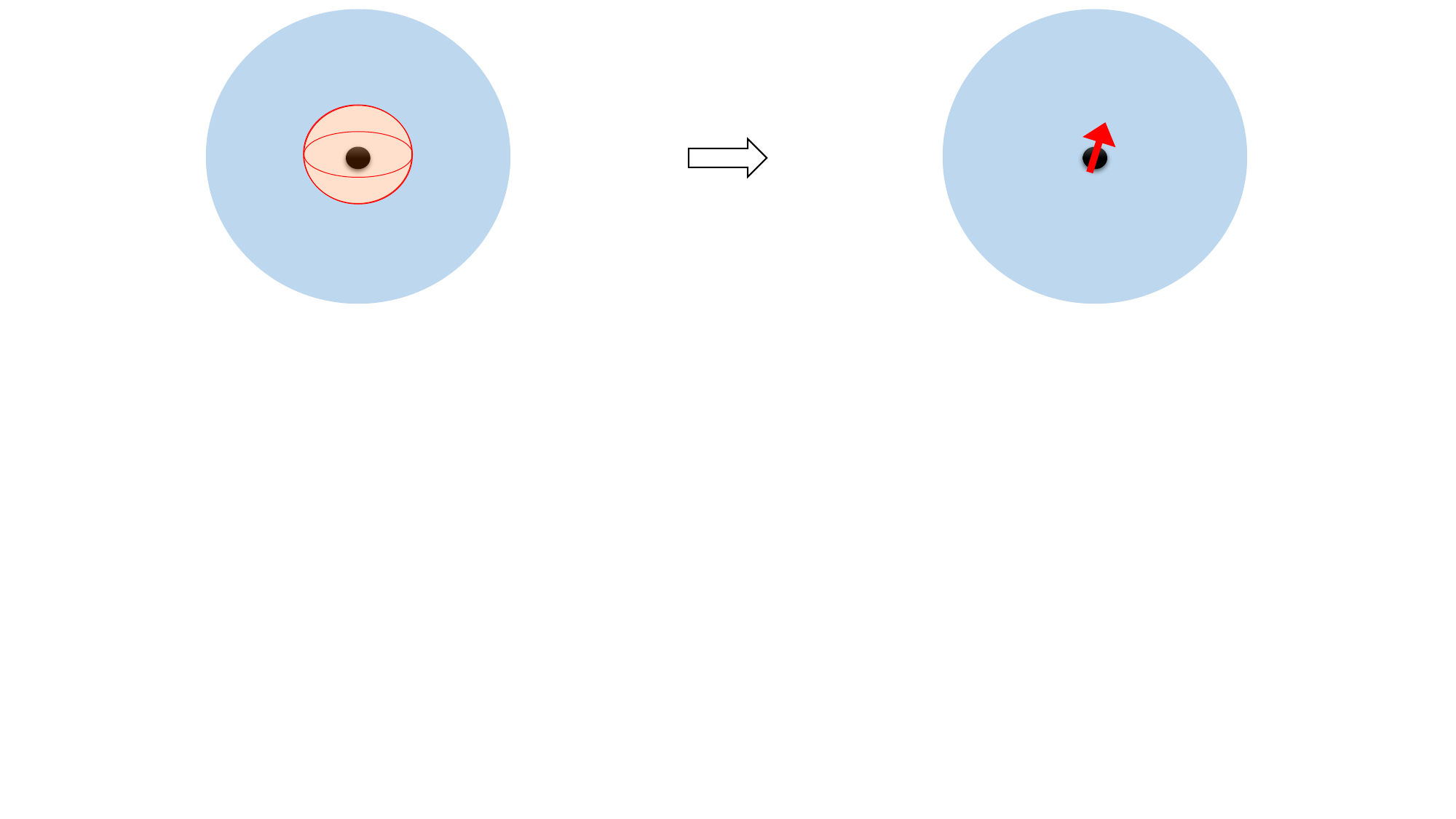}
\end{array}
$$

\end{itemize}

\subsection{Limiting page and filtration} 
The limiting page $E^{\infty}_{p,-q}$ approximates the generalized homology $h^G_{p-q}(X,Y)$, the classification of the degree $n=(p-q)$ SPT phenomena over the real space $X$ with the symmetry $G$ up to the degree $(n-1) = (p-q-1)$ SPT phenomena over the subspace $Y$. 
Let $h^G_n(X_p, X_p \cap Y) (\cong h^G_n(X_p \cup Y,Y))$ be the generalized homology over the pair of $p$-skeleta $(X_p,X_p \cap Y)$. 
$h^G_n(X_p, X_p \cap Y)$ represents the classification of degree $n$ SPT phenomena over the $p$-skeleton $X_p$ up to anomaly over $Y$, and it is understood as the collection of $k$-dimensional layer constructions where $k$ is equal to or lower than $p$. 
With this, we introduce the filtration of $h_n^G(X,Y)$ as in 
\begin{align}
F_p h_n := \im [h^G_n(X_p,X_p \cap Y) \to h^G_n(X,Y)], 
\end{align}
\begin{align}
0 \subset F_0 h_n \subset \cdots \subset F_d h_n = h_d^G(X,Y). 
\label{eq:ahss_filtration}
\end{align}
Here, the quotient $F_p h_n/F_{p-1} h_n$ is isomorphic to the limiting page $E^{\infty}_{p,n-p}$, 
\begin{align}
F_p h_n / F_{p-1} h_n \cong E^{\infty}_{p,n-p}.
\label{eq:e_infty_filtration}
\end{align}
The reason is as follows. 
$E^{\infty}_{p,n-p}$ is the set of SPT$^{p-n}$ phases on $p$-cells which are not trivialized by high-dimensional adjacent cells and consistent with low-dimensional adjacent cells. 
On the one hand, $F_p h_n$ ($F_{p-1} h_n$) is the layer construction of the degree $n$ SPT phenomena on the $p$-skeleton ($(p-1)$-skeleton). 
Therefore, an element of the quotient $F_p h_n/F_{p-1} h_n$ represents a degree $n$ SPT phenomena which is realized as a pure $p$-dimensional layer construction, which is equivalent to $E^{\infty}_{p,n-p}$. 
We should note that the structure of the filtration (\ref{eq:ahss_filtration}) of SPT phases in the presence of a crystalline symmetry and its relation to higher-order SPT phases (see Sec.~\ref{sec:Higher-order SPT phenomena}) were pointed out in Refs. \cite{HSHH17,TB18}. 
It is useful to rewrite the relations (\ref{eq:ahss_filtration}) and (\ref{eq:e_infty_filtration}) in the manner of short exact sequences:
\begin{align}
\begin{CD}
0 @>>> F_{d-1} h_n @>>> h^G_n(X,Y) @>>> E^{\infty}_{d.n-d} @>>> 0, \\
0 @>>> F_{d-2} h_n @>>> F_{d-1} h_n @>>> E^{\infty}_{d-1.n-d+1} @>>> 0, \\
\cdots \\
0 @>>> E^{\infty}_{0,n} @>>> F_{1} h_n @>>> E^{\infty}_{1,n-1} @>>> 0.
\end{CD}
\label{eq:ahss_extension}
\end{align}

\subsection{Group extension}
\label{sec:Group extension}

\begin{figure}[!]
\includegraphics[width=\linewidth, trim=0cm 12cm 0cm 0cm]{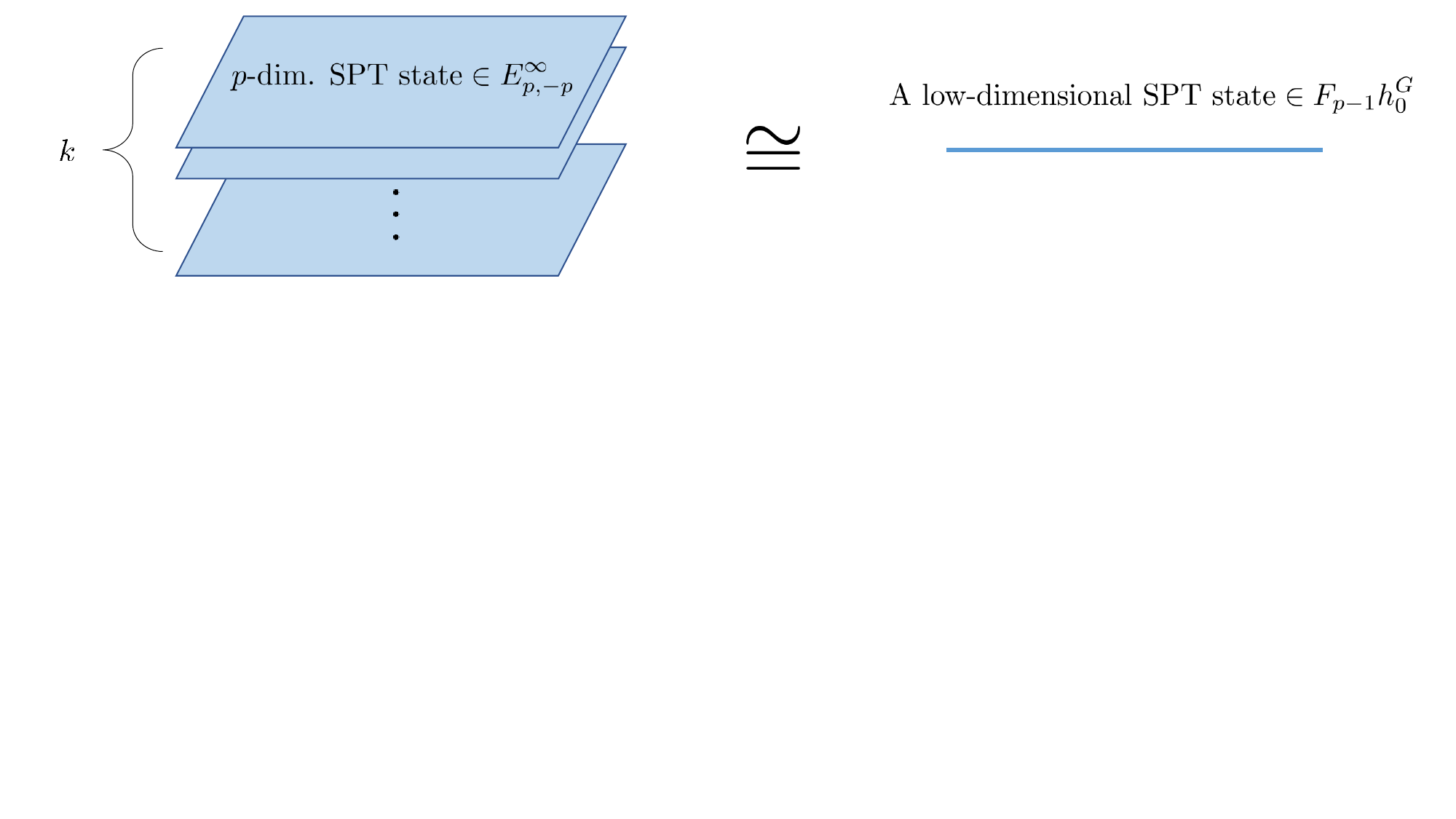}
\caption{
The physical meaning of the group extension (\ref{eq:extension_f}). 
}
\label{fig:extension}
\end{figure}

The group extension of $E^{\infty}_{p,n-p}$ by $F_{p-1} h_n$ in the short exact sequence 
\begin{align}
0 \to F_{p-1} h_n \to F_p h_n \to E^{\infty}_{p,n-p} \to 0
\label{eq:extension_f}
\end{align}
states that a degree $n$ SPT phenomenon inside $p$-cells classified by $E^{\infty}_{p,n-p}$ may be be nontrivially extended by an degree $n$ SPT phenomenon on the $(p-1)$-skeleton $X_{p-1}$. 

To clarify the point, we explain the group extension for SPT phases, i.e., $n=0$. 
(The following discussion does apply to general degrees $n \in \Z$.)
The abelian group $E^{\infty}_{p,-p}$ represents nontrivial SPT phases inside $p$-cells which can not be trivialized by higher-dimensional cells and can extend to lower-dimensional cells. 
Let us consider an SPT phase ${\bf x} \in E^{\infty}_{p,-p}$ with $k$-th order, i.e.\ $k {\bf x} = 0$ as an element of $E^{\infty}_{p,-p}$. 
The triviality of $k {\bf x} \in E^{\infty}_{p,-p}$ implies that the SPT phase $k {\bf x}$ is trivial within in $p$-cells. 
However, the SPT phase $k {\bf x}$ may remain nontrivial in lower-dimensional cells, i.e.\ the $(p-1)$-skeleton $X_{p-1}$ in which SPT phases are classified by $F_{p-1} h_0$. 
See Fig.~\ref{fig:extension} for a schematic picture. 
If this is the case, the SPT phase $k {\bf x}$ should be equivalent to a nonzero element ${\bf y} \in F_{p-1} h_{0}$, and the group extension of $E^{\infty}_{p,-p}$ by $F_{p-1} h_{0}$ becomes nontrivial. 

See Sec.~\ref{sec:1d fermions} for a simple example where the group extension is nontrivial.


\subsection{Physical properties of AHSS}
\label{sec:phys_ahss}
In this subsection, we collect relationships between the mathematical structure of the AHSS and the physics of SPT phases.

\subsubsection{Any crystalline SPT phase is a patchwork of local SPT phases}

The generalized homology formulation of crystalline SPT phases is based on the picture by Thorngren and Else that in the topological limit, the spatial scale of crystalline symmetry $a$ can be regarded to be large enough compared to the correlation length of bulk $\xi$~\cite{TE18}. 
Once we accept the Thorngren and Else prescription, we should conclude that 
\begin{itemize}
\item
Every crystalline SPT phase is made of SPT phases on high-symmetry points (0-cells), open line segments (1-cells), open polygons (2-cells), and open polyhedrons (3-cells). 
Shortly, Every crystalline phase is a ``patchwork" of local SPT phases. 
\end{itemize}
This can be explicitly seen in the filtration (\ref{eq:ahss_filtration}) of the homology $h^G_{n=0}(X,Y)$, where an SPT phase belongs to either of subgroups $F_p h_{n=0}$ that composed of $p$-dimensional cells or cells with dimension lower than $p$.

\subsubsection{Higher-order SPT phenomena}
\label{sec:Higher-order SPT phenomena}

Let $X = \R^d$ be the $d$-dimensional infinite space and $G$ be a point group symmetry acting on $X$. 
The $p$-th order SPT phases protected by a point group symmetry are SPT phases having an anomaly localized on a $(d-p)$-dimensional subspace in the boundary. 
We especially define the $d$-th order SPT phases as SPT phases without an anomalous boundary state. 
Since the existence of $(d-p)$-dimensional anomalous boundary state originated from a $(d-p+1)$-dimensional SPT phase protected by onsite symmetry defined on a $(d-p)$-dimensional layer, we can conclude that the limiting page $E^{\infty}_{p,-p}$ is nothing but the classification of $(d-p+1)$-th order SPT phases. 
More generally, we conclude that 
\begin{itemize}
\item
The classification of $p$-th order degree $n$ SPT phenomena over a real space $X$ is given by $E^{\infty}_{d-p+1,n-d+p-1}$. 
\end{itemize}

We should note that the classification of $p$-th order degree $n$ SPT phenomena itself does not imply the classification of the degree $n$ SPT phenomena, since as discussed in Sec.~\ref{sec:Group extension}, the group extension (\ref{eq:extension_f}) can be nontrivial.

\subsubsection{LSM-type theorems as the boundary of an SPT phase}
\label{sec:lsm_boundary}

Let $X$ be a $d$-dimensional space and $G$ be a point or space group acting on $X$. 
Let us focus on the $d$-th order anomalies, i.e., anomalous edge degrees of freedom classified by the abelian group $F_0 h_{-1} \cong E^{\infty}_{0,-1}$. 
In a Hilbert space ${\cal H}$ belonging to $E^{\infty}_{0,-1}$ there remains an anomaly to not have a unique gapped ground state since, from the quotient by $\im d^r_{r,-r}$ to make the limiting page, the anomaly of ${\cal H}$ can not be trivialized by higher-dimensional cells. 
This phenomenon is known as the LSM theorem as the boundary of an SPT phase~\cite{CGW11,FO17,HHG16,YHR17,PWJZ17,MT17,Che18,KSKR18}. 
We conclude that 
\begin{itemize}
\item
The classification of the LSM theorem as the boundary of an SPT phase is given by $E^{\infty}_{0,-1}$.
\end{itemize}

We would like to emphasize that the $E^2$-page is insufficient to conclude that a given system is anomalous. 
To see this, let us begin with the $E^1$-page for $3$-space dimensions 
\begin{align}
\begin{array}{c|ccccccc}
q=1 & {\rm SPT}^{1} & {\rm SPT}^{1} & &  \\
q=2 &  & {\rm SPT}^{2} & {\rm SPT}^{2} & {\rm SPT}^{2} \\
q=3 & & & {\rm SPT}^{3} & {\rm SPT}^{3} \\
\hline 
E^1_{p,-q} & p=0 & p=1 & p=2 & p=3 \\
\end{array}
\end{align}
The first differential $d^1_{1,-1}: E^{1}_{1,-1} \to E^1_{0,-1}$ represents how anomalous edge states at 0-cells are trivialized by SPT$^1$ phases on 1-cells (see (\ref{fig:d111})). 
Taking the homology of the first differentials, we get the $E^2$-page 
\begin{align}
\begin{array}{c|ccccccc}
q=1 & E^2_{0,-1} &  & &  \\
q=2 &  & & E^2_{2,-2} &  \\
q=3 & & & & E^2_{3,-3} \\
\hline 
E^2_{p,-q} & p=0 & p=1 & p=2 & p=3 \\
\end{array}
\end{align}
In the $E^2$-page, $E^2_{0,-1}$ is the set of anomalous edge states at $0$-cells which cannot be trivialized from SPT$^1$ phases on $1$-cells. 
The second differential $d^2_{2,-2}: E^2_{2,-2} \to E^2_{0,-1}$ represents how anomalous edge states at 0-cells are trivialized by SPT$^2$ phases in $2$-cells. 
This is possible in general. 
If for an SPT$^2$ phase, a symmetry flux enforced by crystalline symmetry traps an anomalous edge state, one may remove the anomalous edge state at $0$-cells via the pair-annihilation. 
Taking the homology of $d^2$, we get the $E^3$-page 
\begin{align}
\begin{array}{c|ccccccc}
q=1 & E^3_{0,-1} &  & &  \\
q=2 &  & & &  \\
q=3 & & & & E^3_{3,-3} \\
\hline 
E^3_{p,-q} & p=0 & p=1 & p=2 & p=3 \\
\end{array}
\end{align}
The third differential $d^3_{3,-3}: E^3_{3,-3} \to E^3_{0,-1}$ can be also nontrivial. 
It represents how anomalous edge states at 0-cells are trivialized by SPT$^3$ phases in $3$-cells. 
Eventually, we get the $E^{\infty} = E^4$-page 
\begin{align}
\begin{array}{c|ccccccc}
q=1 & E^4_{0,-1} &  & &  \\
q=2 &  & & &  \\
q=3 & & & & \\
\hline 
E^4_{p,-q} & p=0 & p=1 & p=2 & p=3 \\
\end{array}
\end{align}
For 3-space dimensions, $E^{\infty}_{0,-1} = E^4_{0,-1}$ gives the classification of the LSM theorem as the boundary of an SPT phase.

\subsubsection{LSM-type theorem to enforce a nontrivial SPT phase}
\label{sec:lsm_spt}

Let us focus on the $r$-th differential 
\begin{align}
d^{r}_{r,-r}: E^r_{r,-r} \to E^r_{0,-1} 
\end{align}
discussed in Sec.~\ref{sec:lsm_boundary}. 
The physical meaning of a Hilbert space ${\cal H}$ belonging to $\im d^r_{r,-r} \subset E^r_{0,-1}$ is that the edge anomalies localized at 0-cells of ${\cal H}$ can be removed by an SPT$^r$ phase in $r$-cells. 
See (\ref{fig:d222}) for $r=2$. 
Put differently, if we have a unique gapped ground state $\ket{\chi}$ in the Hilbert space ${\cal H}$, the state $\ket{\chi}$ should be a nontrivial SPT$^r$ phase composed of a $r$-dimensional layer construction. 
This is the mechanism to have the LSM theorem enforcing an SPT phase discussed in \cite{YJVR17,Lu17}. 
In sum, 
\begin{itemize}
\item
The classification of the LSM theorem enforcing a nontrivial $r$-dimensional SPT phase for a Hilbert space is given by $\im d^r_{r,-r}$. 
\end{itemize}

\subsubsection{LSM-type theorem with $U(1)$ symmetry}
\label{sec:lsm_u1}

In the presence of the $U(1)$ particle number conservation and a space group symmetry $G$, the $U(1)$ charge per unit cell, called the filling number $\nu$, is well-defined. 
We here also assume the absence of particle-hole symmetry flipping the $U(1)$ charge. 
In such systems, the LSM theorem~\cite{LSM} and its generalizations~\cite{Oshi00,Has05,WPVZ15} give a constraint on filling number $\nu$ to have a gapped unique ground state. 
We also have the LSM theorem to enforce a nontrivial SPT phase which is not an atomic insulator by a filling number~\cite{PoFillingConstraint,Lu17}. 

Let us formulate the LSM-type theorem in the presence of the $U(1)$ particle number conservation in the viewpoint of generalized homology. 
Atomic insulators are classified by the term $E^{\infty}_{0,0}$ which is generated by Wyckoff orbitals. 
The possible filling numbers $\nu \in n_{\rm AI} \Z$ of atomic insulators are determined by $E^{\infty}_{0,0}$. 
For generic insulators classified by $h^G_0(X)$, a filling number $\nu$ can be a fractional $\nu \in (n_{\rm AI}/p) \Z$ when the filling number $n_{\rm AI}\Z \in E^{\infty}_{0,0}$ leads to a nontrivial group extension (\ref{eq:ahss_extension}) of SPT phases in higher-dimensional cells. 
We have two consequences 
\begin{itemize}
\item
If the filling number $\nu$ is not in $(n_{\rm AI}/p) \Z$, then the system has no unique gapped symmetric ground state.
\item 
For a unique gapped symmetric ground state, If the filling number $\nu$ is a fractional number $\frac{q}{p} n_{\rm AI} \Z$ with $q\in \{1,\dots p-1\}$, the ground state is an SPT phase defined on a $d$-dimensional layer with $d \geq 1$. 
\end{itemize}

\section{Interacting crystalline SPT phases}
\label{sec:interacting}
In this section, we present case studies of the AHSS for interacting SPT phases. 
We leave the AHSS for free fermions to Sec.~\ref{sec:free}.

\subsection{Interacting fermions with inversion symmetry: the case of $I^2 = (-1)^F$}
\label{sec:5.1}
As a benchmark test of the generalized homology framework and the AHSS, we calculate the fermionic SPT phases with inversion symmetry $(x_1, \dots x_d) \mapsto (-x_1, \dots, -x_d)$ obeying the algebra $I^2 = (-1)^F$, where $(-1)^F$ is the fermion parity. 
The corresponding generalized homology is written by $h^{\Z_4}_n(\R^d,\p \R^d)$ where $\R^d$ is the real space and $\Z_4$ acts on $\R^d$ as the inversion. 
(The mathematically precise meaning of the homology $h^{\Z_4}(\R^d, \partial \R^d)$ is the homology of the pair of the one point compactification $\R^d \cup \{ \infty \}$ and the infinity $\{ \infty \}$.)

The classification of such SPT phases is given by the Anderson dual of the corresponding bordism groups. 
Let's denote a reflection transformation that changes only one spacetime coordinate as $R$, and a $\Z_2$ symmetry that does not change any spacetime points as $U$. 
$R^2=1$ and $R^2=(-1)^F$ correspond to $pin^+$ and $pin^-$ bordism groups, respectively, while $U^2=1$ and $U^2=(-1)^F$ correspond to ${\rm Spin} \times \Z_2$ and $({\rm Spin} \times \Z_4)/\Z_2$ bordism groups, respectively.~\footnote{In $({\rm Spin} \times \Z_4)/\Z_2$, a $2\pi$ rotation of spacetime and $U^2$, symmetry action of the element $2 \in \Z_4$ are identified.} 
Note that the corresponding bordism groups depend on the spatial dimension as the number of spatial directions to be flipped varies with the spatial dimension $d$: 
The Anderson dual of the bordism group classification assumes the symmetry of the TQFT. 
In particular, fermionic fields obey relativistic rotational symmetry, where rotational operations that do not change the orientation of spacetime accompany internal transformations of the internal spinor space, and a phase given by the fermion parity operator $(-1)^F$ arises for a $2\pi$ rotation of spacetime. 
As a result, the correspondence between the spatial inversion $I$ satisfying $I^2=(-1)^F$ and the four bordism groups described above is given as follows for spatial $d$ dimensions:
In the case of $d \equiv 1 \mod 4$, $R^2=(-1)^F$, i.e., ${\rm pin}^-$ bordism.
In the case of $d \equiv 2 \mod 4$, $U^2=1$, i.e., ${\rm Spin} \times \Z_2$ bordism.
In the case of $d \equiv 3 \mod 4$, $R^2=1$, i.e., ${\rm pin}^+$ bordism.
In the case of $d \equiv 0 \mod 4$, $U^2=(-1)^F$, i.e., $({\rm Spin} \times \Z_4)/\Z_2$ bordism.

\subsubsection{1d fermions}
\label{sec:1d fermions}
Let us consider $1d$ fermions over the infinite real space $\R$. 
The inversion symmetric cell decomposition is given as Fig.~\ref{fig:cell_inversion}[a] which is composed of the $0$-cell $A$ and the 1-cell $a$. 
The $E^1$-page is given by 
\begin{align}
\begin{array}{c|ccccccc}
q=0 & \Z_4 & \Z_2 \\
q=1 & 0 & \Z_2 \\
\hline 
E^1_{p,-q} & p=0 & p=1 \\
\end{array}
\end{align}
Here, $E^1_{0,0} = \Z_4$ is generated by the occupied state $f^{\dag} \ket{0}$ of a complex fermion $f^{\dag}$ with the inversion eigenvalue $I f^{\dag} I^{-1} = i f^{\dag}$, where $\ket{0}$ is the Fock vacuum. 
Similarly, $E^1_{1,0} = \Z_2$ is generated by the occupied state $f^{\dag} \ket{0}$. 
$E^1_{1,-1} = \Z_2$ is generated by the topologically nontrivial Kitaev chain on the 1-cell $a$. 
The reason of that $E^1_{0,-1} = 0$ is as follows: 
Due to the $\Z_4$ permutation symmetry generated by $I$, edge Majorana fermions should appear as a pair $\{ \gamma_1, \gamma_2\}$ with $I \gamma_1 I^{-1} = \gamma_2$ and $I \gamma_2 I^{-1}= -\gamma_2$. The Majorana fermions form a complex fermion $f=(\gamma_1+i \gamma_2)/2$, making the Hilbert space nonanomalous. 

\begin{figure}[!]
\includegraphics[width=\linewidth, trim=0cm 9cm 0cm 0cm]{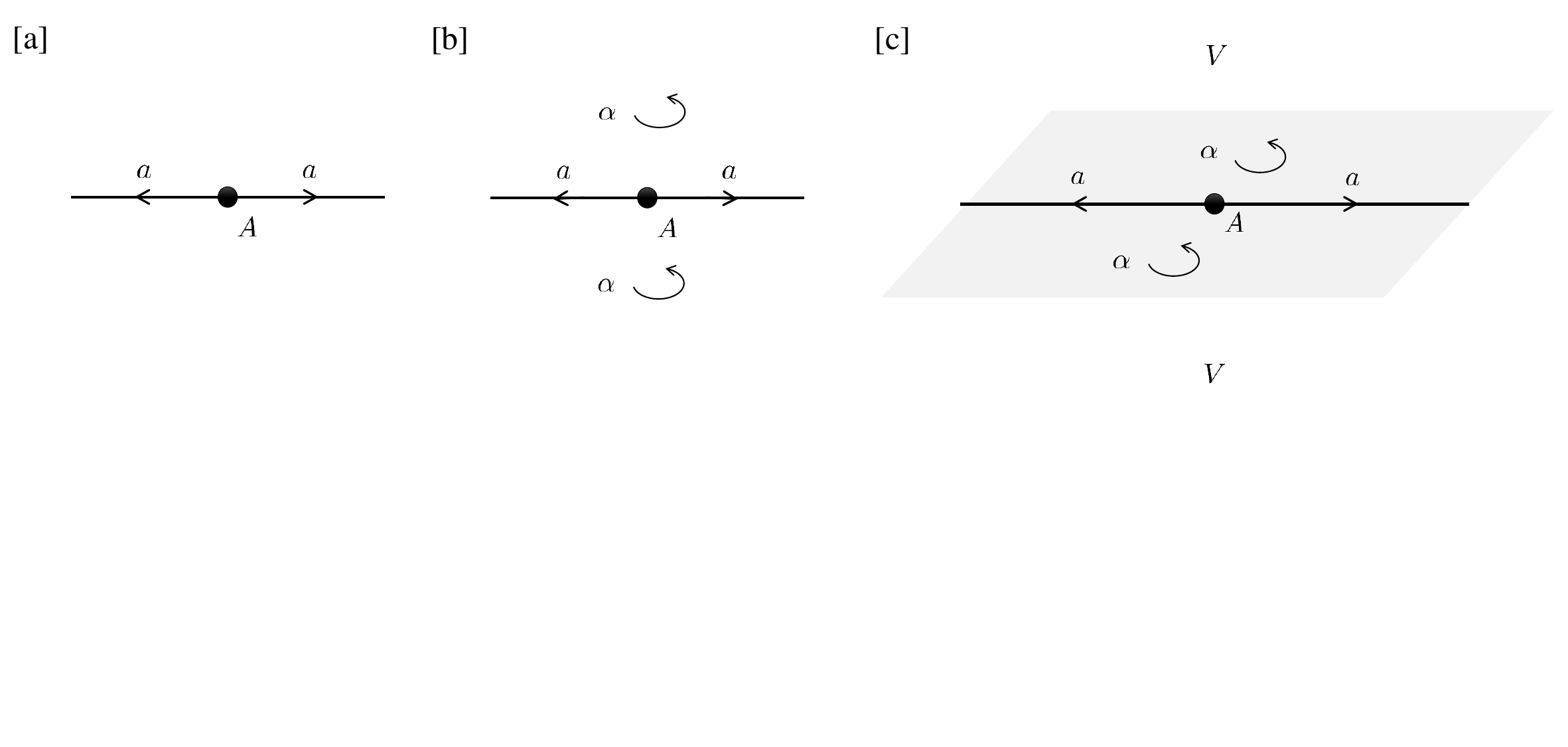}
\caption{
[a],[b] and [c]: inversion symmetric cell decomposition of $\R^1$, $\R^2$ and $\R^3$, respectively.
}
\label{fig:cell_inversion}
\end{figure}

The first differential $d^1_{1,0}$ is found to be trivial $d^1_{1,0}=0$. 
The first differential $d^1_{1,0}$ represents a pair creation of complex fermions $f^{\dag}_1f^{\dag}_2$ at the $1$-cell $a$ and its inversion image $f^{\dag}_3 f^{\dag}_4$, and moving fermions to the center and the infinite with preserving the inversion symmetry: 
\begin{align}
\begin{array}{c}
\includegraphics[width=\linewidth, trim=0cm 17cm 0cm 0cm]{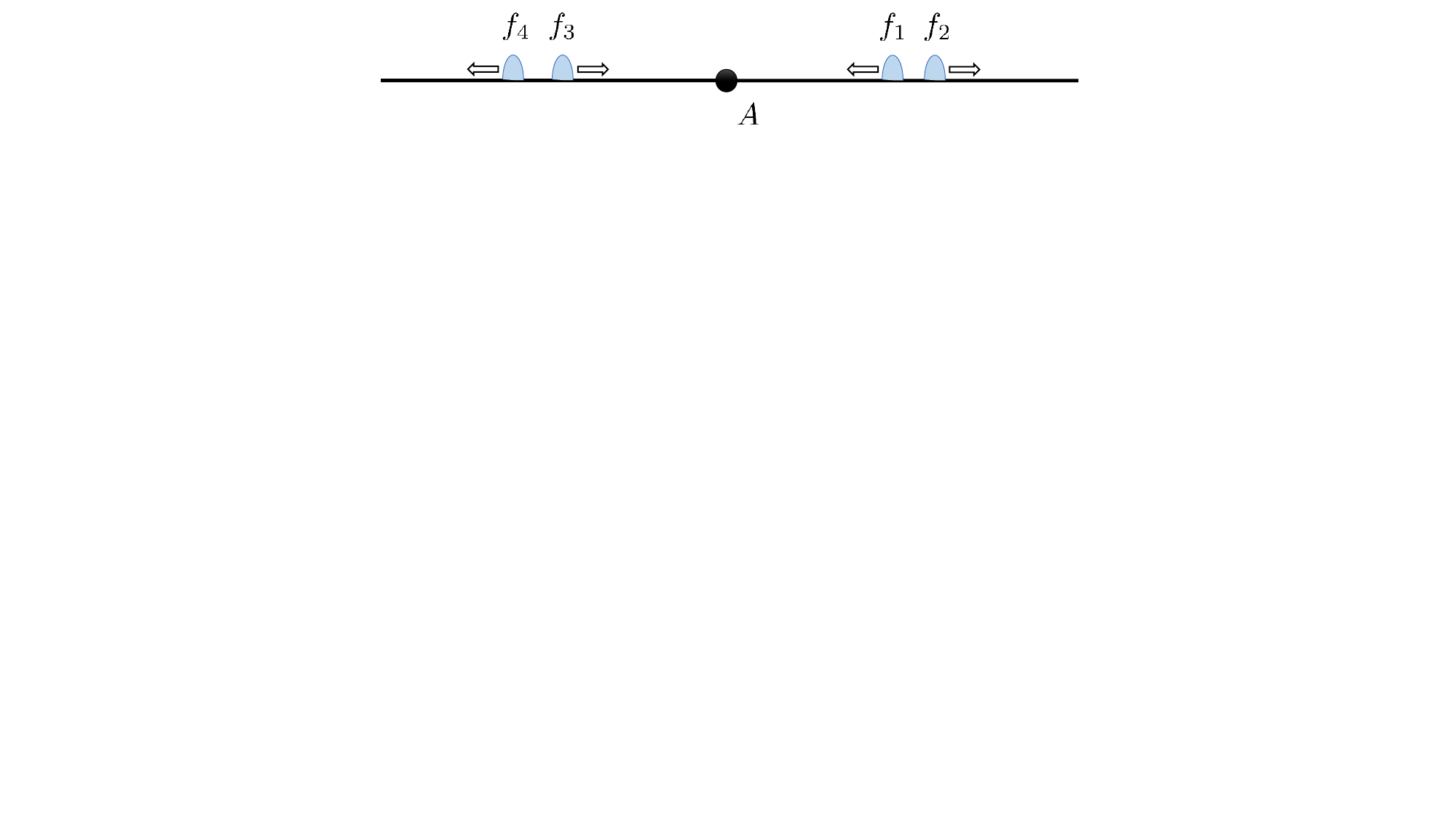}
\end{array}
\label{fig:1d_inv_d110}
\end{align}
At the inversion center, the inversion symmetric pair $f^{\dag}_1f^{\dag}_3$ of complex fermions may trivialize a part of the SPT$^0$ phase classified by $E^1_{0,0}$.  
However, the state $f^{\dag}_1f^{\dag}_3$ is trivial since the $\Z_4$ structure of inversion $I^2 = (-1)^F$ implies that $I$ is the permutation $I f^{\dag}_1 I^{-1} = f^{\dag}_3$ and $I f^{\dag}_3 I^{-1} = - f^{\dag}_1$ and it means the trivial $U(1)$ phase $I f^{\dag}_1 f^{\dag}_3 I^{-1} = f^{\dag}_1 f^{\dag}_3$ under $I$. 

The first differential $d^1_{1,-1}$ is trivial because $E^1_{0,-1}=0$.
This means a pair of topologically nontrivial Kitaev chains defined locally on the 1-cell $a$ and the inversion image can glue together at the inversion center without breaking the inversion symmetry.  

The $E^2$-page is the same as $E^1$, and it is the limiting page $E^{\infty}$. 
The classification of SPT phases fits into the short exact sequence 
\begin{align}
0 \to \underbrace{\Z_4}_{E^{\infty}_{0,0}} \to h^{\Z_4}_0(\R,\p \R) \to \underbrace{\Z_2}_{E^{\infty}_{1,-1}} \to 0.
\label{eq:1d_fermion_inversion_extension}
\end{align}
To determine the group extension, we ask if the double of $\Z_2$ nontrivial Kitaev chains inside the 1-cell is equivalent to the generator of $E^{\infty}_{0,0} = \Z_4$, a complex fermion at the inversion center. 
The only facts necessary for calculating the group extension are 
\begin{itemize}
\item
For four Majorana fermions $\{ a,b,c,d \}$ there is a 1-parameter family of gapped Hamiltonian $H(\theta)$ which switches the inter-Majorana hopping terms 
\begin{align}
H(\theta) = \cos \theta (iab+icd) + \sin \theta(iac-ibd). 
\end{align}
Especially, $H(0) = iab+icd$ is adiabatically deformed to $H(\pm \pi/2) = \pm (iac-ibd)$. 
\item 
The nontrivial Kitaev chain with the periodic boundary condition (corresponding to the $\pi$-flux piercing the chain) has the odd fermion parity~\cite{KitaevUnpaired}. 
\end{itemize}
A Hamiltonian for the topologically nontrivial Kitaev chain classified by $E^{\infty}_{1,-1}$ is given by 
\begin{align}
H_{\nu} = 2\sum_{x \in \Z+1/2} f^{\dag}_{x+1,\nu} (f_{x,\nu}+f^{\dag}_{x,\nu}) + h.c. = \sum_{x \in \Z+1/2} i b_{x,\nu} a_{x+1,\nu}, 
\label{eq:majorana_equivalence}
\end{align}
with the inversion symmetry $I f^{\dag}_{x,\nu} I^{-1} = i f^{\dag}_{-x,\nu}$ and $\nu$ the flavor index. 
We have introduced the Majorana fermions by $a_{x,\nu} = f_{x,\nu}+f^{\dag}_{x,\nu}$ and $b_{x,\nu} = -i (f_{x,\nu}-f^{\dag}_{x,\nu})$, and the inversion acts on them by $I a_{x,\nu} I^{-1} = b_{-x,\nu}$ and $I b_{x,\nu} I^{-1} = - a_{-x,\nu}$. 
It should be noticed that the inversion of the Kitaev chain representing $E^{\infty}_{1,-1}$ is {\it the bond center inversion}, since $E^{2}_{1,-1}$ is defined by gluing the left and right edge Majorana fermions at the inversion center (See Fig.~\ref{fig:1d_inversion_3}[a]). 
Now we consider the two Kitaev chains $H_{\ua} + H_{\da}$. 
Applying the equivalent relation (\ref{eq:majorana_equivalence}) to two quartets of Majorana fermions $\{ b_{-\frac{3}{2},\ua}, a_{-\frac{1}{2},\ua}, b_{-\frac{3}{2},\da}, a_{-\frac{1}{2},\da}\}$ and 
$\{ b_{\frac{1}{2},\ua}, a_{\frac{3}{2},\ua}, b_{\frac{1}{2},\da}, a_{\frac{3}{2},\da}\}$ with preserving the inversion symmetry, we find that the Hamiltonian $H_{\ua} + H_{\da}$ decouples at $x=-1$ and $x=1$. 
By moving the decoupled Kitaev chains at $|x| >1$ to the infinite, the Hamiltonian $H_{\ua}+H_{\da}$ is recast as the Kitaev chain composed of $8$ Majorana fermions (Fig.~\ref{fig:1d_inversion_3}[b])
\begin{align}
H_{\ua}+H_{\da} \sim 
i b_{-\frac{1}{2} \ua} a_{\frac{1}{2} \ua} 
+i b_{-\frac{1}{2} \da} a_{\frac{1}{2} \da} 
+i a_{-\frac{1}{2} \da} a_{-\frac{1}{2} \ua} 
+i b_{\frac{1}{2} \da} b_{\frac{1}{2} \ua}.
\label{eq:1d_8_majoranas}
\end{align}
Notice that the inversion symmetry imposed on the relative sign in between $i a_{-\frac{1}{2} \da} a_{-\frac{1}{2} \ua}$ and $i b_{\frac{1}{2} \da} b_{\frac{1}{2} \ua}$ to be $1$. 
The r.h.s.\ of (\ref{eq:1d_8_majoranas}) obeys the periodic boundary condition, i.e., the $\pi$-flux inserted in the Kitaev chain. 
Therefore, we conclude that the Hamiltonian $H_{\ua}+H_{\da}$ traps an odd fermion parity at the inversion center, implying that $H_{\ua}+H_{\da}$ is equivalent to an odd state $1$ or $3 \in E^{\infty}_{0,0} = \Z_4$. 
This means the group extension (\ref{eq:1d_fermion_inversion_extension}) is nontrivial. 
The homology group becomes $h^{\Z_2}_0(\R) = \Z_8$, which is consistent with the pin$_-$ cobordism group $\Omega^{{\rm pin}_-}_2 = \Z_8$~\cite{KTTW15}. 

\begin{figure}[!]
\includegraphics[width=\linewidth, trim=0cm 13cm 0cm 0cm]{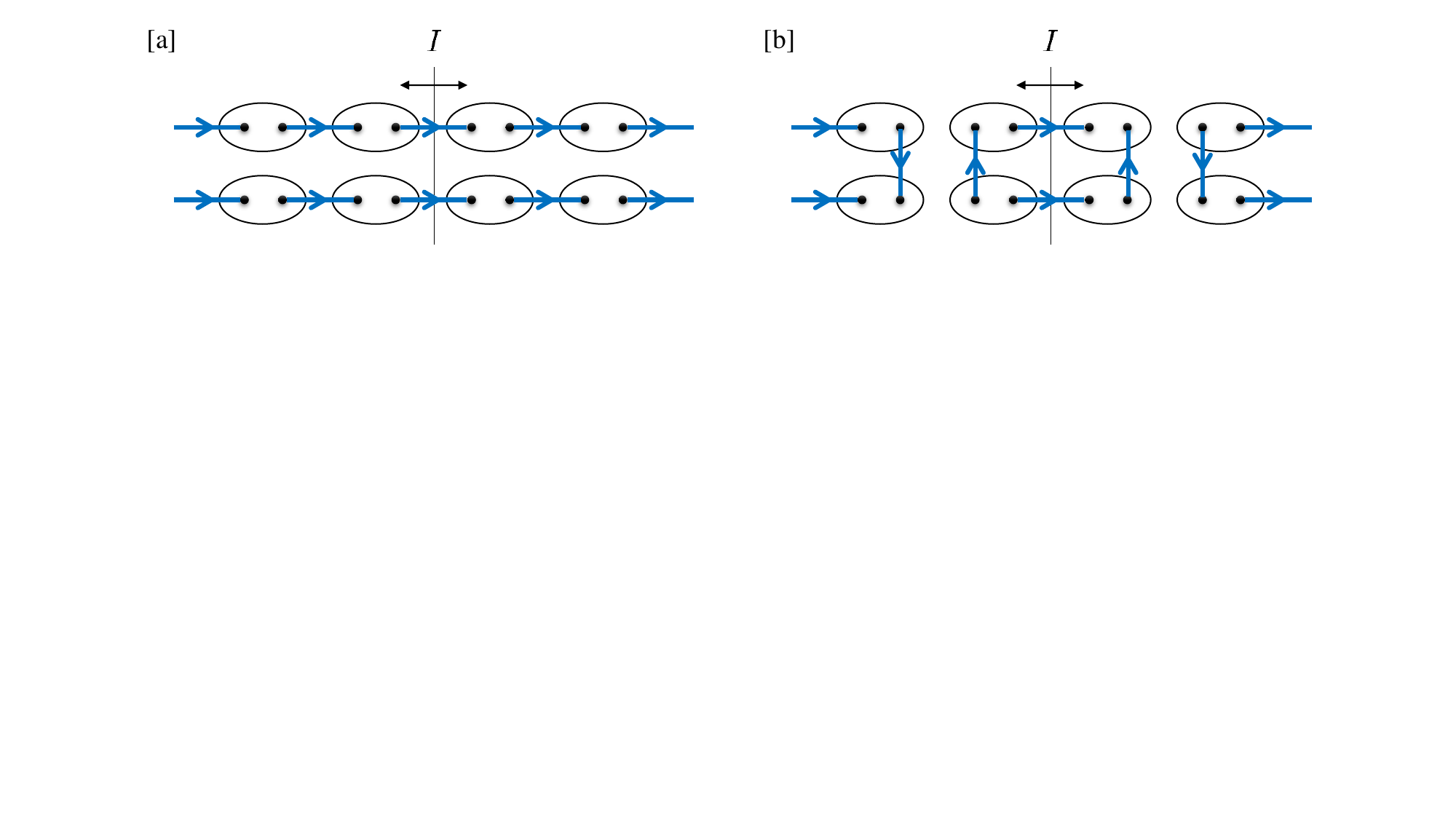}
\caption{
[a] The double nontrivial Kitaev chains preserving the inversion symmetry. 
[b] An adiabatically equivalent Hamiltonian effectively consisting only of four sites around the inversion center. 
}
\label{fig:1d_inversion_3}
\end{figure}

It is instructive to derive the same result using the Dirac Hamiltonian in the continuous system. 
Near the critical point, the Kitaev chain is described by the Dirac Hamiltonian ${\cal H} = -i \p_x \tau_y + m \tau_z$ on the basis of the Nambu fermion $\Psi(x) = (f(x),f^{\dag}(x))^T$. 
We assume that $m < 0$ for nontrivial and $m>0$ for trivial phases.
The inversion is defined by $I f^{\dag}(x) I^{-1} = i f^{\dag}(-x)$. 
It is found that the double stack of Hamiltonians ${\cal H}_{\ua} \oplus {\cal H}_{\da}$ admits a inversion symmetric mass term $M(x)$ as in 
\begin{align}
{\cal H} = -i \p_x \tau_y + m \tau_z + M(x) \tau_x \sigma_y, \qquad M(-x) = - M(x), 
\end{align}
where $\sigma_{\alpha} (\alpha \in 0,x,y,z)$ is the Pauli matrix for the layer indices $\{\ua,\da\}$. 
At the inversion center, there exists a single kink. 
For a kink with $M(x) > 0$ for $x>0$, the two Jackiw-Rebi soliton modes
\begin{align}
\begin{pmatrix}
1 \\
0 \\
\end{pmatrix}_{\tau}
\otimes 
\begin{pmatrix}
1\\
i\\
\end{pmatrix}_{\sigma}
e^{-\int^xM(x') dx'}, \qquad 
\begin{pmatrix}
0 \\
1 \\
\end{pmatrix}_{\tau}
\otimes 
\begin{pmatrix}
1\\
-i\\
\end{pmatrix}_{\sigma}
e^{-\int^xM(x') dx'} 
\end{align}
appear, and the effective low-energy Hamiltonian within the localized modes reads ${\cal H}_{\rm eff} = m \tau_z$. 
The localized modes contribute to the ground state via the creation operator 
\begin{align}
f^{\dag}_{\rm loc} \sim \int dx e^{-\int^x M(x') dx'} \left\{ f^{\dag}_{\ua}(x)-i f^{\dag}_{\da}(x) \right\}. 
\end{align}
This has the inversion eigenvalue $I f^{\dag}_{\rm loc} I^{-1} = i f^{\dag}_{\rm loc}$, meaning that the localized mode $f^{\dag}_{\rm loc}$ generate $E^1_{0,0} = \Z_4$. 

\subsubsection{2d fermions}
\label{sec:2d_fermion_inv_even}
Let us consider $2d$ fermions with inversion symmetry $I: (x,y) \mapsto (-x,-y)$ with $I^2 = (-1)^F$. 
An inversion symmetric cell decomposition of the infinite real space $\R^2$ is given as Fig.~\ref{fig:cell_inversion}[b].
The $E^1$-page is 
\begin{align}
\begin{array}{c|ccccccc}
q=0 & \Z_4 & \Z_2 & \Z_2 \\
q=1 & 0 & \Z_2 & \Z_2 \\
q=2 &  & \Z & \Z \\
\hline 
E^1_{p,-q} & p=0 & p=1 & p=2 \\
\end{array}
\end{align}
Here, $E^1_{p \in \{1,2\},-2} = \Z$ is generated by the $(p_x+i p_y)$ superconductor.
$E^1_{0,-2}$ remains blank since it does not matter to SPT phases. 
It is easy to see that the first differentials are trivial in this table. 
For example, the first differential $d^1_{2,-2}: \Z \to \Z$ represents how the boundary anomalies of $(p_x+i p_y)$ states in 2-cells contribute to anomalous edges in 1-cells. 
An inversion symmetric pair of chiral edges cancels out, which means $d^1_{2,-2}=0$:
\begin{align}
\begin{array}{c}
\includegraphics[width=0.8\linewidth, trim=0cm 12cm 0cm 0cm]{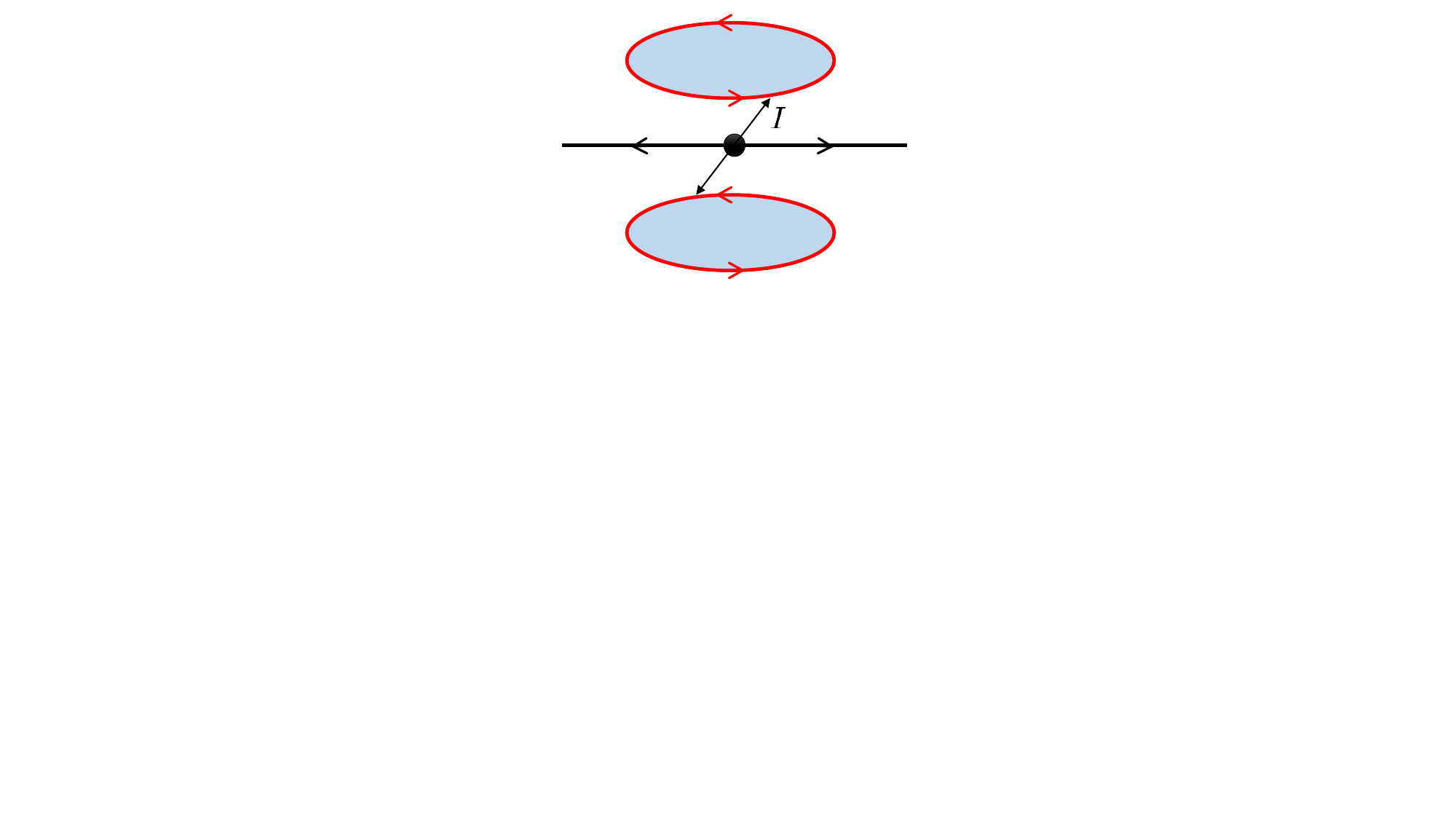}
\end{array}
\label{fig:2d_inv_chern}
\end{align}
The $E^2$-page is 
\begin{align}
\begin{array}{c|ccccccc}
q=0 & \Z_4 & \Z_2 & \Z_2 \\
q=1 & 0 & \Z_2 & \Z_2 \\
q=2 &  & & \Z \\
\hline 
E^2_{p,-q} & p=0 & p=1 & p=2 \\
\end{array}
\label{eq:2d_fermion_inversion_e2}
\end{align}
($E^2_{1,-2}$ is not determined since we did not specify $d^1_{1,-2}$.)
In this table, $d^2_{2,-1}$ can be nontrivial. It represents how an adiabatically created nontrivial Kitaev chain enclosing the inversion center $A$ trivializes SPT$^0$ state at the inversion center $A$. 
We find that $d^2_{2,-1}$ is trivial. 
This is because $I^2 = (-1)^F$ implies the anti-periodic boundary condition (NS sector) of the Kitaev chain, and its ground state is the disc state of the invertible spin TQFT, meaning that the ground state of the Kitaev chain with the inversion symmetry with $I^2 = (-1)^F$ can collapse.~\footnote{
Explicitly, the ground state of the Kitaev chain Hamiltonian composed of an even number of sites enclosing the inversion center $A$ is written as 
\begin{align}
\ket{GS_{\rm ns}} = \sum_{j \in {\rm even}} \sum_{1 \leq p_1 \leq \cdots \leq p_j \leq 2L} f^{\dag}_{p_1} \cdots f^{\dag}_{p_j} \ket{0}.
\end{align}
Under the inversion $I f^{\dag}_x I^{-1} = i f^{\dag}_{x+L}$, this state has no $U(1)$ phase. 
}
As a result, the $E^2$-page displayed in (\ref{eq:2d_fermion_inversion_e2}) is the limit. 

The homology group $h^{\Z_4}_0(\R^2,\p \R^2)$, the classification of SPT phases, fits into the short exact sequences 
\begin{equation}\begin{split}
&0 \to F^1 h_0 \to h^{\Z_4}_0(\R^2,\p \R^2) \to \underbrace{\Z}_{E^{\infty}_{2,-2}} \to 0, \\
&0 \to \underbrace{\Z_4}_{E^{\infty}_{0,0}} \to F^1 h_0 \to \underbrace{\Z_2}_{E^{\infty}_{1,-1}} \to 0.
\end{split}\end{equation}
We have already determined the latter extension $F^1 h_0 = \Z_8$ in Sec.~\ref{sec:1d fermions}, yielding to $h^{\Z_4}_0(\R^2,\p \R^2) = \Z_8 \times \Z$, which is consistent with ${\rm Tor}\Omega_3(B\Z_2) \times {\rm Free} \Omega_4(B\Z_2)$.

\subsubsection{3d fermions}
\label{sec:3d_fermion_inv_odd}
Let us consider $3d$ fermions with inversion symmetry $I: (x,y,z) \mapsto (-x,-y,-z)$ with $I^2 = (-1)^F$. 
An inversion symmetric cell decomposition of the infinite real space $\R^3$ is given as Fig.~\ref{fig:cell_inversion}[c].
The $E^1$-page is 
\begin{align}
\begin{array}{c|ccccccc}
q=0 & \Z_4 & \Z_2 & \Z_2 & \Z_2 \\
q=1 & 0 & \Z_2 & \Z_2 & \Z_2 \\
q=2 &  & \Z & \Z & \Z \\
q=3 & & 0 & 0 & 0 \\
\hline 
E^1_{p,-q} & p=0 & p=1 & p=2 & p=3 \\
\end{array}
\end{align}
The first differential $d^1_{3,-2}: \Z \to \Z$, which represents how adiabatically created $(p_x+ip_y)$ states in 3-cells trivialize ones in 2-cells, is nontrivial. 
Because the inversion does not change the Chern number, we find that $d^1_{3,-2} = 2$: 
$$
\includegraphics[width=\linewidth, trim=0cm 13cm 0cm 0cm]{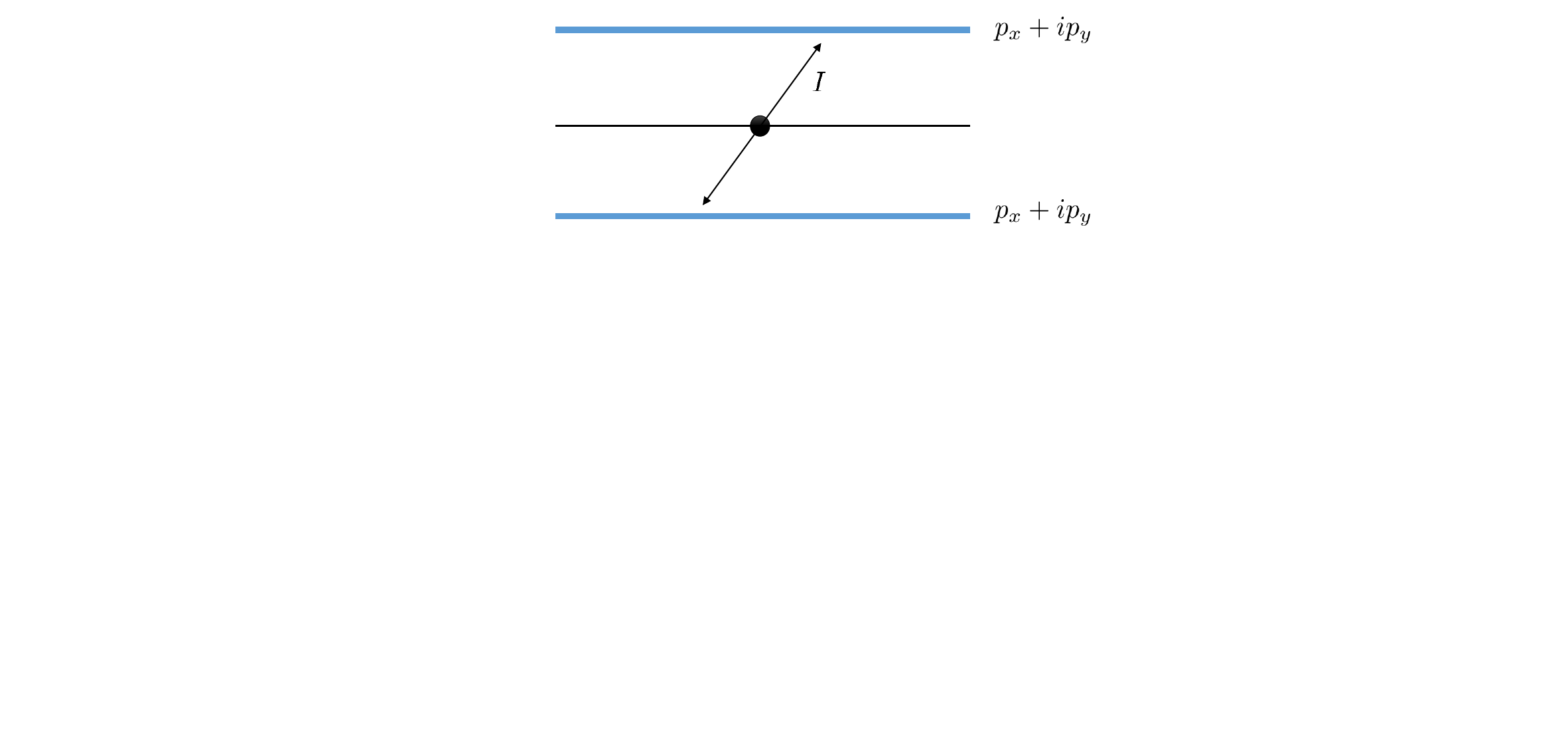}
$$
The homology of $d^1$ gives the $E^2$-page 
\begin{align}
\begin{array}{c|ccccccc}
q=0 & \Z_4 & \Z_2 & \Z_2 & \Z_2 \\
q=1 & 0 & \Z_2 & \Z_2 & \Z_2 \\
q=2 &  & & \Z_2 & 0 \\
q=3 & & 0 & 0 & 0 \\
\hline 
E^2_{p,-q} & p=0 & p=1 & p=2 & p=3 \\
\end{array}
\label{eq:3d_fermion_inv_e2}
\end{align}
For SPT phases, the $E^2$-page is sufficient to give the $E^{\infty}$ terms. 
The homology group $h^{\Z_4}_0(\R^3,\p R^3)$ fits into the exact sequences 
\begin{equation}\begin{split}
&0 \to F^2 h_0 \to h^{\Z_4}_0(\R^3,\p \R^3) \to \underbrace{0}_{E^{\infty}_{3,-3}} \to 0, \\
&0 \to F^1 h_0 \to F^2 h_0 \to \underbrace{\Z_2}_{E^{\infty}_{2,-2}} \to 0, \\
&0 \to \underbrace{\Z_4}_{E^{\infty}_{0,0}} \to F^1 h_0 \to \underbrace{\Z_2}_{E^{\infty}_{1,-1}} \to 0.
\end{split}\end{equation}
We already solved the third extension in Sec.~\ref{sec:1d fermions}. 
The above exact sequences are recast into 
\begin{align}
0 \to \underbrace{\Z_8}_{F^1 h_0} \to h^{\Z_4}_0(\R^3,\p \R^3) \to \underbrace{\Z_2}_{E^{\infty}_{2,-2}} \to 0. 
\label{eq:3d_fermion_inv_extension}
\end{align}
One can fix the extension of $E^{\infty}_{2,-2}$ by $F^1 h_0$ in a way similar to Sec~\ref{sec:1d fermions}. 
It is useful to describe the $(p_x+i p_y)$ state by the coupled wire 
\begin{equation}\begin{split}
H_{\nu} 
&= \sum_{x \in \Z+1/2} \int d y \left\{ L_{x,\nu}(y) i \p_y L_{x,\nu}(y) - R_{x,\nu}(y) i \p_y R_{x,\nu}(y) \right\}  + \sum_{x \in \Z+1/2} it \int dy R_{x,\nu}(y) L_{x+1,\nu}(y).
\end{split}\end{equation}
Here, $L_{x,\nu}(y)$ and $R_{x,\nu}(y)$ are left and right mover chiral Majorana fermions along the $y$-direction at the wire $x \in \Z + 1/2$ with $\nu$ the flavor index. 
The second term is the inter-wire hopping term to make the system gapped. 
Since the $(p_x+ip_y)$ state of $E^{\infty}_{2,-2}$ is made from gluing local $(p_x+ip_y)$ states at the 2-cell and its inversion image, the site index $x$ should be in odd integers. 
The inversion is defined by $I (L_{x,\nu}(y),R_{x,\nu}(y)) I^{-1} = (R_{-x,\nu}(-y),-L_{-x,\nu}(-y))$. 
We use the following facts to calculate the group extension: 
\begin{itemize}
\item
For two pairs of left- and right-mover chiral Majorana fermions $\{R_1(y),L_1(y),R_2(y),L_2(y)\}$, there is a 1-parameter family of gapped Hamiltonian $H(\theta)$ which switches the hopping terms 
\begin{equation}\begin{split}
H(\theta) 
&= \sum_{a=1,2} \int d y \left\{ L_{a}(y) i \p_y L_{a}(y) - R_{a}(y) i \p_y R_{a}(y) \right\} \\
&\qquad + \cos \theta \int dy it \left\{ R_1(y) L_1(y) + R_2(y) L_2(y) \right\} \\
&\qquad + \sin \theta \int dy it \left\{ R_1(y) L_2(y) - R_2(y) L_1(y) \right\}. 
\end{split}
\label{eq:1d_chiral_majorana_switch}
\end{equation}
\item 
For the $(p_x \pm i p_y)$ state defined on a cylinder $S^1 \times \R$ with the $\pi$-flux piercing $S^1$, the dimensional reduction along $S^1$ gives the nontrivial Kitaev chain. 
\end{itemize}
The latter is due to the existence of a Majorana zero mode localized at the $\pi$-flux defect in the $(p_x \pm i p_y)$ state. 
\begin{figure}[!]
\includegraphics[width=\linewidth, trim=0cm 0cm 0cm 0cm]{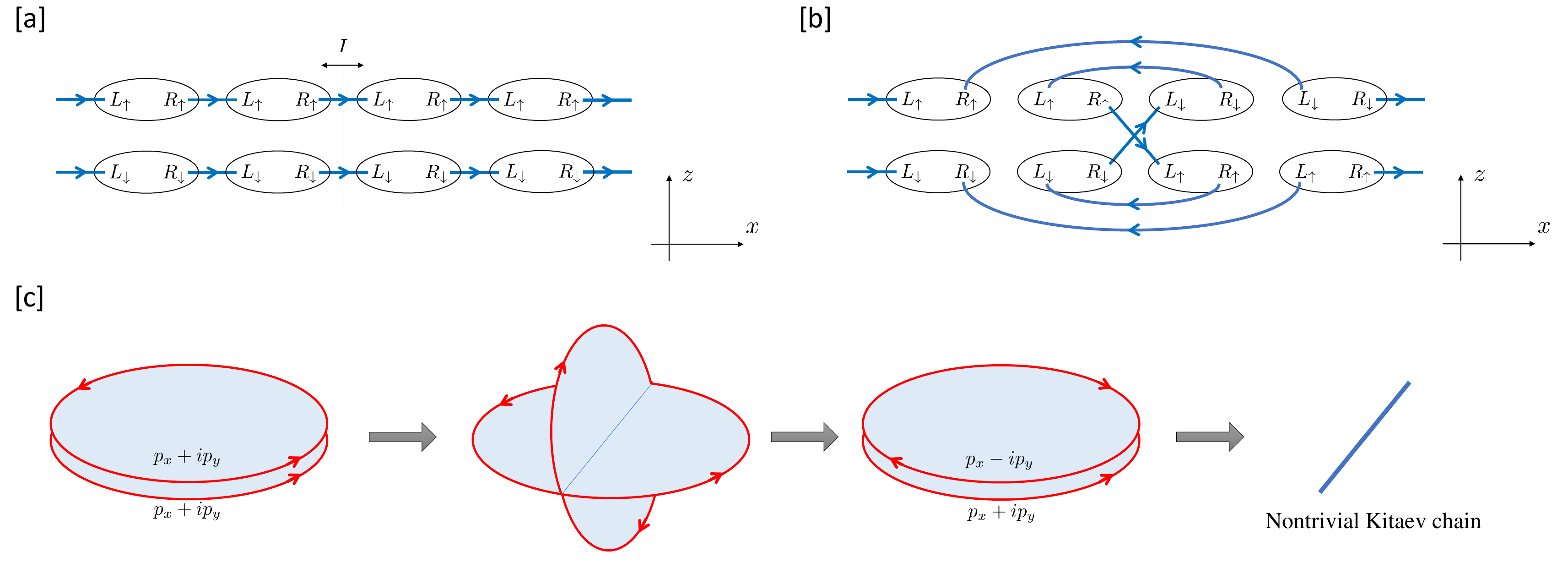}
\caption{
[a] The coupled wire Hamiltonian $H_\ua+H_\da$. 
[b] A coupled wire Hamiltonian equivalent to [a].
[c] An alternative process to the nontrivial Kitaev chain from the double layered $(p_x+i p_y)$ states. 
First we rotate a $(p_x+i p_y)$ state by $\pi$-angle to get the nonchiral $(p_x+i p_y) \oplus (p_x-i p_y)$ state. 
Next, we induce an inversion symmetric mass term, resulting in the nontrivial Kitaev chain. 
}
\label{fig:3d_fermion_inv_extension}
\end{figure}
Now let us consider the two-layered $(p_x+i p_y)$ states $H_{\ua}+H_{\da}$ (see Fig.~\ref{fig:3d_fermion_inv_extension}[a]). 
Applying the adiabatic deformation (\ref{eq:1d_chiral_majorana_switch}) to two quartets of chiral Majoranas 
$\{ R_{-\frac{3}{2},\ua}, L_{-\frac{1}{2},\ua}, R_{\frac{1}{2},\da}, L_{\frac{3}{2},\da}\}$ and 
$\{ R_{-\frac{3}{2},\da}, L_{-\frac{1}{2},\da}, R_{\frac{1}{2},\ua}, L_{\frac{3}{2},\ua}\}$ with preserving the inversion symmetry, we find that the Hamiltonian $H_{\ua} + H_{\da}$ is decomposed into the three layers of $(p_x+i p_y)$ states (see Fig.~\ref{fig:3d_fermion_inv_extension}[b]). 
Moving the inversion symmetric pair of $(p_x+i p_y)$ states to the infinite $z = \pm \infty$, the Hamiltonian $H_{\ua}+H_{\da}$ is found to be equivalent to the $(p_x+ip_y)$ state localized at the inversion center 
\begin{equation}\begin{split}
\wt H 
= ({\rm kinetic\ terms})
 + it \int d y \left\{ 
R_{-\frac{1}{2},\ua} L_{\frac{1}{2},\ua}
+R_{\frac{1}{2},\ua} L_{-\frac{1}{2},\da}
+R_{-\frac{1}{2},\da} L_{\frac{1}{2},\da}
+R_{\frac{1}{2},\da} L_{-\frac{1}{2},\ua}
\right\}
\end{split}\end{equation}
Here, the relative sign of the coefficients in between $it R_{\frac{1}{2},\ua} L_{-\frac{1}{2},\da}$ and $it R_{\frac{1}{2},\da} L_{-\frac{1}{2},\ua}$ is fixed to be 1 from the inversion symmetry. 
The boundary condition of the Hamiltonian $\wt H$ is periodic, meaning $\wt H$ is equivalent to the nontrivial Kitaev chain along the $y$-direction, the generator of $F^1 h_0 = \Z_8$. 
Therefore, the extension (\ref{eq:3d_fermion_inv_extension}) is nontrivial and fixed as $h^{\Z_4}_0(\R^3,\p \R^3) = \Z_{16}$, which is consistent with the pin$_+$ cobordism group $\Omega^{{\rm pin}_+}_4 = \Z_{16}$~\cite{KTTW15}. 

The same conclusion can be derived by using the continuous Dirac Hamiltonian. 
Here we show this in a slightly different (but equivalent) manner. 
The single layer of the $(p_x+i p_y)$ state is described by the BdG Hamiltonian 
\begin{align}
{\cal H}
= - i \p_x \tau_x -i \p_y \tau_y + m \tau_z, \qquad 
C=\tau_x K, \qquad 
I = i \tau_z, 
\end{align}
with $m<0$, where $\tau_{\mu \in \{x,y,z\}}$ is the Pauli matrix for the Nambu space. 
Because the generalized homology $h^{\Z_4}_0(\R^3,\p \R^3)$ classifies SPT phases over the disc $D^3 \sim \R^3$ which may be anomalous on the boundary $\p D^3 \sim \p \R^3$ of the disc, one can rotate the $(p_x+i p_y)$ layer around the inversion center without breaking the inversion symmetry~\cite{FangMapping}. 
The $\pi$ rotation around the $y$-axis makes the Hamiltonian ${\cal H}$ a $(p_x-i p_y)$ state ${\cal H}'$, 
\begin{align}
{\cal H} \sim {\cal H}'
= i \p_x \tau_x -i \p_y \tau_y + m \tau_z, \qquad 
C=\tau_x K, \qquad 
I = i \tau_z. 
\end{align}
Then, the double layer of $(p_x+i p_y)$ states ${\cal H} \oplus {\cal H}$ is equivalent to the non-chiral state 
\begin{align}
{\cal H} \oplus {\cal H}'
= - i \p_x \tau_x \sigma_z -i \p_y \tau_y + m \tau_z, \qquad 
C=\tau_x K, \qquad 
I = i \tau_z, 
\end{align}
where $\sigma_{\mu}$ is the Pauli matrix for the layer indices. 
See Fig.~\ref{fig:3d_fermion_inv_extension}[c]. 
This Hamiltonian admits an inversion symmetric mass term 
\begin{align}
M(x) \tau_x \sigma_y, \qquad M(-x) = - M(x), 
\end{align}
which induces the localized doublet modes at $x \sim 0$  
\begin{align}
\{\phi_{{\rm loc},1}(x,y),\phi_{{\rm loc},2}(x,y)\}
\sim
\left\{ 
\begin{pmatrix}
1 \\
0 \\
\end{pmatrix}_{\tau} , 
\begin{pmatrix}
0 \\
1 \\
\end{pmatrix}_{\tau} \right\} \otimes \begin{pmatrix}
1\\
1\\
\end{pmatrix}_{\sigma}
e^{- \int^x dx' M(x')}. 
\label{eq:3d_fermion_inversion_odd_extension_Hamiltonian_doublet}
\end{align}
in the case where $M(x)>0$ for $x>0$. 
The effective low-energy Hamiltonian within the doublet reads 
\begin{align}
{\cal H}_{\rm loc}
= -i \p_y \tau_y + m \tau_z, \qquad 
C=\tau_x K, \qquad 
I=i \tau_z. 
\label{eq:3d_fermion_inversion_odd_extension_Hamiltonian}
\end{align}
This is the nontrivial Kitaev Hamiltonian, the generator of $F^1 h_0 = \Z_8$.

\subsection{Interacting fermions with inversion symmetry: the case of $I^2 = 1$}
Next, we consider the inversion symmetric fermionic systems again, but the inversion square is $I^2 = 1$, the identity operator. 
Because the fermion permutation operator $U_{12}$, $U_{12} f_1 U_{12}^{-1} = f_2, U_{12} f_2 U_{12}^{-1} = - f_1$, satisfies $U_{12}^2 = (-1)^F$, the algebraic relation $I^2 = 1$ implies a $\pi$-flux line ending at the inversion center. 
Put differently, $I^2 = 1$ enforces the $\pi$-flux. 
In the following, we see that the higher differentials $d^r$ in the AHSS reflect the constraint from $I^2=1$. 

At the inversion center, the inversion symmetry behaves as $\Z_2$ onsite symmetry. 
The classification of SPT$^d$ phases is given by the Anderson dual to the spin cobordism, which is noncanonically isomorphic to ${\rm Tor} \Omega_{d+1}^{\rm Spin}(B\Z_2) \times {\rm Free} \Omega_{d+2}^{\rm Spin}(B\Z_2)$, where the free part represents the Chern-Simons term from the theta term in $(d+2)$-spacetime dimensions~\cite{KTTW15,FH16}. 

In this section, we also present the classification of anomalies $h^{\Z_2}_{-1}(\R^d,\p \R^d)$ for $d=1,2,3$. 

For the same manner at the beginning of Sec.~\ref{sec:5.1}, the correspondence between the spatial inversion $I$ satisfying $I^2=1$ and the four bordism groups described above is given as follows for spatial $d$ dimensions: In the case of $d \equiv 1 \mod 4$, $R^2=1$, i.e., ${\rm pin}^+$ bordism. In the case of $d \equiv 2 \mod 4$, $U^2=(-1)^F$, i.e., $({\rm Spin} \times \Z_4)/\Z_2$ bordism. In the case of $d \equiv 3 \mod 4$, $R^2=(-1)^F$, i.e., ${\rm pin}^-$ bordism. In the case of $d \equiv 0 \mod 4$, $U^2=1$, i.e., ${\rm Spin} \times \Z_2$ bordism.

\subsubsection{1d fermions}
\label{sec:1d fermions_inv_2}
Let us consider $1d$ fermions over the infinite line $\R$. 
We shall compute the homology $h^{\Z_2}_n(\R,\p \R)$ where $\Z_2$ acts on $\R$ by the inversion $x \mapsto -x$.
The $E^1$-page is given by 
\begin{align}
\begin{array}{c|ccccccc}
q=0 & \Z_2 \times \Z_2 & \Z_2 \\
q=1 & \Z_2 \times \Z_2 & \Z_2 \\
q=2 & \Z \times \Z_8 & \Z \\
\hline 
E^1_{p,-q} & p=0 & p=1 \\
\end{array}
\end{align}
Here, $E^1_{0,0} = \Z_2 \times \Z_2$ is generated by the occupied state $f^{\dag}_{\pm} \ket{0}$ of a complex fermion with inversion parity $I f^{\dag}_{\pm} I^{-1} = \pm f^{\dag}_{\pm }$. 
Similarly, $E^1_{0,-1}$ is generated by a Majorana femrion $\gamma_{\pm}$ with the inversion parity $I \gamma_{\pm} I^{-1} = \pm \gamma_{\pm}$. 
$E^1_{0,-2} = \Z \times \Z_8$ is generated, in the viewpoint of anomaly, a right-mover chiral Majorana fermion $R_+(y)$ with the trivial $\Z_2$ onsite symmetry action for the subgroup $\Z$, and a pair $(R_+(y),L_-(y))$ of right- and left-mover chiral Majorana fermions with the even and odd parities under the $\Z_2$ onsite symmetry, respectively, for $\Z_8$. 

The first differential $d^1_{1,0}$ is defined as the trivialization of SPT$^0$ phases at the 0-cell from 1-cells. See Fig.~(\ref{fig:1d_inv_d110}). 
Unlike the case of $I^2 = (-1)^F$, the inversion symmetric pair has the odd inversion parity $I f^{\dag}_1 f^{\dag}_3 \ket{0} = - f^{\dag}_1 f^{\dag}_3$ because $I f^{\dag}_1 I^{-1} = f^{\dag}_3$ and $I f^{\dag}_3 I^{-1} = f^{\dag}_1$, which means $d^1_{1,0} = (1,1)$. 
In the same way, we find that $d^1_{1,-1}=(1,1)$. 
The first differential $d^1_{1,-2}$ is computed as follows~\cite{Hermele_torsor}: 
An inversion symmetric pair $(R^1(y),R^3(y))$, $I R^1(y) I^{-1} = R^3(y)$ and $I R^3(y) I^{-1} = R^1(y)$, of right-mover chiral Majorana fermions has the inversion parity $I R_{\pm }(y) I^{-1} = \pm R_{\pm}(y)$ where $R_{\pm}(y) = (R^1(y)\pm R^3(y))/\sqrt{2}$ is the linear combination. 
Therefore, by adding the trivial state $R_-(y) \oplus L_-(y)$, the inversion symmetric pair is equivalent to $(R_+(y),R_-(y)) \sim (R_-(y)) \oplus (R_-(y)) \oplus (R_+(y),L_-(y))$, the anomalous state $(-2,1)$ in $\Z \times \Z_8$, which means $d^1_{1,-2} = (-2,1)$. 
The homology of $d^1$ gives us the $E^2$-page
\begin{align}
\begin{array}{c|ccccccc}
q=0 & \Z_2 & 0 \\
q=1 & \Z_2 & 0 \\
q=2 & \Z_{16} & 0 \\
\hline 
E^2_{p,-q} & p=0 & p=1 \\
\end{array}
\end{align}
With this, we find that SPT phases are classified by $\Z_2$, and it is generated by a complex fermion $f^{\dag}\ket{0}$ at the inversion center. 
We also conclude that the $1d$ anomalies with inversion symmetry are classified by $\Z_2$, and it is generated by a single Majorana fermion at the inversion center.

It is noteworthy that $E^2_{0,-2} = \Z_{16}$ classifies the three-dimensional SPT phase with reflection symmetry, where the square of the reflection is $R^2=1$. This is known to be a $\Z_{16}$ classification~\cite{HSHH17}.

\subsubsection{2d fermions}
\label{sec:2d_fermion_inversion_i^2=1}
Let us consider $2d$ fermions with inversion symmetry $I: (x,y) \mapsto (-x,-y)$ with $I^2 = 1$. 
This corresponds to even (odd) parity $2d$ superconductors in spinless (spinful) fermions. 
The $E^1$-page is 
\begin{align}
\begin{array}{c|ccccccc}
q=0 & \Z_2 \times \Z_2 & \Z_2 & \Z_2 \\
q=1 & \Z_2 \times \Z_2 & \Z_2 & \Z_2 \\
q=2 & \Z \times \Z_8 & \Z & \Z \\
q=3 & 0 & 0 & 0 \\
\hline 
E^1_{p,-q} & p=0 & p=1 & p=2 \\
\end{array}
\end{align}
By the same discussion as in the previous section~\ref{sec:1d fermions_inv_2}, the $E^2$-page is computed as 
\begin{align}
\begin{array}{c|ccccccc}
q=0 & \Z_2 & 0 & \Z_2 \\
q=1 & \Z_2 & 0 & \Z_2 \\
q=2 & \Z_{16} & 0 & \Z \\
q=3 & 0 & 0 & 0 \\
\hline 
E^2_{p,-q} & p=0 & p=1 & p=2 \\
\end{array}
\end{align}
In this table, the second differentials $d^2_{2,-1}$ and $d^2_{2,-2}$ can be nontrivial. 

The second differential $d^2_{2,-1}: E^2_{2,-1} \to E^2_{0,0}$ represents how an adiabatically generated nontrivial Kitaev chain in the 2-cell $\alpha$ trivializes the $\Z_2$ fermion parity at the inversion center with preserving the inversion symmetry. 
Put differently, $d^2_{2,1}$ measures the obstruction whether the Kitaev chain enclosing the inversion center collapses or not. 
See the following figure:
$$
\includegraphics[width=0.7\linewidth, trim=0cm 9cm 0cm 0cm]{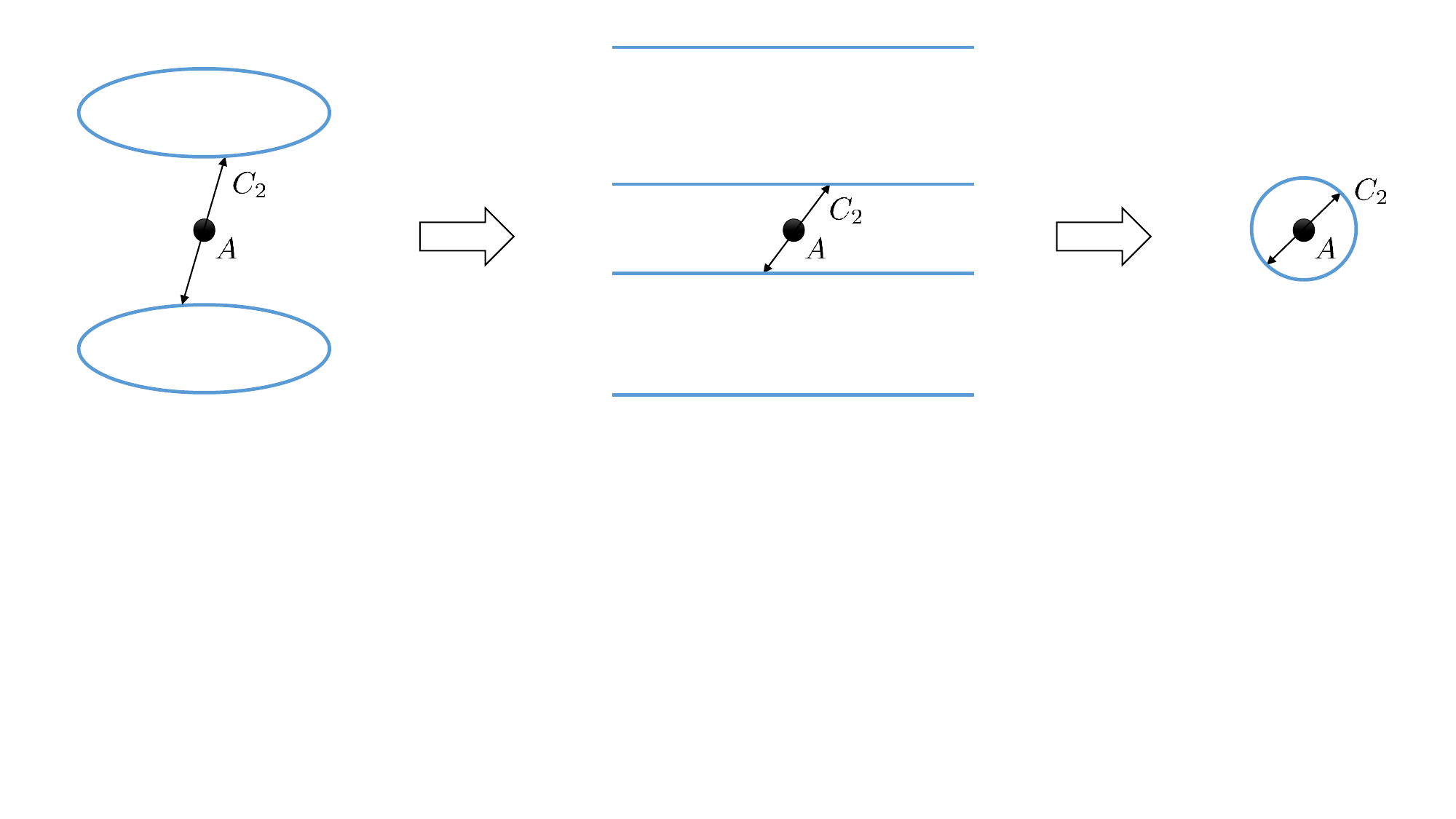}
$$
In the above figure, the blue lines represent the nontrivial Kitaev chains. 
In the middle of the figure, the top and bottom Kitaev chains go to infinite.
Thanks to $d^1_{2,-1}=0$, it is guaranteed that the Kitaev chains in the upper and lower half-planes related by the inversion symmetry can glue together at $1$-cells.
However, it is nontrivial if the Kitaev chain on the circle enclosing the $0$-cell (the right in the above figure) collapses in unique and gapped ground states.
We find that $d^2_{2,-1}$ is actually nontrivial. 
Because the square of the inversion is identified with the $2\pi$-rotation, the Kitaev chain enclosing the $0$-cell A obeys the periodic boundary condition, i.e.\ the $\pi$-flux inside the chain. 
It is well-known that the topologically nontrivial Kitaev chain with the periodic boundary condition has the odd fermion parity $\braket{GS_{\rm PBC}|(-1)^F|GS_{\rm PBC}} = -1$ for the ground state $\ket{GS_{\rm PBC}}$~\cite{KitaevUnpaired}, which implies the nontrivial second differential $d^2_{2,-1}: 1 \mapsto 1$.
Another viewpoint is that the nontrivial Kitaev chain with the inversion symmetry with $I^2=1$ must accompany a $\pi$-flux defect inside the closed chain and the $\pi$-flux defect behaves as the obstruction to collapsing. 

The nontriviality of the second differential $d^2_{2,-2}: E^2_{2,-2} \to E^2_{0,-1}$ is found in a similar way to $d^2_{2,-1}$.
It represents how an $(p_x+ip_y)$-wave superconductor generated in the $2$-cell $\alpha$ without chaining the anomaly of the system trivializes the $\Z_2$ Majorana fermion at the 0-cell $A$ in the presence of the inversion symmetry. 
See the following figure:
$$
\includegraphics[width=0.7\linewidth, trim=0cm 9cm 0cm 0cm]{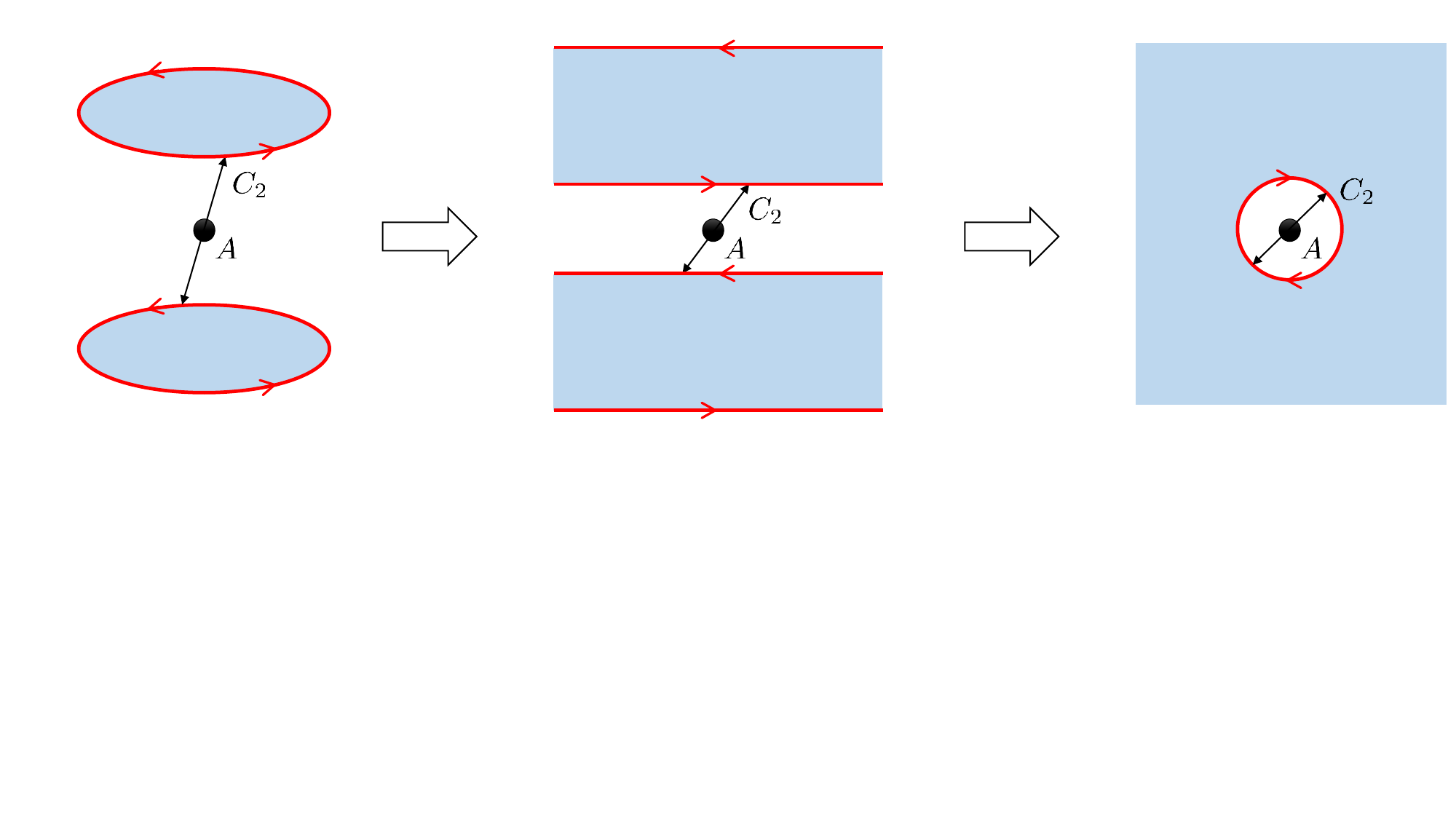}
$$
In the above figure, the blue regions represent the $(p_x+i p_y)$ superconductors preserving the inversion symmetry. 
Thanks to $d^1_{2,-2}=0$, the upper and lower $(p_x+i p_y)$ superconductors can glue together at 1-cells. 
However, it is nontrivial if they glue together at the $0$-cell, the inversion center. 
As for the Kitaev chain enclosing the $0$-cell, the $(p_x+i p_y)$ superconductor with $I^2=1$ must have a $\pi$-flux defect piercing the inversion center to make the boundary condition of the chiral Majorana fermion periodic (the Ramond sector). 
It is well-known that there is an exact zero Majorana mode localized at the $\pi$-flux reflecting the bulk topology, which implies that the second differential is nontrivial $d^2_{2,-2}: 1 \mapsto 1$. 

We get the $E^3 = E^{\infty}$-page 
\begin{align}
\begin{array}{c|ccccccc}
q=0 & 0 & 0 &  \\
q=1 & 0 & 0 & 0 \\
q=2 & \Z_{16} & 0 & 2\Z \\
q=3 & 0 & 0 & 0 \\
\hline 
E^3_{p,-q} & p=0 & p=1 & p=2 \\
\end{array}
\end{align}
We conclude that the classification of SPT phases is $h^{\Z_2}_0(\R^2,\p \R^2) = 2\Z$ which is generated by the $(p_x+i p_y)\oplus(p_x+i p_y)$ superconductor. 
Also, the $E^{\infty}$-page indicates the absence of nontrivial anomaly.

\subsubsection{3d fermions}
Let us consider $3d$ fermions with inversion symmetry $I: (x,y,z) \mapsto (-x,-y,-z)$ with $I^2 = 1$. 
This corresponds to even (odd) parity $3d$ superconductors in spinless (spinful) fermions. 
The $E^1$-page is 
\begin{align}
\begin{array}{c|ccccccc}
q=0 & \Z_2 \times \Z_2 & \Z_2 & \Z_2 & \Z_2\\
q=1 & \Z_2 \times \Z_2 & \Z_2 & \Z_2 & \Z_2\\
q=2 & \Z \times \Z_8 & \Z & \Z & \Z\\
q=3 & 0 & 0 & 0 & 0\\
q=4 & 0 & 0 & 0 & 0\\
\hline 
E^1_{p,-q} & p=0 & p=1 & p=2 & p=3 \\
\end{array}
\end{align}
The $E^2$-page is computed as 
\begin{align}
\begin{array}{c|ccccccc}
q=0 & \Z_2 & 0 & \Z_2 & \Z_2\\
q=1 & \Z_2 & 0 & \Z_2 & \Z_2\\
q=2 & \Z_{16} & 0 & \Z_2 & 0 \\
q=3 & 0 & 0 & 0 & 0\\
q=4 & 0 & 0 & 0 & 0\\
\hline 
E^2_{p,-q} & p=0 & p=1 & p=2 & p=3 \\
\end{array}
\end{align}
Here, $d^1_{3,-2} = 2$ is found by the same discussion in Sec.~\ref{sec:3d_fermion_inv_odd}. 
We have the $E^3$-page 
\begin{align}
\begin{array}{c|ccccccc}
q=0 & 0 & 0 & & \\
q=1 & 0 & 0 & 0 & \Z_2\\
q=2 & \Z_{16} & 0 & 0 & 0 \\
q=3 & 0 & 0 & 0 & 0\\
q=4 & 0 & 0 & 0 & 0\\
\hline 
E^3_{p,-q} & p=0 & p=1 & p=2 & p=3 \\
\end{array}
\end{align}
We conclude that both the classification of SPT phases and anomalies are trivial.

\subsection{$2d$ fermion with $U(1)$ and $n$-fold rotation symmetry}
In this section, we shall discuss the classification of $2d$ SPT phases of fermions with $U(1)$ charge conservation and $n$-fold rotation symmetries. 
There are four cases of $n$-fold rotation: (i) $C_n$ rotation preserving the $U(1)$ charge, (ii) magnetic $C_n T$ rotation where $T$ is the time-reversal transformation, (iii) $C_n C$ rotation where $C$ is the particle-hole transformation, and (iv) $C_n CT$ rotation where $CT$ is an anti-unitary PHS. 
The $n$-fold rotation symmetric cell decomposition of infinite $2d$ space is shown in Fig.~\ref{fig:cn}. 
For $C_n$ rotation, we give the complete classification in Sec.~\ref{sec:2d_fermion_u1_cn}. 
For other rotations, we pick up examples.

\begin{figure}[!]
\includegraphics[width=\linewidth, trim=0cm 7cm 0cm 1cm]{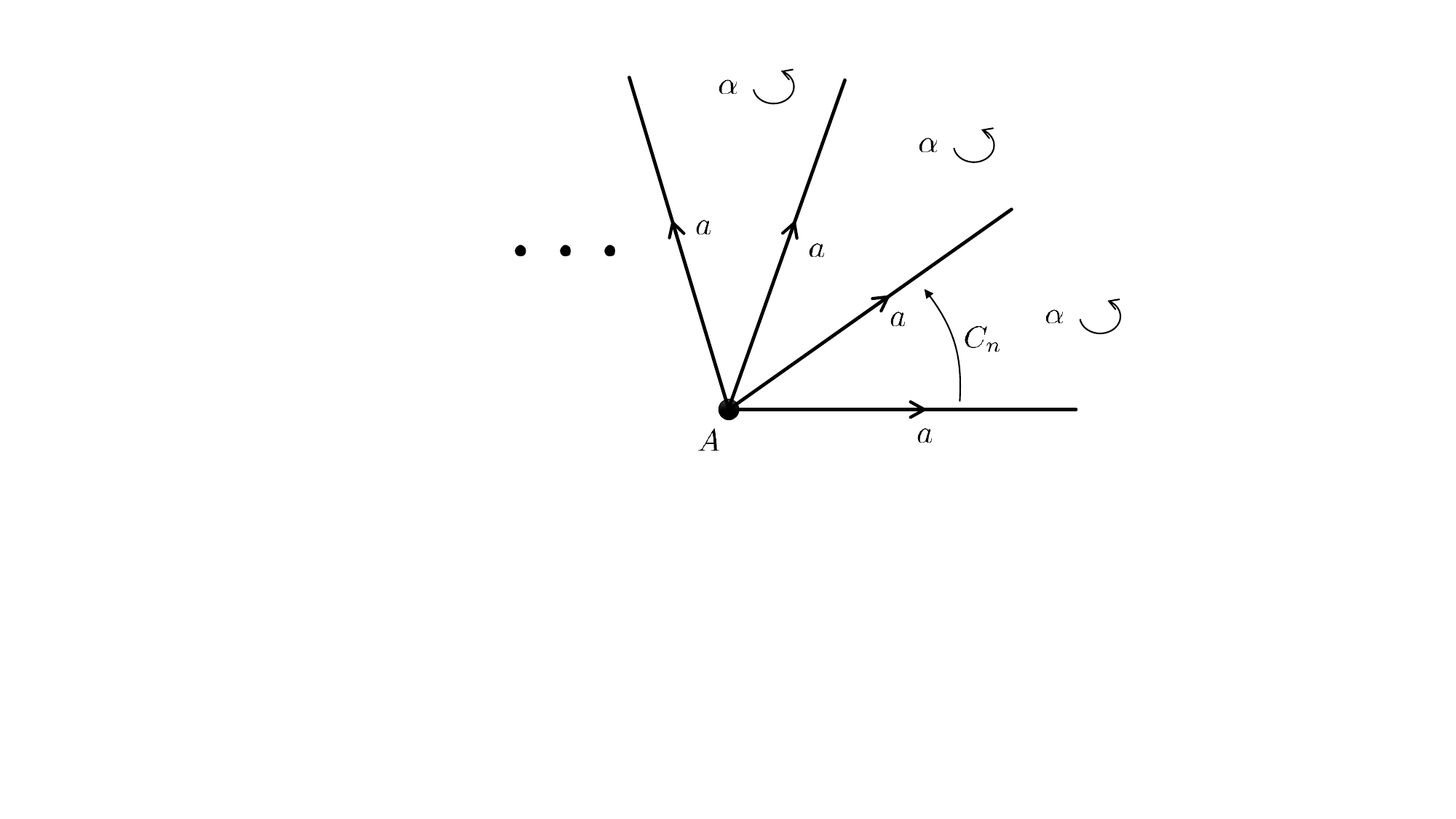}
\caption{
The cell decomposition of $2d$ space with the $n$-fold rotation symmetry.
}
\label{fig:cn}
\end{figure}

\subsubsection{$C_n$ rotation symmetry}
\label{sec:2d_fermion_u1_cn}
It was discussed that the SPT phases with $n$-fold rotation symmetry is closely related to those with onsite $\Z_n$-symmetry, where the classification is given by the spin$^c$ cobordism ${\rm Tor} \Omega^{\rm spin^c}_3(B \Z_n) \times {\rm Free} \Omega^{\rm spin^c}_4(B \Z_n)$~\cite{Hermele_torsor,KSPoint,TE18}. 
Thanks to the $U(1)$ symmetry, without loss of generality, one can assume the $n$-fold rotation symmetry is normalized as $(C_n)^n = 1$. 
The $E^1$-page relevant to SPT phases is given by 
\begin{align}
\begin{array}{c|ccccccc}
q=0 & \Z \times \Z_n & \Z & \\
q=1 & 0 & 0 & 0 \\
q=2 &  & \Z^{\times 2} & \Z^{\times 2}  \\
\hline 
E^1_{p,-q} & p=0 & p=1 & p=2 \\
\end{array}
\end{align}
Here, $E^1_{0,0} = \Z \times \Z_n$ represents an integer-valued $U(1)$ charge and a $\Z_n$ eigenvalue of the $C_n$-rotation, and $E^1_{1,-2} = E^1_{2,-2} = \Z^{\times 2}$ represents Chern insulators and Bosonic integer quantum Hall states. 
There is an even-odd effect in the first differential $d^1_{1,0}: \Z \to \Z \times \Z_n$. 
$d^1_{1,0}$ represents how adiabatically created complex fermions $f_1^{\dag}, \dots, f^{\dag}_n$ on $1$-cells trivialize SPT phases at the rotation center, where the complex fermions are related by the $C_n$ rotation as $C_n f^{\dag}_j C_n^{-1} = f^{\dag}_{j+1}$ with $f^{\dag}_{n+1} = f^{\dag}_{1}$. 
The $n$-plet $\Psi:= f^{\dag}_1 \dots f^{\dag}_n$ has the $U(1)$ charge $n \in \Z$, and the $\Z_n$ charge $C_n \Psi C_n^{-1} = - \Psi$ for even $n$ and $C_n \Psi C_n^{-1} = \Psi$ for odd $n$, which means 
\begin{align}
d^1_{1,0}: 1 \mapsto 
\left\{\begin{array}{ll}
(n,n/2) & \mbox{for even $n$}\\
(n,0) & \mbox{for odd $n$}\\
\end{array}\right. .
\end{align}
Also, we have $d^1_{2,-2} = 0$ because the $C_n$-rotation does not change the charge and thermal Hall conductivities. 
The homology of $d^1$ gives the $E^2$-page 
\begin{align}
&\begin{array}{c|ccccccc}
q=0 & \Z_{2n} \times \Z_{\frac{n}{2}} & 0 & \\
q=1 & 0 & 0 & 0 \\
q=2 &  & & \Z^{\times 2}  \\
\hline 
E^2_{p,-q} & p=0 & p=1 & p=2 \\
\end{array} & \mbox{for even $n$}, 
\end{align}
\begin{align}
&\begin{array}{c|ccccccc}
q=0 & \Z_{n} \times \Z_{n} & 0 & \\
q=1 & 0 & 0 & 0 \\
q=2 &  & & \Z^{\times 2}  \\
\hline 
E^2_{p,-q} & p=0 & p=1 & p=2 \\
\end{array} & \mbox{for odd $n$}, 
\end{align}
where $\Z_{2n} \times \Z_{\frac{n}{2}}$ ($\Z_n \times \Z_n$) is generated by SPT phases $(1,0), (2,1) \in E^1_{0,0}$ ($(1,0), (0,1) \in E^1_{0,0}$) at the rotation center when $n$ is even (odd). 
(The $(0,1)$ state is understood as the particle and hole excitation $f^{\dag} f'$ with $C_n f^{\dag} C_n^{-1} = e^{\frac{2 \pi i}{n}} f^{\dag}$ and $C_n f' C_n^{-1} = f'$.)
The $E^2$-page is the limiting page for SPT phases. 
Since the exact sequences to get the homology $h^{\Z_n}_0(\R^2,\p \R^2)$ from the $E^{\infty}$-page split, we have the classification 
\begin{align}
h^{\Z_n}_0(\R^2,\p \R^2) 
= 
\left\{\begin{array}{ll}
\Z^{\times 2} \times \Z_{2n} \times \Z_{\frac{n}{2}} & \mbox{for even $n$}, \\
\Z^{\times 2} \times \Z_{n} \times \Z_{n} & \mbox{for odd $n$}. \\
\end{array}\right.
\end{align}
This agrees with the spin$^c$ cobordism group $\Omega^{{\rm spin}^c}_*(B \Z_n)$ for onsite $\Z_n$ symmetry~\cite{Gilkey, KSPoint}.

\subsubsection{$C_2 T$ rotation symmetry}
In this section, we give the AHSS for magnetic rotation symmetry $C_2 T$. 
There are two cases: (i) $(C_2 T)^2 = 1$ and (ii) $(C_2 T)^2 = (-1)^F$. 
Due to the additional $U(1)$ phase factor $(C_2)^2 = (-1)^F$ for fermionic systems, in the TQFT limit, the former case would be equivalent to class AII TRS $T^2 = (-1)^F$, while the latter case would be the class AI TRS $T^2 = 1$. 
At the rotation center, the $C_2 T$ rotation behaves as the onsite TRS. 
In this section, we also address the classification of anomalies on $2d$ systems, the surface of $3d$ SPT phases with $C_2 T$ rotation symmetry.

\medskip
\noindent
{\it The case of $(C_2T)^2=1$---} 
In this case, the $E^1$-page relevant to SPT phases and anomalies is given by 
\begin{align}
&\begin{array}{c|ccccccc}
q=0 & \Z & \Z & \Z \\
q=1 & \Z_2 & 0 & 0 \\
q=2 & 0 & \Z^{\times 2} & \Z^{\times 2}  \\
q=3 & 0 & 0 & 0 \\
\hline 
E^1_{p,-q} & p=0 & p=1 & p=2 \\
\end{array}
\end{align}
We have used the classification results in, for example, \cite{FH16,GPW18} for the $E^1$-page at the rotation center. 
$E^1_{0,-1} = \Z_2$ is generated by the $1d$ bosonic SPT phase with TRS, the Haldane chain.
The first differential $d^1_{1,0}: \Z \to \Z$ is given by $d^1_{1,0}: 1 \to 2$, since the doublet $f_1^{\dag} f^{\dag}_2$ of $C_2 T$ invariant complex fermions on 1-cells has the fermion number $2$. 
Also, $d^1_{2,-2}: \Z^{\times 2} \to \Z^{\times 2}$ is found as $d^1_{2,-2}:(n,m) \to (2n,2m)$, since the charge and thermal Hall conductivities change to those inverse under the time-reversal rotation $C_2 T$. 
As a result, we have the $E^2$-page 
\begin{align}
&\begin{array}{c|ccccccc}
q=0 & \Z_2 & 0 & \Z \\
q=1 & \Z_2 & 0 & 0 \\
q=2 & 0 & (\Z_2)^{\times 2} & 0 \\
q=3 & 0 & 0 & 0 \\
\hline 
E^2_{p,-q} & p=0 & p=1 & p=2 \\
\end{array}
\end{align}
We find that $2d$ SPT phases are classified by $h^{\Z_2}_0(\R^2,\p \R^2) = \Z_2$, which is consistent with that for onsite class AII TRS.
Also, the classification of anomalies $h^{\Z_2}_{-1}(\R^2,\p \R^2)$ fits into the short exact sequence 
\begin{align}
0\to \Z_2 \to h^{\Z_2}_{-1}(\R^2,\p \R^2) \to (\Z_2)^{\times 2} \to 0. 
\label{eq:2d_fermion_u1_c2t_case1_extension}
\end{align}
This is also consistent with the $(\Z_2)^{\times 3}$ classification of $3d$ SPT phases of class AII insulators~\cite{WPS14,FH16}. 

\medskip
\noindent
{\it The case of $(C_2T)^2=(-1)^F$---} 
In this case, the $E^1$-page relevant to SPT phases anomalies is given by 
\begin{align}
&\begin{array}{c|ccccccc}
q=0 & \Z & \Z & \Z \\
q=1 & 0 & 0 & 0 \\
q=2 & \Z_2 & \Z^{\times 2} & \Z^{\times 2}  \\
q=3 & (\Z_2)^{\times 3} & 0 & 0 \\
\hline 
E^1_{p,-q} & p=0 & p=1 & p=2 \\
\end{array}
\end{align}
At the rotation center, the $C_2 T$ symmetry becomes the onsite class AII TRS. 
$E^1_{0,-2} = \Z_2$ is generated by the quantum spin Hall insulator. 
The first differential $d^1_{1,0}: \Z \to \Z$ is given by the identity map $d^1_{1,0}: 1 \to 1$, since $E^1_{0,0} = \Z$ is generated by the Kramers doublet of two complex fermions and it can go away to the infinity without breaking the $C_2 T$ symmetry. 
Also, $d^1_{2,-2}: \Z^{\times 2} \to \Z^{\times 2}$ is the same as for $(C_2 T)^2 = 1$.
The first differential $d^1_{1,-2}: \Z^{\times 2} \to \Z_2$ is nontrivial: a $C_2 T$ symmetric pair of chiral edge states of Chern insulators form the quantum spin Hall state, implying $d^1_{1,-2}: (n,m) \to n$ (mod $2$). 
We have the $E^2$-page 
\begin{align}
&\begin{array}{c|ccccccc}
q=0 & 0 & 0 & \Z \\
q=1 & 0 & 0 & 0 \\
q=2 & 0 & \Z_2 & 0 \\
q=3 & 0 & 0 & 0 \\
\hline 
E^2_{p,-q} & p=0 & p=1 & p=2 \\
\end{array}
\end{align}
The absence of nontrivial SPT phases $h^{\Z_4}_0(\R^2,\p \R^2) = 0$ is consistent with that for $2d$ class AI insulators. 
On the one hand, unexpectedly, we have a nontrivial anomaly $h^{\Z_4}_{-1}(\R^2,\p \R^2) = \Z_2$, which can be compared with $2d$ class AI anomalies where the classification is trivial.

\subsection{$3d$ fermions with $U(1)$ and inversion symmetry}
Let us consider $3d$ fermions with $U(1)$ charge conservation and inversion symmetry which preserves the $U(1)$ charge. 
The inversion symmetric cell decomposition of the infinite $3d$ space $\R^3$ is shown in Fig.~\ref{fig:cell_inversion} [c]. 
We focus on the SPT phases. 
The $E^1$-page is given by 
\begin{align}
\begin{array}{l|ccccccc}
q=0 & \Z \times \Z_2 & \Z & \Z & \Z \\
q=1 & 0 & 0 & 0 & 0 \\
q=2 &  & \Z \times \Z & \Z \times \Z & \Z \times \Z \\
q=3 & & 0 & 0 & 0 \\
\hline 
E^1_{p,-q} & p=0 & p=1 & p=2 & p=3 \\
\end{array}
\end{align}
Here we left some terms blank since those terms do not contribute to SPT phases. 
$E^1_{p \in \{1,2,3\},1} = \Z \times \Z$ is generated by the Chern insulator ($\sigma_{xy}=1,\kappa_{xy}=1$) and bosonic integer quantum Hall state ($\sigma_{xy}=8$, $\kappa_{xy}=0$). 

The first differential $d^1_{1,0}$ represents a pair creation of complex fermions with a unit charge $e$ and $-e$ at $1$-cell $a$ and moving the fermion with charge $e$ to the inversion center and that with $-e$ to the infinite (See Fig.~(\ref{fig:1d_inv_d110})).
One can fix the inversion to be $I^2=1$. 
Then, the inversion symmetric pair of complex fermions has the odd inversion parity $I f^{\dag}_1 f^{\dag}_3 I^{-1}= - f^{\dag}_1 f^{\dag}_3$. 
We find that $d^1_{1,3}: 1 \to (2,1) \in \Z \times \Z_2$. 
The first differential $d^1_{2,1}$ is trivial, since the boundary anomalies of SPT$^2$ phases in the 2-cells canceled out (See Fig.~(\ref{fig:2d_inv_chern})).
The first differential $d^1_{3,1}$ is nontrivial: 
$E^1_{3,1} = \Z \times \Z$ means that a $\Z \times \Z$ SPT$^2$ phase on a 2-sphere is created adiabatically in the 3-cell $V$. 
The north and the south SPT phases contribute to SPT$^2$ phases on the 2-cell $\alpha$ with an equal weight, since the inversion transformation does not change the Hall conductivity. 
We have $d^1_{3,1}: (n,m) \mapsto (2n,2m)$. 

Taking the homology of the first differentials, we get the $E^2$-page 
\begin{align}
\begin{array}{c|ccccccc}
q=0 & \Z_4 & 0 & &  \\
q=1 & 0 & 0 & 0 & 0 \\
q=2 &  & & \Z_2 \times \Z_2 &0 \\
q=3 & & 0 & 0 & 0 \\
\hline 
E^2_{p,-q} & p=0 & p=1 & p=2 & p=3 \\
\end{array}
\end{align}
Here, the generator of $E^2_{0,3} = \Z_4$ is $[(1,0)]$ with $(1,0) \in \Z \times \Z_2$, that is, a complex fermion with a unit charge. 
This is the limiting page for the homology $h^{\Z_2}_0(\R^3,\p \R^3)$. 
We find that the classification of the 2nd-order SPT phases is $\Z_2 \times \Z_2$, and it is generated by the Chern insulator and Bosonic integer quantum Hall state on the inversion symmetric plane. 

The classification of SPT phases fits into the exact sequence 
\begin{align}
0 \to \Z_4 \to h_0^{\Z_2}(\R^3,\p \R^3) \to \Z_2 \times \Z_2 \to 0. 
\end{align}
The group extension is determined in a way similar to Sec.~\ref{sec:3d_fermion_inv_odd}. 
Using the coupled wire construction of the Chern insulator, one can show that the two layers of Chern insulators are equivalent to the Chern insulator compactified on $S^1 \times \R$ with $\pi$-flux (the periodic boundary condition) piercing $S^1$. 
Because of the absence of edge anomaly in SPT$^1$ phases of complex fermions, one can cut the Chern insulator at $y>0$ and $y<0$. 
The resulting Chern insulator is defined on the $2$-sphere with the periodic boundary condition along $S^1$, meaning an odd monopole charge inside $S^2$. 
Due to the quantum Hall effect, the Chern insulator defined on a closed manifold with a magnetic monopole $m_g$ has the $U(1)$ charge $m_g \times ch$, where $ch$ is the Chern number. 
This implies that the two layers of Chern insulators are eventually equivalent to the generator of $E^{\infty}_{0,0} = \Z_4$. 

The above observation is also verified by the continuous Dirac Hamiltonian with a texture of the mass term. 
Taking the $\pi$-rotation to a single layer of the Chern insulator as in Fig.~\ref{fig:3d_fermion_inv_extension}[c], 
we find that the double layer of inversion symmetric Chern insulators ${\cal H} \oplus {\cal H}$ is equivalent to a nonchiral state 
\begin{align}
{\cal H} \oplus {\cal H}'
= - i \p_x \tau_x \sigma_z -i \p_y \tau_y + m \tau_z, \qquad 
I = \tau_z, 
\end{align}
where $\sigma_{\mu}$ is the Pauli matrix for the layer indices. 
In the same way as Sec.~\ref{sec:3d_fermion_inv_odd}, by adding a mass term $M(x) \tau_x \sigma_y$ with $M(-x) = - M(x)$, the effective low-energy Hamiltonian reads 
\begin{align}
{\cal H}_{\rm loc}
= - i \p_y \tau_y + m \tau_z, \qquad 
I = \tau_z.
\end{align}
In the absence of particle-hole symmetry, we further have a mass term varying in the $y$-direction
\begin{align}
M'(y) \tau_x, \qquad M'(-y) = - M'(y), 
\end{align}
which yields a single mode $\wt \phi_{\rm loc}(x,y)$ localized at the inversion center. 
The low-energy effective Hamiltonian for this single mode is $\wt{{\cal H}} = m$ with $I = 1$. 
Therefore, for $m<0$ the ground state is the occupied state $\wt f^{\dag}_{\rm loc} \ket{0}$ of the localized mode $\wt \phi_{\rm loc}(x,y)$ with the even parity $I = 1$, which is the generator of $E^{\infty}_{0,0} = \Z_4$. 

Also, since there is no localized mode at the inversion center for the double-layer bosonic integer quantum Hall states, the extension of the bosonic integer quantum Hall state by $E^{\infty}_{0,0}$ is trivial. 
We conclude that $h^{\Z_2}_0(\R^3) \cong \Z_8 \times \Z_2$, which is consistent with the pin$^c$ cobordism group $\Omega^{\rm Pin^c}_4=\Z_8 \times \Z_2$.

\subsection{1d boson with $\Z_2$ onsite and $\Z_2^R$ reflection symmetry}
In this section, we present an example of the AHSS for a bosonic system. 
Let us consider the $1d$ bosonic system with $\Z_2$ onsite and $\Z_2^R$ reflection symmetry $x \mapsto -x$. 
We denote the generator of $\Z_2$ by $\sigma$ and that of $\Z_2^R$ by $r$. 
We assume $\sigma r = r \sigma$, i.e.\ the total symmetry group is $\Z_2 \times \Z_2^R$. 
Some basic calculation techniques in the AHSS for bosonic systems can be seen in this example. 
We will compare the symmetry class in the present section with a slightly different symmetry class in the next section, which highlights some features of the AHSS and the LSM-type theorems in bosonic systems. 

The reflection symmetric decomposition of the $1d$ space $\R$ was shown in Fig.~\ref{fig:cell_inversion} [a]. 
The little groups $G^p$ of $p$-cells ($p=0,1$) are $G^0 = \Z_2 \times \Z_2^R$ and $G^1 = \Z_2$. 
The $E^1$-page is 
\begin{align}
\begin{array}{c|ccccccc}
q=0 & \Z_2 \times \Z_2 & \Z_2 \\
q=1 & \Z_2 & 0 \\
\hline 
E^1_{p,-q} & p=0 & p=1 \\
\end{array}
\end{align}
Here, $E^1_{p \in \{0,1\},0} \cong H^1_{\rm group}(G^p,U(1))$ is generated by 1-dimensional irreps.\ at the $p$-cell with the little group $G^p$, and $E^1_{p \in \{0,1\},-1} \cong H^2_{\rm group}(G^p,U(1))$ is generated by nontrivial projective representations at $p$-cells. 

The first differential $d^1_{1,0}$ is determined from the induced representation explained as follows. 
Let $\ket{\epsilon}, \epsilon \in \{1,-1\},$ be the basis of the 1-dimensional irrep.\ at the right $1$-cell $a$ so that $\hat \sigma \ket{\epsilon} = \epsilon \ket{\epsilon}$. 
We formally introduce $\hat r \ket{\epsilon}$ as the basis of 1-dimensional irrep.\ for the left $1$-cell. 
We ask what the 1-dimensional irrep.\ of the tensor product $\ket{\epsilon} \otimes \hat r \ket{\epsilon}$ as a 1-dimensional irrep.\ of $\Z_2 \times \Z_2^R$ is.
Noting that 
\begin{align}
&(\hat \sigma \otimes \hat \sigma) (\ket{\epsilon} \otimes \hat r \ket{\epsilon}) = \ket{\epsilon} \otimes \hat r \ket{\epsilon}, \\
&(\hat r \otimes \hat r) (\ket{\epsilon} \otimes \hat r \ket{\epsilon}) = \hat r \ket{\epsilon} \otimes \ket{\epsilon} = \ket{\epsilon} \otimes \hat r \ket{\epsilon}, 
\end{align}
we find that the tensor product irrep.\ is the trivial irrep.\ of $\Z_2 \times \Z_2^R$. 
This means $d^1_{1,0} = 0$. 
Therefore, $E^1$-page is the limiting page. 

We find SPT phases are classified by $h^{\Z_2 \times \Z_2^R}_0(\R,\p \R) \cong E^2_{0,0} = \Z_2 \times \Z_2$ and they are generated by 1-dimensional irreps.\ at the reflection center. 
Also, $E^2_{0,-1} = \Z_2$ means the presence of the LSM theorem as a boundary of an SPT phase in the sense of Sec.~\ref{sec:lsm_boundary}. 
This means the Hilbert space composed of a nontrivial projective representation at the reflection center must not have a unique gapped symmetric ground state. 
To make the point clear, we consider the LSM theorem with the translation symmetry in the following. 

\subsubsection{With translation symmetry}

\begin{figure}[!]
\begin{center}
\includegraphics[width=0.8\linewidth, trim=0cm 15cm 0cm 0cm]{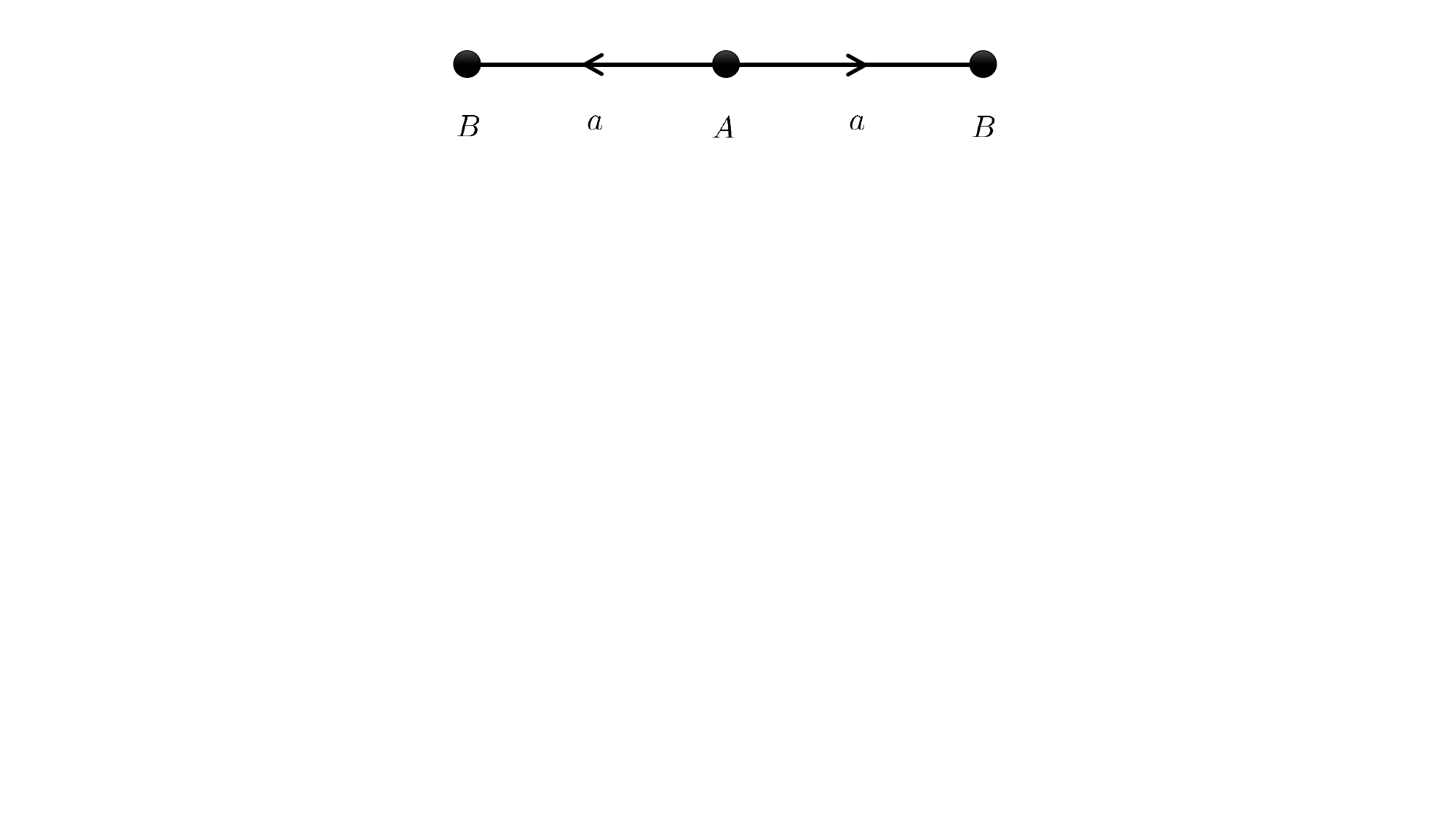}
\end{center}
\caption{
The cell-decomposition of $1d$ real space with translation and reflection symmetry. 
The figure shows a unit cell. 
}
\label{fig:1d_tra_ref}
\end{figure}

In addition to the $\Z_2 \times \Z_2^R$ symmetry, we add the translation symmetry $\Z$, $\Z \ni n : x \mapsto x+n$. 
The total symmetry becomes $(\Z \rtimes \Z_2^R) \times \Z_2$ where $r \in \Z_2^R$ acts on $\Z$ as reflection. 
The reflection and translation symmetric cell decomposition is shown in Fig.~\ref{fig:1d_tra_ref}, it is composed of $0$-cells (reflection centers) $\{A,B\}$ and a $1$-cell $\{a\}$. 
The $E^1$-page is given by 
\begin{align}
\begin{array}{c|ccccccc}
q=0 & \Z_2^{\times 2} \times \Z_2^{\times 2} & \Z_2 \\
q=1 & \Z_2 \times \Z_2 & 0 \\
\hline 
E^1_{p,-q} & p=0 & p=1 \\
\end{array}
\end{align}
and this becomes the $E^2$-page because $d^1_{1,0}=0$. 

Let us focus on $E^2_{0,-1} = \Z_2 \times \Z_2$, which indicates the LSM theorem classified by $\Z_2 \times \Z_2$. 
The former (latter) $\Z_2 \times \Z_2$ is generated by a nontrivial projective irrep.\ at the refection center $A$ ($B$).
It should be noticed that $(1,1) \in \Z_2 \times \Z_2$ Hilbert space, which is composed of a projective irrep.\ at $A$ and also at $B$, remains anomalous even if the Hilbert space per unit cell has two projective irreps.\ (i.e.\ a linear representation per a unit cell), which is consistent with literature~\cite{PWJZ17, HSHH17}.

\subsection{1d boson with $\Z_2$ onsite and $\Z_4^R$ reflection symmetry}
\label{sec:1d_boson_z2_z4}
In this section, we consider a slightly different symmetry from that in the previous section. 
Let us consider $1d$ bosonic systems with $\Z_2$ onsite and $\Z_4^R$ reflection symmetry, in which the square $r^2$ of reflection $r \in \Z_4^R$ is another onsite $\Z_2$ symmetry. 
We assume $\sigma r = r \sigma$, i.e.\ the total symmetry group is $\Z_2 \times \Z_4^R$. 
The little groups are $G^0 = \Z_2 \times \Z_4^R$ for the 0-cell and $G^1 = \Z_2 \times \Z_2^{R^2}$ for the 1-cell. 
The $E^1$-page becomes 
\begin{align}
\begin{array}{c|ccccccc}
q=0 & \Z_2 \times \Z_4 & \Z_2 \times \Z_2  \\
q=1 & \Z_2 & \Z_2 \\
\hline 
E^1_{p,-q} & p=0 & p=1 \\
\end{array}
\end{align}

The first differential $d^1_{1,0}: E^1_{1,0} \to E^1_{0,0}$ is computed as follows. 
Let $\{ \ket{\epsilon_1,\epsilon_2} \}_{\epsilon_1,\epsilon_2 \in \{1,-1\}}$ be the basis of the 1-dimensional irrep.\ with $\hat \sigma \ket{\epsilon_1,\epsilon_2} = \epsilon_1 \ket{\epsilon_1,\epsilon_2}$ and $\widehat{r^2} \ket{\epsilon_1,\epsilon_2} = \epsilon_2 \ket{\epsilon_1,\epsilon_2}$ for the right 1-cell $a$. 
We formally introduce $\hat r \ket{\epsilon_1,\epsilon_2}$ as the basis of 1-dimensional irrep.\ for the left $1$-cell. 
We ask what the 1-dimensional irrep.\ of the tensor product $\ket{\epsilon_1,\epsilon_2} \otimes \hat r \ket{\epsilon_1,\epsilon_2}$ is.
Noting that 
\begin{align}
&(\hat \sigma \otimes \hat \sigma) (\ket{\epsilon_1,\epsilon_2} \otimes \hat r \ket{\epsilon_1,\epsilon_2}) = \ket{\epsilon_1,\epsilon_2} \otimes \hat r \ket{\epsilon_1,\epsilon_2}, \\
&(\hat r \otimes \hat r) (\ket{\epsilon_1,\epsilon_2} \otimes \hat r \ket{\epsilon_1,\epsilon_2}) = \epsilon_2 (\ket{\epsilon_1,\epsilon_2} \otimes \hat r \ket{\epsilon_1,\epsilon_2}), 
\end{align}
we find that the tensor product irrep.\ has a nontrivial $\Z_4^R$ phase $-1$ when $\epsilon_2=-1$. 
This means $d^1_{1,0}:\Z_2 \times \Z_2^{R^2} \to \Z_2 \times \Z_4, (n,m) \mapsto (0,2m)$. 

The calculation of $d^1_{1,-1}$ is similar to $d^1_{1,0}$. 
Let $\ket{j}_{j \in \{\ua,\da\}}$ be a nontrivial projective irrep.\ of $\Z_2 \times \Z_2^{R^2}$ at the right $1$-cell so that 
\begin{align}
\hat \sigma \ket{j} = \ket{i} [D_{\sigma}]^i_j, \qquad 
\widehat{r^2} \ket{j} = \ket{i} [D_{r^2}]^i_j, \qquad 
D_{\sigma} D_{r^2} = - D_{r^2} D_{\sigma}.
\end{align}
We formally introduce the basis $\hat r \ket{j}$ for the left $1$-cell which is the induced representation so as to satisfy 
\begin{align}
\hat \sigma (\hat r\ket{j}) = (\hat r \ket{j}) \alpha [D_{\sigma}]^i_j, \qquad 
\widehat{r^2} (\hat r\ket{j}) = (\hat r \ket{j}) \beta [D_{r^2}]^i_j, 
\end{align}
where $\alpha$ and $\beta$ are unfixed $U(1)$ phases. 
We ask the factor system of the tensor product representation $\ket{j_1} \otimes \hat r \ket{j_2}$ is nontrivial or not. 
The representation matrices for the tensor product representation read as 
\begin{align}
&(\hat \sigma \otimes \hat \sigma) 
(\ket{j_1} \otimes \hat r \ket{j_2}) 
= (\ket{i_1} \otimes \hat r \ket{i_2}) [D_{\sigma \otimes \sigma}]^{i_1i_2}_{j_1j_2}, \qquad 
[D_{\sigma \otimes \sigma}]^{i_1i_2}_{j_1j_2} = \alpha [D_{\sigma}]^{i_1}_{j_1} [D_{\sigma}]^{i_2}_{j_2}, \\
&(\hat r \otimes \hat r) 
(\ket{j_1} \otimes \hat r \ket{j_2}) 
= (\ket{i_1} \otimes \hat r \ket{i_2}) [D_{r \otimes r}]^{i_1i_2}_{j_1j_2}, \qquad 
[D_{r \otimes r}]^{i_1i_2}_{j_1j_2} = \beta [D_{r^2}]^{i_1}_{j_2} \delta^{i_2}_{j_1}, 
\end{align}
from which we find that $\sigma$ anticommutes with $r$, $D_{\sigma \otimes \sigma} D_{r \otimes r} = - D_{r \otimes r} D_{\sigma \otimes \sigma}$. 
This means the tensor product representation belongs to a nontrivial projective representation of $\Z_2 \times \Z_4^R$, and so $d^1_{1,-1}: \Z_2 \to \Z_2, n \mapsto n$. 
Also, from $\tr[D_{r^2 \otimes r^2}]=\beta^2 \tr[D_{r\otimes r}]^2 = 0$, we find that the tensor product representation is the direct sum of the two projective irreps.\ of $\Z_2 \times \Z_4$.
We also notice that, as a projective representation of $G^1 = \Z_2 \times \Z_2^{R^2}$, the tensor product state is linear.

The homology of $d^1$ gives the $E^2$-page 
\begin{align}
\begin{array}{c|ccccccc}
q=0 & \Z_2 \times \Z_2 & \Z_2 \\
q=1 & 0 & 0 \\
\hline 
E^2_{p,-q} & p=0 & p=1 \\
\end{array}
\end{align}
From this table, we find that SPT phases are classified by $h^{\Z_2 \times \Z_4^R}_0 \cong E^2_{0,0} = \Z_2 \times \Z_2$ and they are generated by 1-dimensional irreps.\ at the reflection center. 
$E^{2}_{0,-1} = 0$ implies the absence of the LSM theorem as a boundary of SPT phases in the sense of Sec.~\ref{sec:lsm_boundary}. 

We observed that the first differential $d^1_{1,-1}$ is nontrivial. 
In the basis of terminology in Sec.~\ref{sec:lsm_spt}, this means the LSM theorem that enforces that a state is an SPT phase if the Hilbert space is composed of nontrivial projective representations of the onsite symmetry $\Z_2 \times \Z_4^R$ at the reflection center. 
To make the point clear, we consider this LSM theorem in the presence of translation symmetry. 

\subsubsection{With translation symmetry}

In addition to the $\Z_2 \times \Z_4^R$ symmetry, we add the translation symmetry $\Z$. 
The total symmetry becomes $(\Z \rtimes \Z_4^R) \times \Z_2$ where $r \in \Z_4^R$ acts on $\Z$ as reflection. 
The reflection and translation symmetric cell decomposition is shown in Fig.~\ref{fig:1d_tra_ref}, and it is composed of $0$-cells $\{A,B\}$ and a $1$-cell $\{a\}$. 
The $E^1$- and $E^2$-pages are given as 
\begin{align}
\begin{array}{c|ccccccc}
q=0 & (\Z_2 \times \Z_4) \times (\Z_2 \times \Z_4) & \Z_2 \times \Z_2  \\
q=1 & \Z_2 \times \Z_2 & \Z_2 \\
\hline 
E^1_{p,-q} & p=0 & p=1 \\
\end{array}, \qquad 
\end{align}
\begin{align}
\begin{array}{c|ccccccc}
q=0 & \Z_2^{\times 3} \times \Z_4 & \Z_2  \\
q=1 & \Z_2 & 0 \\
\hline 
E^2_{p,-q} & p=0 & p=1 \\
\end{array}
\end{align}
Here, the first differentials are $d^1_{1,0}: (n,m) \mapsto (0,2m,0,2m)$ and $d^1_{1,-1}: n \mapsto (n,n)$. 
Here, $E^2_{1,0} = \Z_2$ is understood as the usual LSM theorem to forbid a system having a unique gapped ground state if the Hilbert space contains a nontrivial projective representation per unit cell. 

Let us consider the physical consequence of the nontrivial first differential $d^1_{1,-1}$. 
As we introduced in Sec.~\ref{sec:lsm_spt}, $\im d^1_{1,-1} \neq 0$ yields the LSM theorem enforcing a nontrivial SPT phase. 
The Hilbert space ${\cal H}$ belonging to $(1,1) \in \im d^1_{1,-1} \subset E^1_{0,-1}$ is composed of a nontrivial projective irrep.\ per a reflection center $A$ and $B$. 
$\im d^1_{1,-1} \neq 0$ means that if we have a unique symmetric state $\ket{\psi}$ in the Hilbert space ${\cal H}$, $\ket{\psi}$ should be the Haldane chain state, i.e.\ $\ket{\psi}$ shows a projective representation of the onsite $\Z_2 \times \Z_2^{R^2}$ symmetry at the edge. 
Let us demonstrate it. 
As a projective representation of $\Z_2 \times \Z_4^R$ at $A$ and $B$, we consider two Ising spins $\hat{\bm{\sigma}}_j(j=1,2)$ obeying $\hat \sigma \hat{\bm{\sigma}}_j \hat\sigma^{-1} = \sigma_j^z \hat{\bm{\sigma}}_j \sigma_j^z$, $\hat r \hat{\bm{\sigma}}_1 \hat r^{-1} = \hat{\bm{\sigma}}_2$, and $\hat r \hat{\bm{\sigma}}_2 \hat r^{-1} = \sigma^x_1 \hat{\bm{\sigma}}_1 \sigma^x_1$ under $\Z_2 \times \Z_4^{R^2}$ for each $A$ and $B$. 
We can not make a unique state within the Hilbert space of two Ising spins due to a nontrivial factor system, however, we may be able to create singlet bonds in between sites $A$ and $B$ without breaking the symmetry to form a unique state. 
Let us denote the Ising spin at the inversion center by $\hat{\bm{\sigma}}_j(n)$ with $n \in \Z,\Z+1/2$. 
The full symmetry acts on Ising spins as 
$\hat T \hat{\bm{\sigma}}_j(n) \hat T^{-1} = \hat{\bm{\sigma}}_j(n+1)$, 
$\hat \sigma \hat{\bm{\sigma}}_j(n) \hat \sigma^{-1} = \sigma_j^z(n) \hat{\bm{\sigma}}_j(n) \sigma_j^z(n)$, 
$\hat r \hat{\bm{\sigma}}_1(n) \hat r^{-1} = \hat{\bm{\sigma}}_2(-n)$, and 
$\hat r \hat{\bm{\sigma}}_2(n) \hat r^{-1} = \sigma_1^x(-n) \hat{\bm{\sigma}}_1(-n) \sigma_1^x(-n)$, where $\hat T$ is the lattice translation. 
We have a tensor product of singlet bonds 
\begin{align}
\ket{\psi}
= \cdots 
\big(\ket{\ua}_{2,n-\frac{1}{2}}\ket{\da}_{2,n}-\ket{\da}_{2,n-\frac{1}{2}}\ket{\ua}_{2,n}\big)
\big(\ket{\ua}_{1,n}\ket{\da}_{1,n+\frac{1}{2}}-\ket{\da}_{1,n}\ket{\ua}_{1,n+\frac{1}{2}}\big)
\cdots 
\end{align}
as a symmetric unique state. 
This is equivalent to the Haldane state protected by onsite $\Z_2 \times \Z_2^{R^2}$ symmetry.

\subsection{2d boson with $\Z_4$ two-fold rotation symmetry}
\label{sec:2d boson with Z_4 two-fold rotation symmetry}
In this section, we see an example of the nontrivial group extension from $E^{\infty}$-page in a bosonic system. 
Let us consider the $2d$ bosonic system with $C_2$-rotation symmetry of which the square is an onsite symmetry $C_2^2 = U$. 
The total symmetry is $\Z_4^{C_2}$ and it is orientation-preserving, hence the classification of SPT phases is expected to be $H^3(B\Z_4;U(1)) \cong \Z_4$. 

The $C_2$-rotation symmetric cell-decomposition was shown in Fig.~\ref{fig:cell_inversion} [b].
The $E^1$-page is 
\begin{align}
\begin{array}{c|ccccccc}
q=0 & \Z_4 & \Z_2 & \Z_2 \\
q=1 & 0 & 0 & 0 \\
q=2 & \Z_4 & \Z_2 & \Z_2 \\
\hline 
E^1_{p,-q} & p=0 & p=1 & p=2 \\
\end{array} 
\end{align}
Let us focus on the $E^2$-page relevant to the classification of  SPT phases. 
In the same way as Sec.~\ref{sec:1d_boson_z2_z4}, we find $d^1_{1,0}:1 \to 2$. 
Also, $d^1_{2,-2}=0$ because the little groups of 1- and 2-cells are the same. 
The $E^2$-page is 
\begin{align}
\begin{array}{c|ccccccc}
q=0 & \Z_2 & 0 & \Z_2 \\
q=1 & 0 & 0 & 0 \\
q=2 & & & \Z_2 \\
\hline 
E^2_{p,-q} & p=0 & p=1 & p=2 \\
\end{array} 
\end{align}
To determine the classification of SPT phases $h_0^{\Z_4^{C_2}}(\R^2,\p \R^2)$, we should solve the extension problem 
\begin{align}
0 \to \underbrace{\Z_2}_{E^2_{0,0}} \to h_0^{\Z_4^{C_2}}(\R^2,\p \R^2) \to \underbrace{\Z_2}_{E^2_{2,-2}} \to 0, 
\label{eq:2d_z4_rot_extension}
\end{align}
that is, we ask whether or not the double layer of the Levin-Gu $\Z_2$ SPT phase with the $\pi$-flux piercing the rotation center (this is enforced by the symmetry algebra $C_2^2=U$) is equivalent to a $0d$ SPT phase (i.e.\ a linear representation) with a $U(1)$ phase $i$ or $-i$ under the $\Z_4^{C_2}$ rotation. 
Because the double stack of $\Z_2$ Levin-Gu state becomes a trivial state as a $2d$ SPT phase in a generic region, the $C_2$-eigenvalue can be calculated by the partial $C_2$-rotation acting on a disc~\cite{KSPoint}.
Since the contribution from each layer is in common, the $C_2$-eigenvalue is equivalent to the $2\pi$-rotation on a single layer. 
Therefore, the $C_2$-eigenvalue is the same as the topological spin $e^{i \theta_a}$ of the anyon $a$ yielding the twisted boundary condition on the edge CFT, which is known to be a fourth root of unity $e^{i \theta_a} = \pm i$. 
Therefore, we conclude that the extension (\ref{eq:2d_z4_rot_extension}) is nontrivial and $h_0^{\Z_4^{C_2}}(\R^2,\p \R^2) \cong \Z_4$ as expected.

\subsection{Magnetic translation symmetry}
It is shown that the magnetic translation symmetry gives rise to various LSM-type theorems to enforce a nontrivial SPT phase solely from degrees of freedom per a unit cell~\cite{YJVR17,Lu17,LuRanOshikawa,Matsugatani}. 
We describe how such LSM-type theorems in the presence of magnetic translation symmetry are formulated in the AHSS. 

Magnetic translation symmetry is defined so that the lattice translations $T_x$ and $T_y$ are accompanied by an Aharonov-Bohm flux per a unit cell 
\begin{align}
T_y^{-1} T_x^{-1} T_y T_x = g, 
\end{align}
where $g$ is an onsite unitary symmetry with a finite order $g^n=1$. 
A useful way to visualize the $g$-flux is to introduce background $g$-symmetry open lines which start and end at the magnetic $g$-fluxes. 
In this section, we only discuss the cases of order-two magnetic flux, i.e.\ $g^2=1$. 
The AHSS discussed in this section is straightforwardly generalized to general magnetic translation symmetry.
See Fig.~\ref{fig:magnetic_tr} for an example of the configuration of $g$-symmetry lines for $g^2 = 1$. 
Fig.~\ref{fig:magnetic_tr} also shows the cell decomposition of the unit cell. 
It is composed of $0$-cell $\{A\}$, 1-cells $\{a,b\}$ and $2$-cell $\{\alpha\}$. 

\begin{figure}[!]
\begin{center}
\includegraphics[width=0.7\linewidth, trim=0cm 3cm 0cm 0cm]{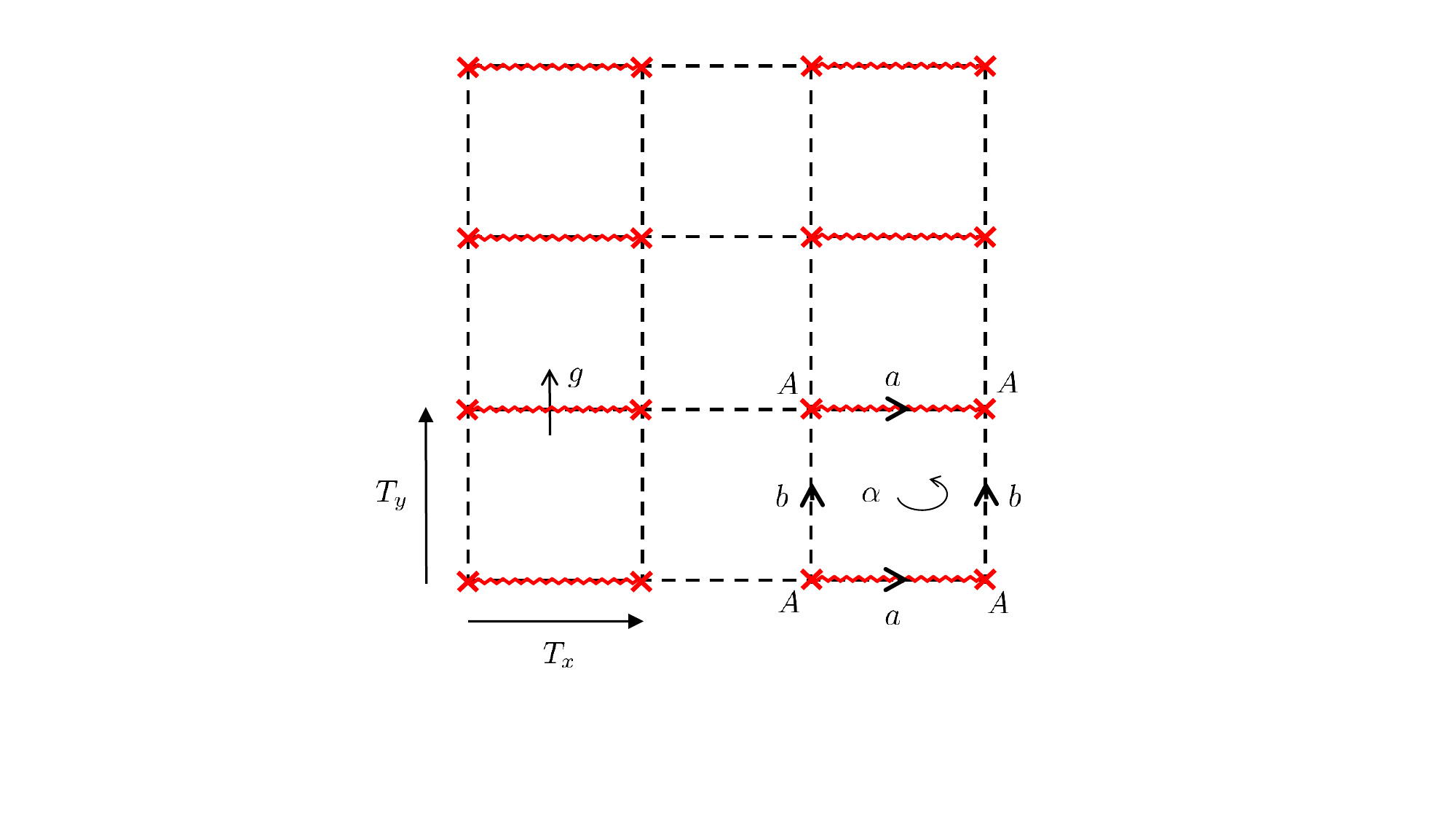}
\end{center}
\caption{Symmetry defect lines for magnetic translation symmetry with $g^2=1$.
Marks with the red cross represent magnetic $g$-fluxes.
}
\label{fig:magnetic_tr}
\end{figure}

\subsubsection{Fermion with magnetic translation symmetry}
\label{sec:Fermion with magnetic translation symmetry}
In this section, we illustrate the AHSS for $2d$ fermionic SPT phases with magnetic translation symmetry with $\pi$-flux per unit cell 
\begin{align}
T_x T_y = (-1)^F T_y T_x, 
\label{eq:2d_magnetric_tr_pi_flux}
\end{align}
where $(-1)^F$ is the fermion parity. 
For fermions, the $\pi$-flux enforces the periodic boundary condition of fermions enclosing the $\pi$-flux. 
The LSM-type theorem was discussed in Ref \cite{Lu17}. 
The $E^1$-page is given by 
\begin{align}
\begin{array}{c|ccccccc}
q=0 & \Z_2 & \Z_2 \times \Z_2 & \Z_2 \\
q=1 & \Z_2 & \Z_2 \times \Z_2 & \Z_2 \\
q=2 & \Z & \Z \times \Z & \Z \\
\hline 
E^1_{p,-q} & p=0 & p=1 & p=2 \\
\end{array} 
\end{align}
Because the first differential does not see the $\pi$-flux, all the first differentials $d^1_{p,-q}$ are zero, and it holds that $E^2 = E^1$ as for usual translation symmetry without magnetic flux. 
The second differentials can be nontrivial. 
As discussed in Sec.~\ref{sec:2d_fermion_inversion_i^2=1}, in the presence of the $\pi$-flux, we have that (i) nontrivial Kitaev chain enclosing the $\pi$-flux has the fermion parity $(-1)^F = -1$, and (ii) a $(p_x +i p_y)$-superconductor traps a Majorana zero mode at the $\pi$-flux. 
These mean 
\begin{align}
&d^2_{2,-1}: \Z_2 \to \Z_2, \qquad 1 \mapsto 1, \\
&d^2_{2,-2}: \Z \to \Z_2, \qquad 1 \mapsto 1. 
\end{align}
We have the $E^3$-page 
\begin{align}
\begin{array}{c|ccccccc}
q=0 & 0 & \Z_2 \times \Z_2 & \\
q=1 & 0 & \Z_2 \times \Z_2 & 0 \\
q=2 & & \Z \times \Z & 2\Z \\
\hline 
E^3_{p,-q} & p=0 & p=1 & p=2 \\
\end{array} 
\end{align}
The classification of SPT phases is given by $h^{\Z^{\times 2}}_0(\R^2) \cong 2 \Z \times \Z_2 \times \Z_2$, where $2 \Z$ is generated by the superconductor with chiral central charge $c-\bar c = 1$. 
It should be noticed that SPT phases represented by the homology $h^{\Z^{\times 2}}_0(\R^2)$ are made in a non-anomalous Hilbert space, where odd Majorana fermions per unit cell are  forbidden. 

According to the terminology in Sec.~\ref{sec:lsm_spt}, the nontrivial second differential $d^2_{2,-2}: E^2_{2,-2} \to E^2_{0,-1}$ implies the LSM theorem to enforce the $2d$ SPT phase. 
$E^2_{0,-1} = \Z_2$ is generated by the anomalous Hilbert space composed of a Majorana fermion per unit cell. 
$d^2_{2,-2}: 1 \mapsto 1$ means that the $(p_x+i p_y)$ superconductor with $\pi$-flux per unit cell belongs to the same anomaly as $E^2_{0,-1}$, an odd number of Majorana fermions per unit cell. 
Because the $(p_x+i p_y)$ superconductor is a unique gapped state, it holds that if, under the magnetic translation symmetry, in the Hilbert space composed of odd Majorana fermions per unit cell, a unique gapped ground state should be a nontrivial SPT state with a half-integer chiral central charge $c-\bar c \in \Z+1/2$~\cite{Lu17}.

\subsubsection{Class AII insulators with magnetic translation symmetry}
Let us consider $2d$ fermionic systems with $U(1)$ symmetry and TRS $T$ with Kramers $T^2=(-1)^F$. 
Also, we assume the magnetic translation symmetry (\ref{eq:2d_magnetric_tr_pi_flux}) with $\pi$-flux per unit cell. 
The $E^1$-page relevant to SPT phases is given by 
\begin{align}
\begin{array}{c|ccccccc}
q=0 & 2\Z & 2\Z &  \\
q=1 & 0 & 0 & 0 \\
q=2 &  & \Z_2 & \Z_2 \\
\hline 
E^1_{p,-q} & p=0 & p=1 & p=2 \\
\end{array} 
\label{tab:magnetic_tr_aii}
\end{align}
Because the first differential does not see the $\pi$-flux, $d^1=0$. 
Also, the second differential $d^2$ is trivial in the table (\ref{tab:magnetic_tr_aii}), the $E^1$-page is already the $E^3$-page. 
The classification of SPT phases fits into the short exact sequence 
\begin{align}
0 \to \underbrace{2\Z}_{E^3_{0,0}} \to h^{\Z^{\times 2}}_0(\R^2) \to \underbrace{\Z_2}_{E^3_{2,-2}} \to 0, 
\label{eq:extension_2d_aii_magnetic_tr}
\end{align}
where $E^3_{0,0}$ is generated by a pair of complex fermions forming the Kramers degeneracy per a unit cell, and $E^3_{2,-2}$ is generated by the quantum spin Hall state. 
We find that the group extension (\ref{eq:extension_2d_aii_magnetic_tr}) is nontrivial. 
The key is that the $\pi$-flux in the quantum spin Hall state traps a mid-gap localized state with the fermion parity $(-1)^F = -1$~\cite{QiSpinCharge,RanSpinCharge}. 
Thanks to the magnetic translation symmetry, the bound state energies of mid-gap localized states at $\pi$-fluxes are in common, and so are localized fermion numbers. 
Therefore, the double stack of the quantum spin Hall states with magnetic $\pi$-flux is adiabatically equivalent to an atomic insulator with the Kramers degeneracy per a unit cell, that is, the generator of $E^3_{0,0}$. 
This means the group extension (\ref{eq:extension_2d_aii_magnetic_tr}) is nontrivial and the homology is given by $h^{\Z^{\times 2}}_0(\R^2) \cong \Z$. 
According to the terminology in Sec.~\ref{sec:lsm_u1}, the mismatch of the fermion number per a unit cell between the atomic insulator $E^3_{0,0} = 2 \Z$ and the generator of SPT phases $h^{\Z^{\times 2}}_0(\R^2) = \Z$ leads the LSM theorem for filling-enforced SPT phases: if a pure state has an odd fermion number per unit cell, this state should be the quantum spin Hall state~\cite{Lu17}.

\subsection{1-form $\Z_N$ and mirror symmetry in $3d$}
In this section, we present an example of the AHSS for SPT phases with $\Z_N$ 1-form symmetry~\cite{GeneGlo} in 3-space dimensions. 
In addition to the $1$-form $\Z_N$ symmetry, we consider the mirror symmetry that commutes with the $\Z_N$ 1-form charge of the line object, which is CPT dual to the TRS. 
Such symmetry is realized in the $SU(N)$ pure Yang-Mills theory. 

\begin{figure}[!]
\includegraphics[width=\linewidth, trim=0cm 13cm 0cm 1cm]{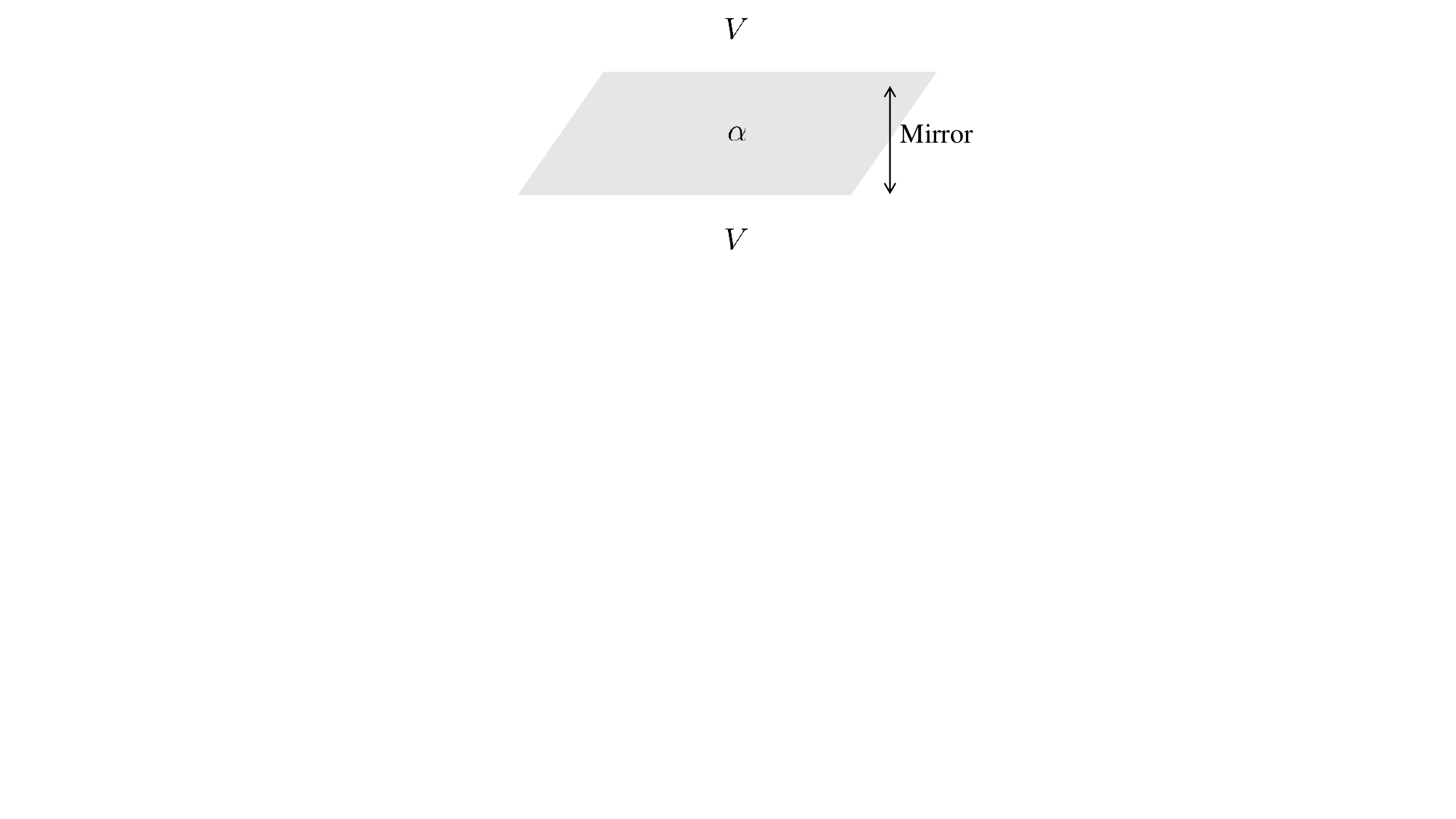}
\caption{
The mirror symmetric decomposition of $3d$ real space $\R^3$.}
\label{fig:3d_mirror}
\end{figure}

To calculate the AHSS, we introduce the mirror-symmetric decomposition of the infinite real space $\R^3$ as shown in Fig.~\ref{fig:3d_mirror} where there is no $0$- and $1$-cell.
We regard the $1$-form $\Z_N$ symmetry as onsite symmetry. 
As the classification of $d$-dimensional SPT phases on each $p$-cell, we consider the ordinary cohomology $H^{d+2}({\cal B},\Z)$ over the corresponding classifying space ${\cal B}$. 
(Adapting the ordinary cohomology as the classifications of SPT phases implies omitting the gravitational contribution to SPT phases.) 
Some low-dimensional cohomologies for $\Z_N$ 1-form and $\Z_2$ 0-form symmetries are given as follows. (See Appendix~\ref{app:k(zn,2)} for a derivation.)
\begin{align}
\begin{array}{c|ccccccccc}
&d=-2&d=-1&d=0&d=1&d=2&d=3&d=4\\
\hline
H^{d+2}(B\Z_2;\Z) & \Z & 0 & \Z_2 & 0 & \Z_2 & 0 & \Z_2 \\
\hline
H^{d+2}(K(\Z_{N\in {\rm even}},2);\Z) & \Z & 0 & 0 & \Z_N & 0 & \Z_{2N} & \Z_2 \\
\hline
H^{d+2}(K(\Z_{N\in {\rm odd}},2);\Z) & \Z & 0 & 0 & \Z_N & 0 & \Z_{N} & 0 \\
\end{array}. 
\end{align}
Here, $K(\Z_N,2)$ is the Eilenberg-MacLane space. 
$H^3(K(\Z_N,2);\Z) = \Z_N$ is generated by the charged line object with a unit $\Z_N$ 1-form symmetry charge.
$H^5(K(\Z_N,2);\Z)$ represents the 3-dimensional SPT phase with the $\Z_N$ 1-form symmetry and, at the same time, 2-dimensional anomalies generated by the $SU(N)_1$ Chern-Simons theory~\cite{GeneGlo}. 
$H^6(K(\Z_{N\in {\rm even}},2);\Z) = \Z_2$ means the existence of an unknown 4-dimensional SPT phase and $3$-dimensional anomaly from 1-form $\Z_{N \in {\rm even}}$ symmetry.
On the mirror plane, the classification of $d$-dimensional SPT phases is given by $H^{d+2}(K(\Z_N,2) \times B\Z_2;\Z)$. 

\subsubsection{Odd $N$}
Using the K\"unneth formula we have the $E^1$-page 
\begin{align}
\begin{array}{c|ccccccc}
q=0 & 0 & 0 & \Z_2 & 0 \\
q=1 & 0 & 0 & \Z_N & \Z_N \\
q=2 & 0 & 0 & \Z_2 & 0 \\
q=3 & 0 & 0 & \Z_N & \Z_N \\
q=4 & 0 & 0 & \Z_2 & 0 \\
\hline 
E^1_{p,-q} & p=0 & p=1 & p=2 & p=3 \\
\end{array}
\end{align}
Because the $SU(N)_1$ Chern-Simons theory is compatible with the mirror reflection, the first differential $d^1_{3,-3}:\Z_N \to \Z_N$ should be $d^1_{3,-3}=2$. 
Then, because $N$ is odd, $d^1_{3,-3}$ is an isomorphism. 
Similarly, we have $d^1_{3,-1}=2$. 
We have the $E^2$-page 
\begin{align}
\begin{array}{c|ccccccc}
q=0 & 0 & 0 & \Z_2 & 0 \\
q=1 & 0 & 0 & 0 & 0 \\
q=2 & 0 & 0 & \Z_2 & 0 \\
q=3 & 0 & 0 & 0 & 0 \\
q=4 & 0 & 0 & \Z_2 & 0 \\
\hline 
E^2_{p,-q} & p=0 & p=1 & p=2 & p=3 \\
\end{array}
\end{align}
We find that there is no SPT or anomalous phase protected by the $\Z_N$ 1-form symmetry in 3-space dimensions. 
The absence of anomaly for odd $N$ is consistent with the classification of the topological action~\cite{GeneGlo}.

\subsubsection{Even $N$}
The K\"unneth formula gives the $E^1$-page 
\begin{align}
\begin{array}{c|ccccccc}
q=0 & 0 & 0 & \Z_2 & 0 \\
q=1 & 0 & 0 & \Z_N & \Z_N \\
q=2 & 0 & 0 & \Z_2^{\times 2} & 0 \\
q=3 & 0 & 0 & \Z_{2N} \times \Z_2 & \Z_{2N} \\
q=4 & 0 & 0 & \Z_2^{\times 4} & \Z_2 \\
\hline 
E^1_{p,-q} & p=0 & p=1 & p=2 & p=3 \\
\end{array}
\end{align}
Here, $E^1_{2,-3}$ is generated by $SU(N)_1$ Chern-Simons theory labeled by the mirror parity $\Z_2 = \{+,-\}$.
The first differential $d^1_{3,-3}$ is found to be $d^1_{3,-3}: 1 \mapsto (2,0)$, because the quasi particle (anyon) $s_{\ua/\da}$ of the $SU(N)_1$ Chern-Simons theory living in upper/lower planes are mutually commuted, i.e.\ the mirror parity of composite $s_{\ua} s_{\da}$ is even. 
We have a part of $E^2$-page 
\begin{align}
\begin{array}{c|ccccccc}
q=0 & 0 & 0 & \Z_2 & 0 \\
q=1 & 0 & 0 & \Z_2 & 0 \\
q=2 & 0 & 0 & \Z_2^{\times 2} & 0 \\
q=3 & 0 & 0 & \Z_2 \times \Z_2 & 0 \\
q=4 & 0 & 0 &  &  \\
\hline 
E^2_{p,-q} & p=0 & p=1 & p=2 & p=3 \\
\end{array}
\end{align}
We find the existence of the second-order anomaly $E^{\infty}_{2,-3}$ for even $N$, which is consistent with the mixed {}'t Hooft anomaly among the $\Z_N$ 1-form and time-reversal symmetries in $SU(N)$ pure Yang-Mills with $\theta = \pi$~\cite{GaiottoTheta}.

\section{Free fermions}
\label{sec:free}

In this section, we study the real space homological description of SPT phases of free fermion. 
For free fermions the $\Omega$-spectrum for invertible states with onsite symmetry is the $K$-theory spectrum~\cite{Kit09}, which results in, in the presence of crystalline symmetry $G$, the $K$-homology $K^G_{n}(X,Y)$ over a pair of real space $(X,Y)$ with crystalline symmetry $G$ acting on the real space $X$ and the $K$-theory spectrum. 
There are two characteristics in free fermion SPT phases that are not in generic SPT phases in many-body Hilbert spaces.
In free fermions an integer grading $n\in\Z$ of the $K$-homology group $K^G_n(X,Y)$, which we have called the degree of SPT phenomena in Sec.~\ref{sec:genehomo}, has a different meaning:
An integer $n \in \Z$ also means the shift of the symmetry class (called the Altland-Zirnbauer (AZ) symmetry class~\cite{AZ,SRFL08}) by adding $n$ chiral symmetries (see Sec.~\ref{sec:k_homology}). 
Another feature is that the $K$-homology $K^G_n(\R^d)$ over the infinite real space $\R^d$ with a (magnetic) space group symmetry $G$ is isomorphic to a twisted equivariant $K$-cohomology group ${}^{\phi}K^{\tau-n}_P(T^d)$ over the Brillouin zone torus $T^d$, where $P$ includes point and AZ symmetry group and $\tau$ represents nonprimitive lattice translation of the space group symmetry $G$~\cite{FM13,Thiang16,Gomi17,KSGeneric,KSAHSS}. 
We see that the AHSS in real space $K$-homologies provides a different filtration of the $K$-group $K^G_n(\R^d) \cong {}^{\phi}K_P^{\tau-n}(T^d)$, which helps us to determine the $K$-group algebraically.~\footnote{
The isomorphism $K^G_n(\R^d) \cong {}^{\phi}K_P^{\tau-n}(T^d)$ between the $K$-homology group over the real space $\R^d$ and the $K$-cohomology group over the Brillouin zone torus $T^d$ in the case of complex AZ classes can be justified mathematically~\cite{GOMI201950}.
} 





\subsection{Integer grading and Altland-Zirnbauer symmetry class}
\label{sec:k_homology}

%
%

The purpose of this section is to provide the connection between the integer grading $n \in \Z$ in the $K$-homology group $K_n^G(X, Y)$ and the AZ symmetry classes of free fermions. 
For initial grading $n=0$, the $K$-homology group $K^G_0(X,Y)$ is represented by a family of the differential operator $H(x)$, which we call a Hamiltonian, depending continuously on the real space $x \in X$ and acting on the internal degrees of freedom at $x \in X$. 
For a symmetry group $G$ acting on the real space $X$, as for the case of the twisted and equivariant $K$-theory by Freed and Moore~\cite{FM13}, we introduce the two homomorphisms $\phi, c: G \to \Z_2 = \{1,-1\}$ and the factor system $\tau$ as follows. 
Let $U_g$ be the $g$-action on the internal degrees of freedom. 
For our purpose to formulate the free fermion crystalline SPT phases, one can assume that $U_g$ does not depend on real space coordinates $x\in X$. 
Per the homomorphism $\phi_g$, $U_g$ is written as 
\begin{align}
U_g = \left\{\begin{array}{ll}
{\cal U}_g & (\phi_g = 1) \\
{\cal U}_g {\cal K} & (\phi_g=-1) 
\end{array}\right.
\end{align}
with ${\cal U}_g$ a unitary matrix and ${\cal K}$ the complex conjugate.
The factor system $\tau$ determines the algebra among $g$-actions 
\begin{align}
U_g U_h = e^{i \tau_{g,h}} U_{gh}, \qquad e^{i \tau_{g,h}} \in U(1). 
\end{align} 
The factor system belongs to the two-group cocycle $\tau \in Z^2(G,U(1)_\phi)$ where $U(1)_\phi$ means the left $G$-module $U(1)$ with the $g$-action so that $g.z = z^{\phi_g}$ for $z \in U(1)$. 
The homomorphism $c_g$ specifies whether $U_g$ commutes or anticommutes with the Hamiltonian $H(x)$, i.e.\ 
\begin{align}
U_g H(x) U_g^{-1} = c_g H(g x). 
\label{eq:k_homo_0}
\end{align}
The $K$-homology group $K^G_0(X,Y)$ represents ``the topological classification of gapped Hamiltonian $H(x)$ over $X$ which can be gapless on $Y$".
\footnote{
A precise description of the $K$-homology is as follows. 
We illustrate this for the zeroth relative homology group $K_0(X,Y)$ without crystalline symmetry. 
The group $K_0(X,Y)$ is represented by a triple $({\cal H},\pi,F)$ consisting of a $\Z_2$-graded Hilbert space ${\cal H}$, $\pi: C(X) \to B({\cal H})$, and an odd self adjoint operator $F \in B({\cal H})$ satisfying that $[F,\pi(f)]$ and $\pi(f) (F^2-1)$ are compact for all $f \in C(X)$ s.t.\ $f|_Y = 0$. 
With the ``Dirac Hamiltonian" $D$, $F$ is written as $F = D (1+D^2)^{-\frac{1}{2}}$. 
For example, $K_0(\R^2,\p \R^2)$ is generated by a triple with ${\cal H} = L^2(\R^2,\C^2)$ with the $\Z_2$-grading labeled by $\sigma_z = \pm 1$, $\pi$ is the left scalar multiplication, and $F$ is given by the Dirac operator $D = -i \sigma_x \p_x -i \sigma_y \p_y$. 
Also, there is an isomorphism $KK_{-1}(\R^2/\p \R^2,(0,1)) \cong K_0(\R^2,\p \R^2)$ where the $KK_{-1}(\R^2,(0,1))$ is generated by the 1-parameter family of massive Hamiltonian $D = -i \sigma_x \p_x -i \sigma_y \p_y + m \sigma_z$ with $m \in (0,1)$, which is what we call the classification of gapped Hamiltonians. 
KS thank Yusuke Kubota for teaching me these points. 
}

As for the case of $K$-cohomology~\cite{KSGeneric,KSAHSS}, the symmetry for $n>0$ is defined by adding chiral symmetries $\Gamma_i (i = 1,\dots n)$ as in 
\begin{align}
&\Gamma_i H(x) \Gamma_i^{-1} = - H(x), \label{eq:k_homo_n} \\
&\{ \Gamma_i, \Gamma_j \} = 2 \delta_{ij}, \\
&\Gamma_i U_g = c_g U_g \Gamma_i, 
\end{align}
where $\Gamma_i$ are unitary matrices. 
We define the $n$th $K$-homology group $K^G_n(X,Y)$ as the classification of gapped Hamiltonians on $X$ with the symmetries (\ref{eq:k_homo_0}) and (\ref{eq:k_homo_n}) up to gapless states on $Y$. 
For example, when the initial symmetry class $n=0$ is composed only of the class AI TRS $T$ with $T^2 = 1$, we find that the symmetry for $n>0$ run over the real AZ symmetry classes as in the following table,  
\begin{align}
\begin{array}{lcllc}
n & \mbox{AZ class} &T&C & \mbox{$K$-spectrum} \\
\hline
n=0 & \mbox{AI} & T^2=1 & & BO \times \Z \\
n=1 & \mbox{BDI} & T^2=1 & C^2 = 1 & O \\
n=2 & \mbox{D} & & C^2=1 & O/U \\
n=3 & \mbox{BDI} & T^2=-1 & C^2 = 1 & U/Sp \\
n=4 & \mbox{AII} & T^2=-1 & & BSp \times \Z \\
n=5 & \mbox{CII} & T^2=-1 & C^2 =-1 & Sp \\
n=6 & \mbox{C} & & C^2=-1 & Sp/U \\
n=7 & \mbox{CI} & T^2=1 & C^2 = -1 & U/O \\
\end{array}
\end{align}
where the empty columns mean the absence of symmetry.

\subsection{$E^1$-page}
The $E^1$-page of the $K$-homology is straightforwardly given by counting the number of irreps.\ at $p$-cells, i.e.\ high-symmetric points in the real space. 
The definition is parallel to the case of SPT phases introduced in Sec.~\ref{sec:e1}
\begin{equation}\begin{split}
E^1_{p,n}
&:= K^G_{p+n}(X_p \cup Y, X_{p-1} \cup Y) \\
&= \prod_{j\in I_p} K^{G_{D^p_j}}_{p+n}(D^p_j,\p D^p_j) \\
&= \prod_{j\in I_p} K^{G_{D^p_j}}_{n}(pt). 
\end{split}\end{equation}
Here, $j$ runs the set (denoted by $I_p$) of inequivalent $p$-cells of $X$ that are not in $Y$, and $G_{D^p_j}$ is the little group that fixing the $p$-cell $D^p_j$.
The first differential 
\begin{align}
d^1_{p,n}: E^1_{p,n} \to E^1_{p-1,n}
\end{align}
is viewed as how free fermion SPT phases with AZ class $(p+n)$ in $(p-1)$-cells are trivialized by ones in $(p-1)$-cells. 
Unlike the AHSS for interacting SPT phases developed in Sec.~\ref{sec:d1}, there is a simple formula to compute the first differential $d^1$. 
It is found that the first differential $d^1_{p,n}$ is determined by the induced representations of the little group $G_{D^{p-1}_j}$ at a $(p-1)$-cell by the adjacent $p$-cells with the little groups $G_{D^p_{j'}} \subset G_{D^{p-1}_j}$, i.e.\ the first differential $d^1_{p,n}$ can be read off solely by the character of the little groups. 
This is analogous to the compatibility relation to compute the first differential $d_1$ in the AHSS in the $K$-cohomology~\cite{KruthoffCompatibilityRelation,po2017complete,bradlyn2017topological,KSAHSS}.

In the rest of this section, we present various examples of the real space AHSS for free fermions. 

\subsection{$1d$ Real AZ classes with mirror reflection symmetry}
Let us consider $1d$ free fermions with reflection symmetry in addition to TRS $T$ and/or PHS $C$.
We assume that the square of reflection is $R^2=-1$ and $R$ commutes with $T$ and $C$ so that the antiunitary operators $T$ and $C$ exchange the reflection eigenvalues $R =i$ and $R=-i$. 
For a complex fermion creation operator $\psi^\dag(x)$ over the $1d$ real space, the symmetries, and their inter relations are summarized as 
\begin{align}
&T \psi^{\dag}(x) T^{-1} = \psi^{\dag}(x) {\cal U}_T, \qquad T^2 \in \{1,-1\}
\end{align}
and/or
\begin{align}
&C \psi(x) C^{-1} = \psi^{\dag}(x) {\cal U}_C, \qquad C^2 \in \{1,-1\} 
\end{align}
and the reflection symmetry 
\begin{align}
&R \psi^{\dag}(x) R^{-1} = \psi^{\dag}(-x) {\cal U}_{R}, \qquad R^2 = -1,  \\
&TR=RT, \qquad CR=RC.
\end{align}
The reflection symmetric cell decomposition of the infinite $1d$ space is given as in Fig.~\ref{fig:cell_inversion}[a].
At the reflection center $A$, the reflection symmetry behaves as a $\Z_2$ onsite symmetry. 
Let us set the initial symmetry class to $T$ and $R$ with $T^2 = 1$. 
We shall compute the relative $K$-homology group $K^{\Z_2^T \times \Z_2^R}_n(\R,\p \R)$. 
The $E^1$-page and the first derivative $d^1$ are straightforwardly determined as 
\begin{align}
\begin{array}{ll|ccccc}
{\rm CI} & n=7 & 0 & 0\\
{\rm C} & n=6& \Z & 0\\
{\rm CII} & n=5& 0 & 0\\
{\rm AII} & n=4& \Z & 2\Z\\ 
{\rm DIII} & n=3 & 0 & 0\\
{\rm D} & n=2 & \Z & \Z_2 \\ 
{\rm BDI} & n=1 & 0 & \Z_2 \\ 
{\rm AI} & n=0 & \Z & \Z \\
\hline
& E^1_{p,n} & p=0 & p=1 \\
\end{array}, 
\end{align}
\begin{align}
&d^1_{1,0} = 1, \qquad d^1_{1,4} = 2.
\end{align}
The homology of $d^1$ gives the $E^2$-page 
\begin{align}
\begin{array}{ll|ccccc}
{\rm CI} & n=7 & 0 & 0\\
{\rm C} & n=6& \Z & 0\\
{\rm CII} & n=5& 0 & 0\\
{\rm AII} & n=4& \Z_2 & 0\\ 
{\rm DIII} & n=3 & 0 & 0\\
{\rm D} & n=2 & \Z & \Z_2 \\ 
{\rm BDI} & n=1 & 0 & \Z_2 \\ 
{\rm AI} & n=0 & 0 & 0 \\
\hline
& E^2_{p,n} & p=0 & p=1 \\
\end{array}
\end{align}
This is the limiting page. 
We got the classification of 1st and 2nd order SPT phases. 
Except for class D ($n=2$) the $K$-groups $K_n^{\Z_2^T \times \Z_2^R}(\R,\p \R)$ have been fixed. 

For class D, there exist two possible group extensions for the short exact sequence 
\begin{align}
0 \to \underbrace{\Z}_{E^2_{0,2}} \to K_2^{\Z_2^T \times \Z_2^R}(\R,\p \R) \to \underbrace{\Z_2}_{E^2_{1,1}} \to 0. 
\label{eq:1d_free_d}
\end{align}
The symmetry class for $n=2$ is the same one discussed in Sec.~\ref{sec:1d fermions}. 
Therefore, we conclude that the group extension (\ref{eq:1d_free_d}) is nontrivial and the $K$-group is $K^{\Z_2^T \times \Z_2^R}_2(\R,\p \R) = \Z$, which is consistent with the result by the $K$-cohomology~\cite{ChiuYaoRyu,MorimotoFurusaki,KS14}.

\subsection{$2d$ Real AZ classes with $C_2$-rotation symmetry}
\label{sec:free_2d_real_c2}
Let us consider $2d$ spinless free fermions with antiunitary symmetry of real AZ classes and $C_2$-rotation symmetry $C_2: (x,y) \mapsto (-x,-y)$ which commutes with TRS $T$ and/or PHS $C$. 
For complex fermion creation and annihilation operators, the symmetries are summarized as 
\begin{align}
&T \psi^{\dag}(x,y) T^{-1} = \psi^{\dag}(x,y) {\cal U}_T, \qquad T^2 \in \{1,-1\}
\end{align}
and/or
\begin{align}
&C \psi(x,y) C^{-1} = \psi^{\dag}(x,y) {\cal U}_C, \qquad C^2 \in \{1,-1\} 
\end{align}
and the $C_2$-rotation symmetry 
\begin{align}
&C_2 \psi^{\dag}(x,y) C_2^{-1} = \psi^{\dag}(-x,-y) {\cal U}_{C_2}, \qquad C_2^2 = 1,  \\
&TC_2=C_2T, \qquad CC_2=C_2C.
\end{align}
A $C_2$-symmetric cell decomposition of the infinite $2d$ real space is shown in Fig.~\ref{fig:cell_inversion}[b].
The $C_2$-rotation symmetry remains only at the $0$-cell $A$, the rotation center. 
The $E^1$-page and the first derivative $d^1$ are straightforwardly determined to be 
\begin{align}
\begin{array}{ll|ccccc}
{\rm CI} & n=7 & 0 & 0&0\\
{\rm C} & n=6& 0 & 0&0\\
{\rm CII} & n=5& 0 & 0&0\\
{\rm AII} & n=4& (2\Z)^{\times 2} & 2\Z & 2\Z \\ 
{\rm DIII} & n=3 & 0 & 0&0\\
{\rm D} & n=2 & (\Z_2)^{\times 2} & \Z_2 & \Z_2 \\ 
{\rm BDI} & n=1 & (\Z_2)^{\times 2} & \Z_2 & \Z_2 \\ 
{\rm AI} & n=0 & (\Z)^{\times 2} & \Z & \Z \\
\hline
& E^1_{p,n} & p=0 & p=1 & p=2 \\
\end{array}
\end{align}
\begin{align}
&d^1_{1,0} = d^1_{1,1} = d^1_{1,2} = d^1_{1,4} = (1,1), \\ 
&d^1_{2,0} = d^1_{2,1} = d^1_{2,2} = d^1_{2,4} = 0.
\end{align}
Taking the homology of $d^1$, we get the $E^2$-page 
\begin{align}
\begin{array}{ll|ccccc}
{\rm CI} & n=7 & 0 & 0&0\\
{\rm C} & n=6& 0 & 0&0\\
{\rm CII} & n=5& 0 & 0&0\\
{\rm AII} & n=4& 2\Z & 0 & 2\Z \\ 
{\rm DIII} & n=3 & 0 & 0&0\\
{\rm D} & n=2 & \Z_2 & 0 & \Z_2 \\ 
{\rm BDI} & n=1 & \Z_2 & 0 & \Z_2 \\ 
{\rm AI} & n=0 & \Z & 0 & \Z \\
\hline
& E^2_{p,n} & p=0 & p=1 & p=2 \\
\end{array}
\end{align}
Comparing this table with the known result of the $K$-cohomology group~\cite{KS14}, we find that the second differentials $d^2_{2,0}$ and $d^2_{2,1}$ must be nontrivial.
The second differentials $d^2_{2,0}$ and $d^2_{2,1}$ are the same ones as $d^2_{2,-2}$ and $d^{2,-1}$ we have computed in Sec.~\ref{sec:2d_fermion_inv_even}, where the $C_2$ rotation square is the identity. 
The homology of $d^2$ gives the $E^3 = E^{\infty}$-page 
\begin{align}
\begin{array}{ll|ccccc}
{\rm CI} & n=7 & 0 & 0&0\\
{\rm C} & n=6& 0 & 0&0\\
{\rm CII} & n=5& 0 & 0&0\\
{\rm AII} & n=4& 2\Z & 0 & 2\Z \\ 
{\rm DIII} & n=3 & 0 & 0&0\\
{\rm D} & n=2 & 0 & 0 & \Z_2 \\ 
{\rm BDI} & n=1 & 0 & 0 & 0 \\ 
{\rm AI} & n=0 & \Z & 0 & 2\Z \\
\hline
& E^3_{p,n} & p=0 & p=1 & p=2 \\
\end{array}
\end{align}
The even integers of $E^3_{0,2}$ reflect that an odd Chern number is forbidden in even parity class D superconductors. 
Except for class AII insulators, the $K$-groups $K^{\Z_2^T \times \Z_2^{C_2}}_n(\R^2,\p \R^2)$ were fixed. 

For class AII, the $K$-group fits into the short exact sequence 
\begin{align}
0 \to \underbrace{2 \Z}_{{\rm Fermion\ number\ at\ the\ rotation\ center}} \to K^{\Z_2^T \times \Z_2^{C_2}}_4(\R^2) \to \underbrace{\Z_2}_{\rm Quantum\ spin\ Hall\ state} \to 0. 
\label{eq:extension_aii_2d_c2}
\end{align}
There exist two inequivalent group extensions 
\begin{align}
&\mbox{(i)}\qquad 0 \to 2 \Z \xrightarrow{n \mapsto (n,0)} 2 \Z \times \Z_2 \xrightarrow{(n,m) \mapsto m} \Z_2 \to 0, \\
&\mbox{(ii)}\qquad 0 \to 2 \Z \xrightarrow{n \mapsto 2n} \Z \xrightarrow{n \mapsto n} \Z_2 \to 0.
\end{align}
To fix the group extension (\ref{eq:extension_aii_2d_c2}), we ask if the double stack of $C_2$-symmetric quantum spin Hall states is adiabatically equivalent to a Kramers pair at the rotation center or not. 
Two layered quantum spin Hall states are modeled as 
\begin{align}
&H = -i \p_x s_x \sigma_x \mu_0-i \p_y s_y \sigma_x \mu_0+ m \sigma_z \mu_0+ M_1(x,y) \sigma_y \mu_y + M_2(x,y) s_z \sigma_x \mu_y, \\
&T = s_y {\cal K}, \qquad C_2 = \sigma_z, 
\end{align}
where $s_{\alpha}$, $\sigma_{\alpha}$, and $\mu_{\alpha}$ ($\alpha \in \{0,x,y,z\}$) are Pauli Matrices for spin, orbital, and layer degrees of freedoms, respectively. 
To preserve the $C_2$-rotation symmetry, the spatially varying mass terms should satisfy $M_j(-x,-y) = - M_j(x,y)$ for $j = 1,2$. 
Thanks to the mass gap $m$, the mass texture of $M_{j \in 1,2}(x,y)$ can be turned on adiabatically. 
In the presence of a single vortex of the mass texture $(M_1(x,y),M_2(x,y))$ with the $C_2$-rotation symmetry, there appears a localized ingap doublet with the effective Hamiltonian $H_{\rm eff} = m s_0$ with $C_2=s_0$ and $T=s_y {\cal K}$, i.e.\ a Kramers pair.  
Therefore, the group extension (\ref{eq:extension_aii_2d_c2}) is nontrivial and the $K$-group for class AII is $K^{\Z_2^T \times \Z_2^{C_2}}_4 (\R^2,\p \R^2) \cong \Z$, which is consistent with the known result by the $K$-cohomology~\cite{KS14}.

\subsection{$C_4 T$-rotation symmetry}
Let us consider, as the symmetry class for $n=0$, a magnetic point group symmetry $C_4 T$ composed of $C_4$-rotation $C_4: (x,y) \mapsto (-y,x)$ and a time-reversal transformation $T$. 
We also assume $(C_4 T)^4 = -1$ as for spinful systems.
According to the recipe in Sec~\ref{sec:k_homology}, the symmetries for $n>0$ read as 
\begin{align}
&n=0: \qquad 
\left\{\begin{array}{ll}
(C_4T) H(x,y) (C_4 T)^{-1} = H(-y,x) \\
(C_4 T)^4=-1, \\
\end{array}\right. \\
&n=1: \qquad 
\left\{\begin{array}{ll}
(C_4T) H(x,y) (C_4 T)^{-1} = H(-y,x) \\
\Gamma H(x,y) \Gamma^{-1} = -H(x,y) \\
(C_4 T)^4=-1, \\
\Gamma (C_4 T) = (C_4 T) \Gamma.
\end{array}\right. \\
&n=2: \qquad 
\left\{\begin{array}{ll}
(C_4T) H(x,y) (C_4 T)^{-1} = -H(-y,x) \\
(C_4 T)^4=-1, \\
\end{array}\right. \\
&n=3: \qquad 
\left\{\begin{array}{ll}
(C_4T) H(x,y) (C_4 T)^{-1} = -H(-y,x) \\
\Gamma H(x,y) \Gamma^{-1} = -H(x,y) \\
(C_4 T)^4=-1, \\
\Gamma (C_4 T) = -(C_4 T) \Gamma.
\end{array}\right. 
\end{align}
Here, we have used that in the presence of a pair of chiral symmetries $\Gamma_1 = \sigma_x$ and $\Gamma_2 = \sigma_y$ a Hamiltonian takes a form of $H = \wt H \otimes \sigma_z$ and the symmetry is recast as for $\wt H$~\cite{KSGeneric}.
The symmetry for $n=4$ is the same as $n=0$, meaning the periodicity $n \sim n+4$. 

A $C_4 T$-symmetric cell decomposition of the infinite real space $\R^2$ is shown as follows. 
$$
\includegraphics[width=\linewidth, trim=0cm 11cm 0cm 0cm]{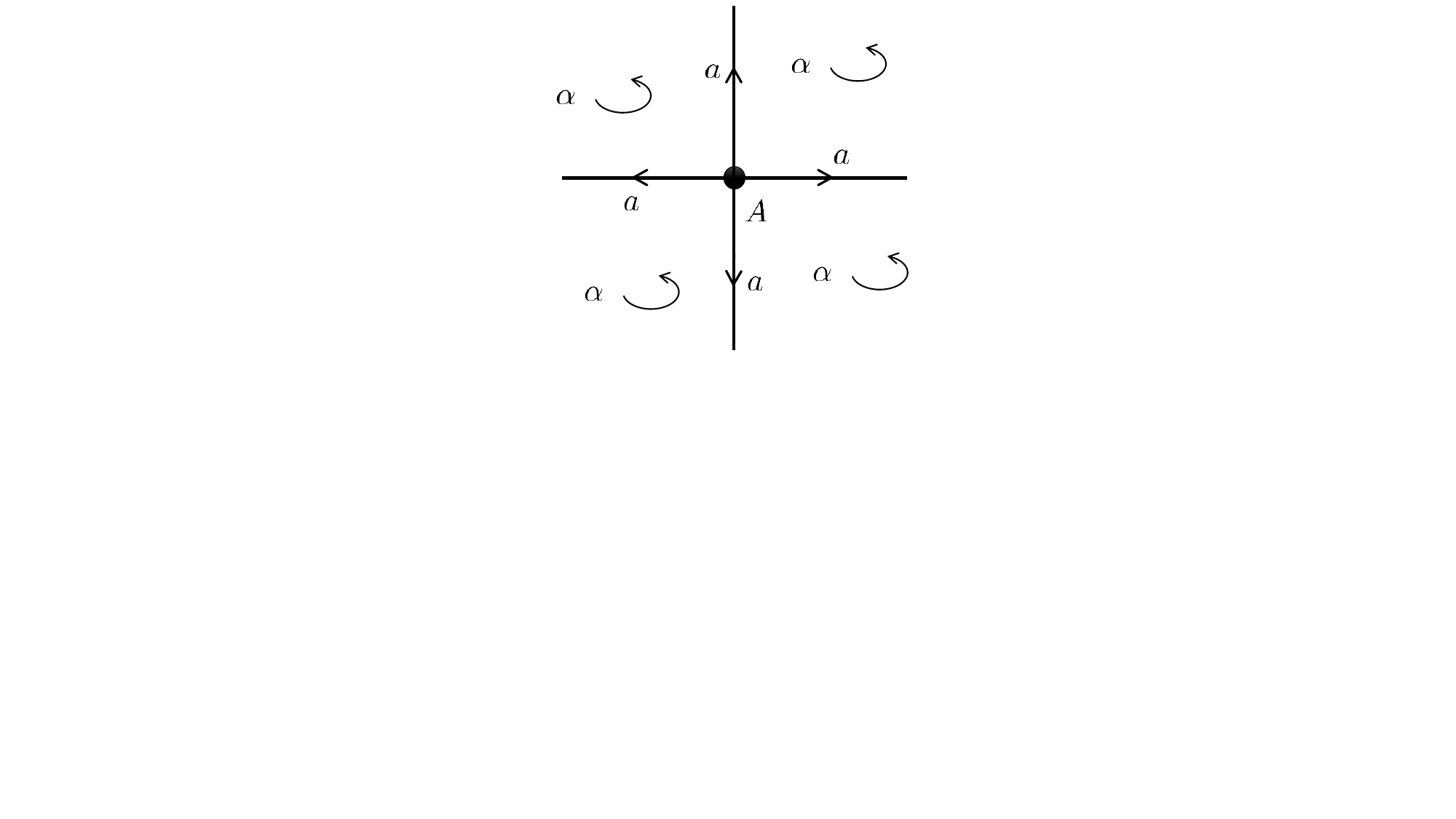}
$$
The $E_1$-page and the 1st differentials are found as 
\begin{align}
\begin{array}{l|ccccc}
n=3 & 0 & 0& 0\\
n=2 & \Z & \Z & \Z \\ 
n=1 & 0 & 0 & 0 \\ 
n=0 & \Z & \Z & \Z \\
\hline
E^1_{p,n} & p=0 & p=1 & p=2 \\
\end{array}
\end{align}
\begin{align}
d^1_{0,1} = 2, \qquad d^1_{1,2} = 0, \qquad 
d^1_{2,0} = 0, \qquad d^1_{2,2} = 2.
\end{align}
Some comments are in order. 
The 1st differentials $d^1_{1,n}$ from 1-cells to the $0$-cell are given by the induced representation. 
The 1st differential $d^1_{2,0}$ ($d^1_{2,2}$) represents how the chiral edge states of Chern insulators in 2-cells contribute to the anomalous edge states on 1-cells in the symmetry class $n=2$ ($n=0$). 
Since for $n=2$ $(n=0)$ the $C_4T$-rotation is a particle-hole (time-reversal) type, chiral edge states cancel out (sum up) at a $1$-cell. 
See the following figure. 
$$
\includegraphics[width=\linewidth, trim=0cm 9cm 0cm 0cm]{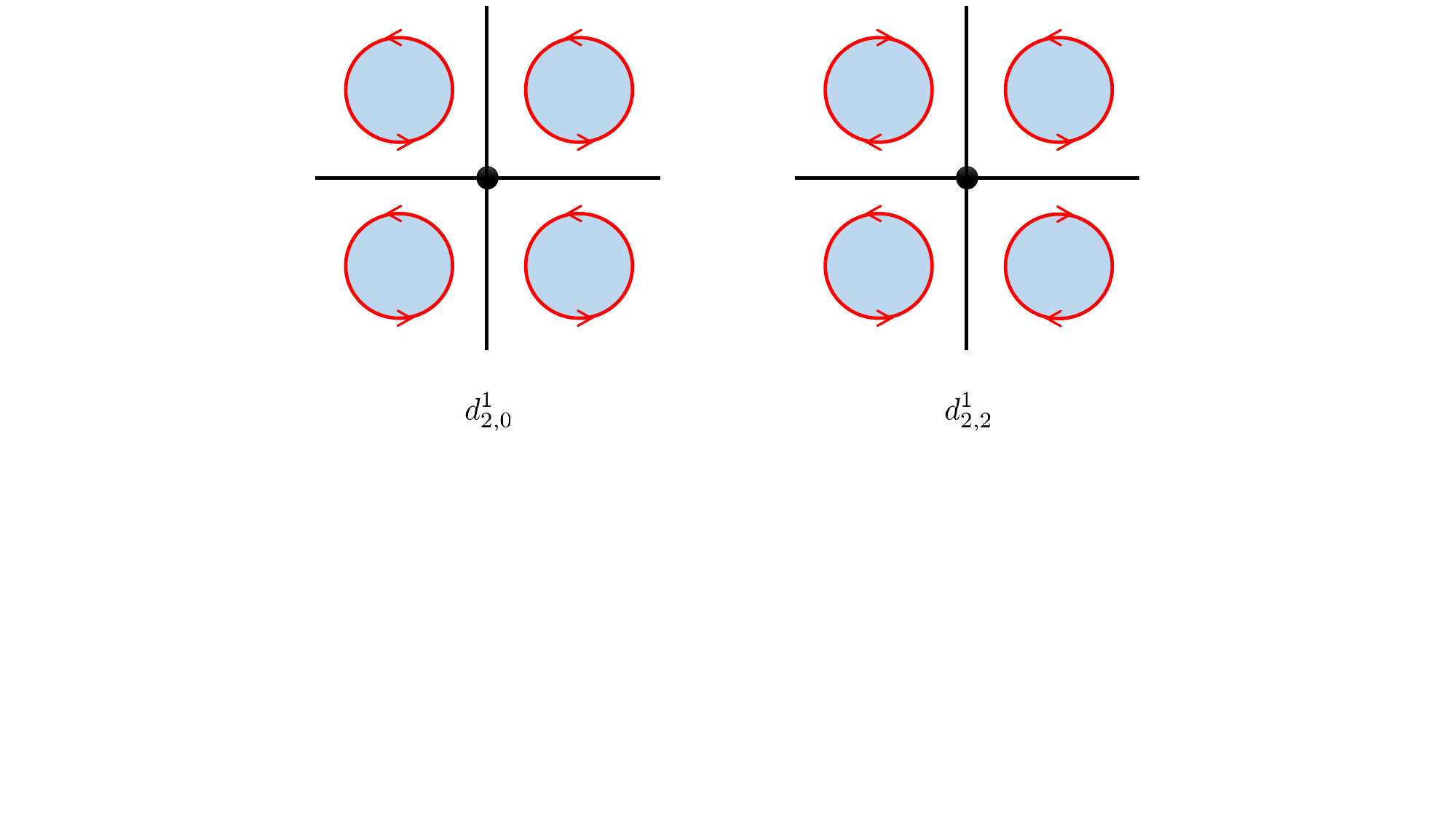}
$$
We have the $E^2$-page 
\begin{align}
\begin{array}{l|ccccc}
n=3 & 0 & 0& 0\\
n=2 & \Z & \Z_2 & 0 \\ 
n=1 & 0 & 0 & 0 \\ 
n=0 & \Z_2 & 0 & \Z \\
\hline
E^2_{p,n} & p=0 & p=1 & p=2 \\
\end{array}
\end{align}
and this is the limit. 
We find that there is a second-order topological insulator $E^{\infty}_{1,2} = \Z_2$ in the symmetry class $n=3$. 
$E^{\infty}_{1,2}$ also describes the second-order anomaly in the symmetry class $n=0$, the magnetic $4$-fold rotation symmetry $C_4T$, where the anomalous edge state is localized at a $1$-skeleton~\cite{SchindlerHigher}: 
$$
\includegraphics[width=\linewidth, trim=0cm 13cm 0cm 0cm]{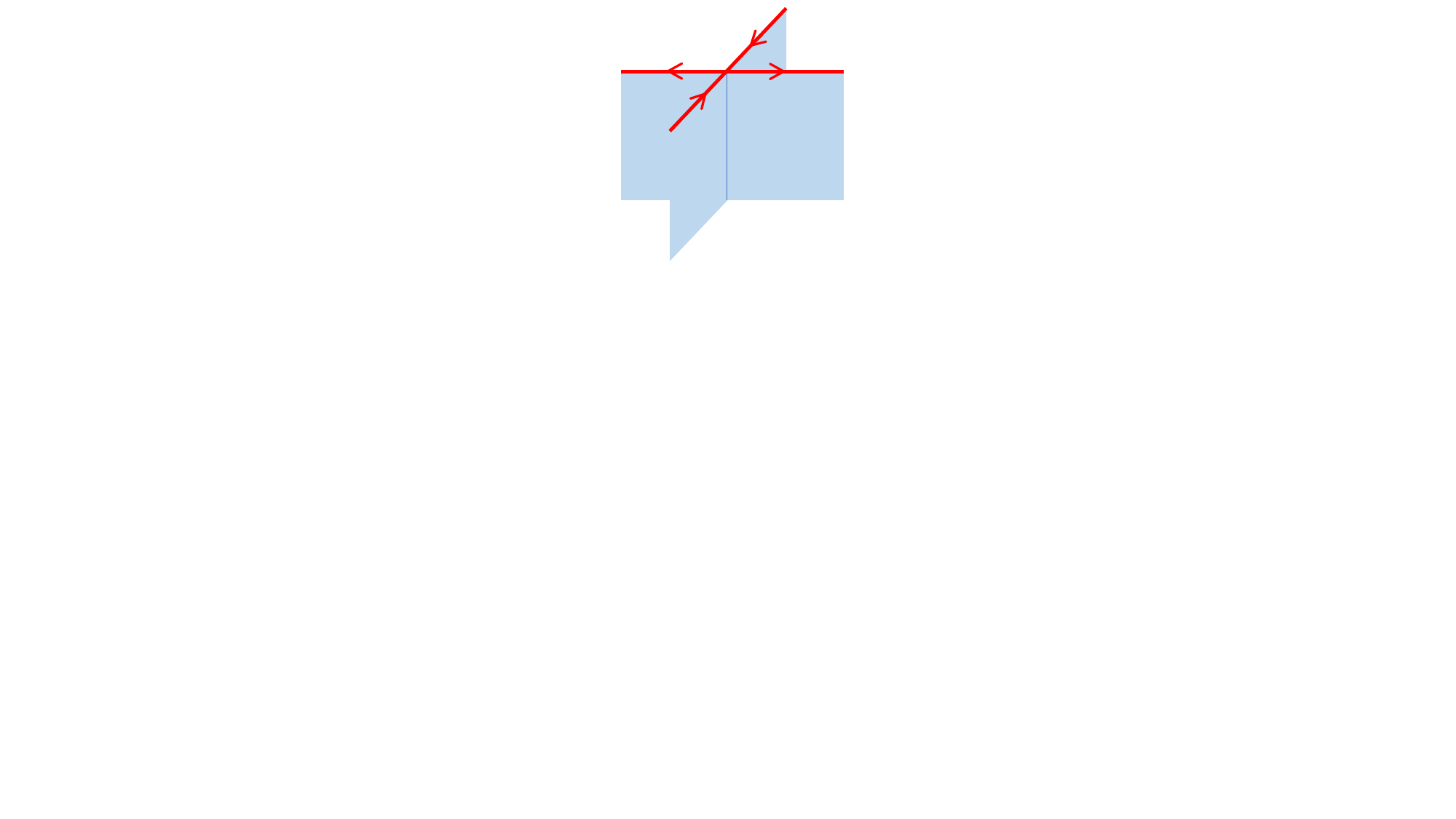}
$$

\subsection{Glide symmetry}

Let us consider $2d$ free fermions with a spatial $\Z$ symmetry generated by the glide transformation $G: (x,y) \mapsto (x+1/2,-y)$. 
Also, we assume a TRS $T$, $T^2=1$, commuting with the glide symmetry as the symmetry class for $n=0$. 
We shall compute the classification of free fermion SPT phases with glide and AZ symmetries that is represented by the $K$-homology $K^{\Z_2^T \times \Z^G}_n(\R^2,\R \times \{\pm \infty\})$ with $\R \times \{\pm\infty\}$ the infinity at $y= \pm \infty$. 
The glide-symmetric filtration of the $2d$ space $\R^2$ is shown as follows:
$$
\includegraphics[width=\linewidth, trim=0cm 12cm 0cm 0cm]{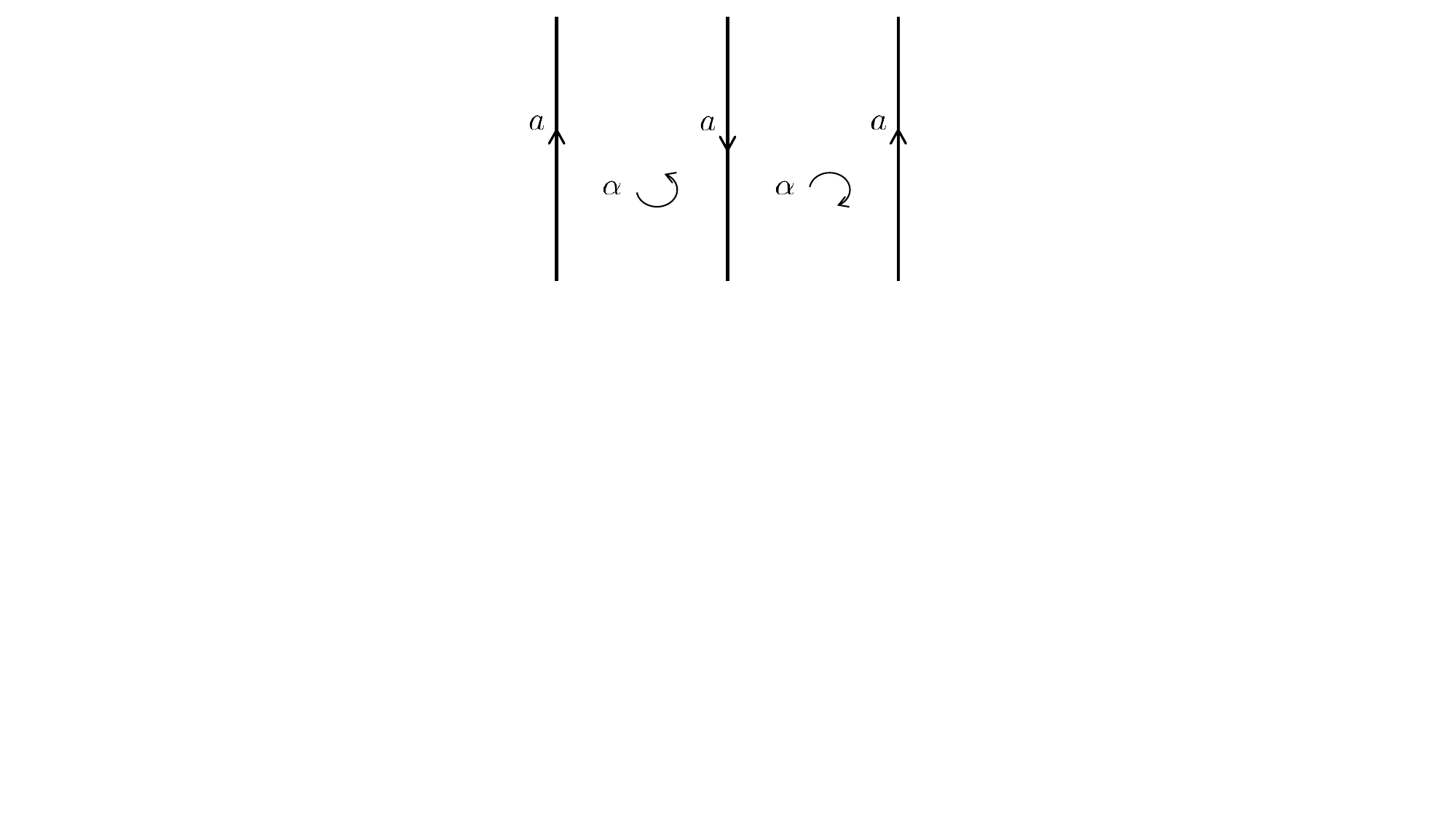}
$$
We have shown a unit cell along the $x$-direction.
Notice that there is no 0-cell.
The $E^1$-page and the first differentials are 
\begin{align}
\begin{array}{ll|ccccc}
{\rm CI} & n=7 & 0 & 0&0\\
{\rm C} & n=6& 0 & 0&0\\
{\rm CII} & n=5& 0 & 0&0\\
{\rm AII} & n=4& 0 & 2\Z & 2\Z \\ 
{\rm DIII} & n=3 & 0 & 0&0\\
{\rm D} & n=2 & 0 & \Z_2 & \Z_2 \\ 
{\rm BDI} & n=1 & 0 & \Z_2 & \Z_2 \\ 
{\rm AI} & n=0 & 0 & \Z & \Z \\
\hline
& E^1_{p,n} & p=0 & p=1 & p=2 \\
\end{array}
\end{align}
\begin{align}
d^1_{2,0}=d^1_{2,4}=2, \qquad d^1_{2,1}=d^1_{2,2}=0. 
\end{align}
Here, $d^1_{2,0}=d^1_{2,4}=2$ is because the glide reflection changes the chirality of the Chern insulator. 
Taking the homology of $d^1$ gives the $E^2$-page 
\begin{align}
\begin{array}{ll|ccccc}
{\rm CI} & n=7 & 0 & 0&0\\
{\rm C} & n=6& 0 & 0&0\\
{\rm CII} & n=5& 0 & 0&0\\
{\rm AII} & n=4& 0 & \Z_2 & 0 \\ 
{\rm DIII} & n=3 & 0 & 0&0\\
{\rm D} & n=2 & 0 & \Z_2 & \Z_2 \\ 
{\rm BDI} & n=1 & 0 & \Z_2 & \Z_2 \\ 
{\rm AI} & n=0 & 0 & \Z_2 & 0 \\
\hline
& E^2_{p,n} & p=0 & p=1 & p=2 \\
\end{array}
\end{align}
This is the limiting page $E^{\infty} = E^2$. 
$E^{\infty}_{2-k,-2+k+n}$ represents the classification of $k$-th order topological insulators/superconductors for the symmetry class $n$, and at the same time, the classification of $k$-th order anomalies for the symmetry class $n+1$. 
Except for $n=3$, the $K$-homology groups $K^{\Z_2^T \times \Z^G}_n(\R^2,\R \times \{\pm \infty\})$ are fixed. 

For $n=3$, the $K$-group fits into the short exact sequence 
\begin{align}
0 \to \underbrace{\Z_2}_{E^{\infty}_{1,2}} \to K^{\Z_2^T \times \Z^G}_3(\R^2,\R \times \{\pm \infty\}) \to \underbrace{\Z_2}_{E^{\infty}_{2,1}} \to 0.
\label{eq:2d_glide_aii}
\end{align}
Let us compute the group extension in the view of an anomaly in class AII system. 
The anomalous state described by $E^{\infty}_{2,1}$, the surface state on top of the $3d$ class AII topological insulator, preserving the glide symmetry is given by 
\begin{align}
H = -i \p_x s_y -i \p_y s_x, \qquad 
T=s_y {\cal K}, \qquad 
G=i s_y.
\end{align}
Doubling the anomalous states allows the surface to have a texture mass term 
\begin{align}
H = -i \p_x s_y \mu_0 -i \p_y s_x \mu_0 + m(x) s_z \mu_y 
\end{align}
with $\mu_{\alpha} (\alpha \in \{0,x,y,z\})$ is the Pauli matrices for two layers.
The glide symmetry implies $m(x+1/2) = -m(x)$, that is, enforcing a kink per the interval $[x,x+1/2]$. 
The low-energy effective Hamiltonian localized at a kink is a helical edge state that is in the same anomaly as $E^{\infty}_{1,2}$. 
Therefore, the extension (\ref{eq:2d_glide_aii}) is nontrivial and we have $K^{\Z_2^T \times \Z^G}_3(\R^2,\R \times \{\pm \infty\}) = \Z_4$. 
This is consistent with the classification by the $K$-cohomology~\cite{KS16}.

\subsection{3d complex AZ classes with time-reversal inversion symmetry}
\label{sec:3d complex AZ classes with time-reversal inversion symmetry}
Here we present an example where a third differential $d^3$ becomes nontrivial.
Let us consider, as the symmetry class for $n=0$, $3d$ complex fermions with the time-reversal inversion symmetry $TI$, 
\begin{align}
(TI) \psi^{\dag}(x,y,z) (TI)^{-1} = \psi^{\dag}(-x,-y,-z) {\cal U}_{TI}, \qquad 
(TI)^2 = 1, 
\end{align}
Fig.~\ref{fig:cell_inversion}[c] showed an inversion symmetric cell decomposition of the real space $\R^3$. 
The $E^1$-page and symmetry classes are summarized as 
\begin{align}
\begin{array}{l|cccc|ll}
n=7 & 0 & 0&0& 0& (TI)^2=1 & (CI)^2=-1\\
n=6& 0 & \Z&\Z&\Z & & (CI)^2=-1\\
n=5& 0 & 0&0 & 0& (TI)^2=-1 & (CI)^2=-1\\
n=4& 2\Z & \Z & \Z &\Z & (TI)^2=-1 & \\ 
n=3 & 0 & 0&0 & 0& (TI)^2=-1 & (CI)^2=1 \\
n=2 & \Z_2 & \Z & \Z &\Z & & (CI)^2=1 \\ 
n=1 & \Z_2 & 0 & 0 & 0 & (TI)^2 & (CI)^2=1 \\ 
n=0 & \Z & \Z & \Z & \Z & (TI)^2=1 & \\
\hline
E^1_{p,n} & p=0 & p=1 & p=2 & p=3 & TI & CI \\
\end{array}
\end{align}
Here, $CI$ is particle-hole inversion symmetry. 
The first differentials are found to be 
\begin{align}
&d^1_{1,0} = 2, \qquad d^1_{1,2}=1, \qquad d^1_{1,4}=1, \\
&d^1_{2,0}=d^1_{2,4}=0, \qquad d^1_{2,2}=d^1_{2,6}=2, \\
&d^1_{3,0}=d^1_{3,4}=2, \qquad d^1_{3,2}=d^1_{3,6}=0.
\end{align}
The $E^2$-page is 
\begin{align}
\begin{array}{l|cccc|ll}
n=7 & 0 & 0&0& 0& (TI)^2=1 & (CI)^2=-1\\
n=6& 0 & \Z_2&0&\Z & & (CI)^2=-1\\
n=5& 0 & 0&0 & 0& (TI)^2=-1 & (CI)^2=-1\\
n=4& 0 & 0 & \Z_2 &0 & (TI)^2=-1 & \\ 
n=3 & 0 & 0&0 & 0& (TI)^2=-1 & (CI)^2=1 \\
n=2 & 0 & 2\Z & 0 &\Z & & (CI)^2=1 \\ 
n=1 & \Z_2 & 0 & 0 & 0 & (TI)^2 & (CI)^2=1 \\ 
n=0 & \Z_2 & 0 & \Z_2 & 0 & (TI)^2=1 & \\
\hline
E^2_{p,n} & p=0 & p=1 & p=2 & p=3 & TI & CI \\
\end{array}
\end{align}
In this table, $d^2_{2,0}:E^2_{2,0}\to E^2_{0,1}$ can be nontrivial and is found to be nontrivial $d^2_{2,0}=1$ in a way similar to $d^2_{2,1}$ in Sec.~\ref{sec:free_2d_real_c2}: 
$E^2_{2,0}$ represents a Chern insulator with an odd Chern number in the 2-cell $\alpha$, and the symmetry $(CI)^2=1$ for $n=2$ implies that the boundary condition of the chiral edge state is periodic, yielding an exact zero energy state generating $E^2_{0,1}$. 

The homology of $d^2$ gives the $E^3$-page 
\begin{align}
\begin{array}{l|cccc|ll}
n=7 & 0 & 0&0& 0& (TI)^2=1 & (CI)^2=-1\\
n=6& 0 & \Z_2&0&\Z & & (CI)^2=-1\\
n=5& 0 & 0&0 & 0& (TI)^2=-1 & (CI)^2=-1\\
n=4& 0 & 0 & \Z_2 &0 & (TI)^2=-1 & \\ 
n=3 & 0 & 0&0 & 0& (TI)^2=-1 & (CI)^2=1 \\
n=2 & 0 & 2\Z & 0 &\Z & & (CI)^2=1 \\ 
n=1 & 0 & 0 & 0 & 0 & (TI)^2 & (CI)^2=1 \\ 
n=0 & \Z_2 & 0 & 0 & 0 & (TI)^2=1 & \\
\hline
E^3_{p,n} & p=0 & p=1 & p=2 & p=3 & TI & CI \\
\end{array}
\end{align}
In this table, $d^3_{3,6}:E^3_{3,6} \to E^3_{0,0}$ can be nontrivial. 
Let us derive $d^3_{3,6}$ in the viewpoint of an adiabatic pump in the symmetry class of $n=0$. 
$E^3_{3,6}$ represents the creation of a Chern insulator with a unit Chern number on a sphere $S^2$ inside the north and south $3$-cells. 
Since $d^1_{3,6}=d^2_{3,6}=0$, these Chern insulators can glue together at 2- and 1-cells. 
Hence, the problem is if the Chern insulator enclosing the inversion center preserving the time-reversal inversion symmetry with $(TI)^2=1$ has a unit $U(1)$ charge or not.
For complex fermions on a 2-sphere, it is known that in the presence of a monopole charge $m_g$ inside the sphere, the $z$-component of the angular momentum is quantized into (i) odd integers if $m_g \in 2 \Z$ and (ii) even integers if $m_g \in 2\Z+1$. 
On the one hand, the symmetry algebra $(TI)^2=1$ implies that the $2 \pi$-rotation is the identity, i.e.\ the angular momentum should be an even integer. 
Therefore, the symmetry algebra $(TI)^2 = 1$ enforces an odd monopole charge $m_g \in 2\Z+1$ inside the 2-sphere. 
Then, from the quantum Hall effect, on the $2$-sphere with an odd monopole charge $m_g$, the Chern insulator with a unit Chern number has an odd $U(1)$ charge $m_g$~\cite{KS_Antiunitary}, the generator of $E^3_{0,0} = \Z_2$. 
Therefore, we conclude that $d^3_{3,6}=1$. 

We arrived at the $E^4 = E^{\infty}$-page
\begin{align}
\begin{array}{l|cccc|ll}
n=7 & 0 & 0&0& 0& (TI)^2=1 & (CI)^2=-1\\
n=6& 0 & \Z_2&0&2\Z & & (CI)^2=-1\\
n=5& 0 & 0&0 & 0& (TI)^2=-1 & (CI)^2=-1\\
n=4& 0 & 0 & \Z_2 &0 & (TI)^2=-1 & \\ 
n=3 & 0 & 0&0 & 0& (TI)^2=-1 & (CI)^2=1 \\
n=2 & 0 & 2\Z & 0 &\Z & & (CI)^2=1 \\ 
n=1 & 0 & 0 & 0 & 0 & (TI)^2 & (CI)^2=1 \\ 
n=0 & 0 & 0 & 0 & 0 & (TI)^2=1 & \\
\hline
E^4_{p,n} & p=0 & p=1 & p=2 & p=3 & TI & CI \\
\end{array}
\end{align}
The $K$-groups and the classification of higher-order topological insulators/superconductors are consistent with \cite{KS14, TB18}.

\subsection{3d complex AZ classes with space group $P222$}
In this and subsequent two sections we illustrate the AHSS for $3d$ space group symmetry. 

\begin{figure}[!]
\begin{center}
\includegraphics[width=0.7\linewidth, trim=0cm 18cm 0cm 0cm]{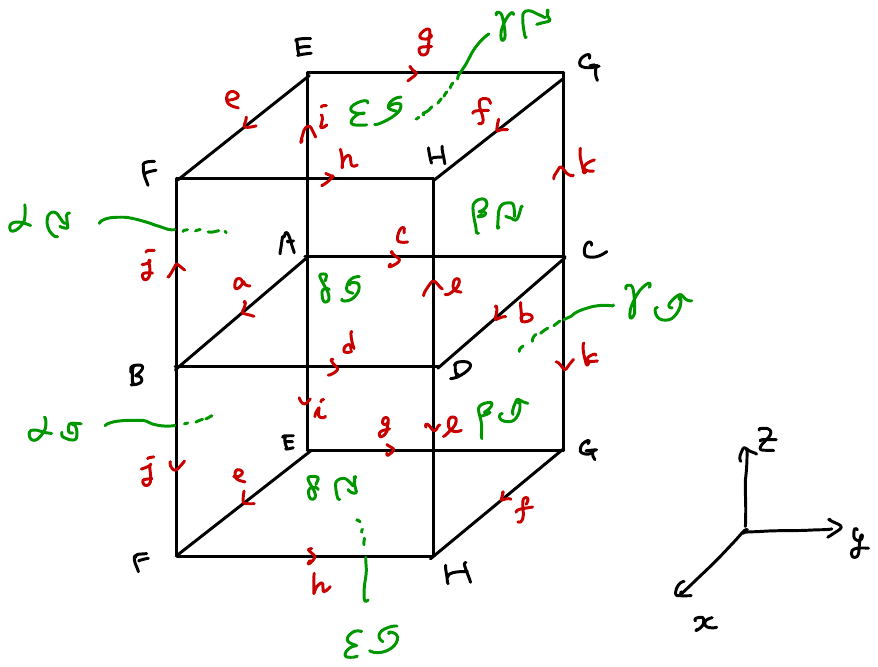}
\end{center}
\caption{
A $P222$-equivariant cell decomposition of the real space $\R^3$.
The figure shows the independent region which is one-quarter of a unit cell. 
}
\label{fig:P222}
\end{figure}
Let us consider the $K$-homology $K^{G}_n(\R^3)$ for complex AZ classes with the space group $G = P222$ (No.\ 16). 
The space group $P222$ is composed of the primitive $3d$ lattice translations $T_{\mu}: \bm{x} \mapsto \bm{x} + \hat x_{\mu} (\mu = x,y,z)$ and the point group $D_2 = \{1,2_{100},2_{010},2_{001}\}$ which is generated by $2_{100}:(x,y,z)\mapsto (x,-y,-z)$ and $2_{010}:(x,y,z)\mapsto(-x,y,-z)$.
We here consider spinless systems. 
A $P222$-symmetric cell decomposition of the real space $\R^3$ is shown in Fig.~\ref{fig:P222}, and it is composed of 
\begin{equation}\begin{split}
&\mbox{0-cells} = \{A,B,C,D,E,F,G,H\}, \\
&\mbox{1-cells} = \{a,b,c,d,e,f,g,h,i,j,k,l\}, \\
&\mbox{2-cells} = \{\alpha,\beta,\gamma,\delta,\epsilon\}, \\
&\mbox{3-cells} = \{vol\}.
\end{split}\end{equation}
The little groups are $D_2$ at $0$-cells and $\Z_2$ at 1-cells. 
Counting the number of irreps.\ we have the $E^1$-page 
\begin{align}
\begin{array}{ll|cccc}
{\rm AIII} & n=1 & 0 & 0 & 0 & 0 \\ 
{\rm A} & n=0 & (\Z^{\times 4})^{\times 8} & (\Z^{\times 2})^{\times 12} & \Z^5 & \Z \\
\hline
& E^1_{p,n} & p=0 & p=1 & p=2 & p=3 \\
\end{array}
\end{align}
From the induced representations, the first differentials are straightforwardly given as 
{\scriptsize
\arraycolsep=0.6pt
\begin{align}
d^1_{1,0}
= \begin{array}{rr|rr|rr|rr|rr|rr|rr|rr|rr|rr|rr|rr||cccccccc}
a&&b&&c&&d&&e&&f&&g&&h&&i&&j&&k&&l&&&\\
1&t&1&t&1&t&1&t&1&t&1&t&1&t&1&t&1&t&1&t&1&t&1&t&&\\
\hline
-1&0&&&-1&0&&&&&&&&&&&-1&0&&&&&&&1&A\\
0&-1&&&-1&0&&&&&&&&&&&0&-1&&&&&&&t_1&\\
-1&0&&&0&-1&&&&&&&&&&&0&-1&&&&&&&t_2&\\
0&-1&&&0&-1&&&&&&&&&&&-1&0&&&&&&&t_1t_2&\\
\hline
1&0&&&&&-1&0&&&&&&&&&&&-1&0&&&&&1&B\\
0&1&&&&&-1&0&&&&&&&&&&&0&-1&&&&&t_1&\\
1&0&&&&&0&-1&&&&&&&&&&&0&-1&&&&&t_2&\\
0&1&&&&&0&-1&&&&&&&&&&&-1&0&&&&&t_1t_2&\\
\hline
&&-1&0&1&0&&&&&&&&&&&&&&&-1&0&&&1&C\\
&&0&-1&1&0&&&&&&&&&&&&&&&0&-1&&&t_1&\\
&&-1&0&0&1&&&&&&&&&&&&&&&0&-1&&&t_2&\\
&&0&-1&0&1&&&&&&&&&&&&&&&-1&0&&&t_1t_2&\\
\hline
&&1&0&&&1&0&&&&&&&&&&&&&&&-1&0&1&D\\
&&0&1&&&1&0&&&&&&&&&&&&&&&0&-1&t_1&\\
&&1&0&&&0&1&&&&&&&&&&&&&&&0&-1&t_2&\\
&&0&1&&&0&1&&&&&&&&&&&&&&&-1&0&t_1t_2&\\
\hline
&&&&&&&&-1&0&&&-1&0&&&1&0&&&&&&&1&E\\
&&&&&&&&0&-1&&&-1&0&&&0&1&&&&&&&t_1&\\
&&&&&&&&-1&0&&&0&-1&&&0&1&&&&&&&t_2&\\
&&&&&&&&0&-1&&&0&-1&&&1&0&&&&&&&t_1t_2&\\
\hline
&&&&&&&&1&0&&&&&-1&0&&&1&0&&&&&1&F\\
&&&&&&&&0&1&&&&&-1&0&&&0&1&&&&&t_1&\\
&&&&&&&&1&0&&&&&0&-1&&&0&1&&&&&t_2&\\
&&&&&&&&0&1&&&&&0&-1&&&1&0&&&&&t_1t_2&\\
\hline
&&&&&&&&&&-1&0&1&0&&&&&&&1&0&&&1&G\\
&&&&&&&&&&0&-1&1&0&&&&&&&0&1&&&t_1&\\
&&&&&&&&&&-1&0&0&1&&&&&&&0&1&&&t_2&\\
&&&&&&&&&&0&-1&0&1&&&&&&&1&0&&&t_1t_2&\\
\hline
&&&&&&&&&&1&0&&&1&0&&&&&&&1&0&1&H\\
&&&&&&&&&&0&1&&&1&0&&&&&&&0&1&t_1&\\
&&&&&&&&&&1&0&&&0&1&&&&&&&0&1&t_2&\\
&&&&&&&&&&0&1&&&0&1&&&&&&&1&0&t_1t_2&\\
\hline
\end{array}, 
\end{align}
\begin{align}
d^1_{2,0}
= 
\begin{array}{r|r|r|r|r||cc}
\alpha&\beta&\gamma&\delta&\epsilon&&\\
\hline
1&&&&&1&a\\
1&&&&&t&\\
\hline
&1&&&&1&b\\
&1&&&&t&\\
\hline
&&-1&&&1&c\\
&&-1&&&t&\\
\hline
&&&1&&1&d\\
&&&1&&t&\\
\hline
-1&&&&1&1&e\\
-1&&&&1&t&\\
\hline
&-1&&&-1&1&f\\
&-1&&&-1&t&\\
\hline
&&1&&-1&1&g\\
&&1&&-1&t&\\
\hline
&&&-1&1&1&h\\
&&&-1&1&t&\\
\hline
-1&&1&&&1&i\\
-1&&1&&&t&\\
\hline
1&&&-1&&1&j\\
1&&&-1&&t&\\
\hline
&-1&-1&&&1&k\\
&-1&-1&&&t&\\
\hline
&1&&1&&1&l\\
&1&&1&&t&\\
\hline
\end{array}, 
\end{align}
}
and $d^1_{3,0}=0$. 
Here, $\{1,t\}$ in 1-cells meant trivial and the sign irreps.\ of $\Z_2$, and $\{1,t_1,t_2,t_1t_2\}$ in $0$-cells are four irreps.\ of $D_2$. 
We can check that $d^1_{2,0} d^1_{3,0}=d^1_{1,0}d^1_{2,0}=0$. 
Taking the homology of $d^1$ gives the $E^2$-page 
\begin{align}
\begin{array}{ll|cccc}
{\rm AIII} & n=1 & 0 & 0 & 0 & 0 \\ 
{\rm A} & n=0 & \Z^{\times 13} \times \Z_2 & 0 & 0 & \Z \\
\hline
& E^2_{p,n} & p=0 & p=1 & p=2 & p=3 \\
\end{array}
\label{eq:e3_p222}
\end{align}
Because $d^2=0$, $E^2 = E^3$. 
In the $E^3$-page, the third differential $d^3_{3,0}: \Z \to \Z^{\times 13} \times \Z_2$ can be nontrivial. 

Interestingly, comparing the $E^3$-page (\ref{eq:e3_p222}) with the $E_{\infty}$-page~\cite{KSAHSS}
\begin{align}
\begin{array}{ll|cccc}
{\rm A} & n=0 & \Z^{\times 13} & \Z_2 & 0 & \Z \\
{\rm AIII} & n=1 & 0 & 0 & 0 & 0 \\ 
\hline
& E_{\infty}^{p,-n} & p=0 & p=1 & p=2 & p=3 \\
\end{array}
\end{align}
of the $K$-cohomology group $K^{-n}_{D_2}(T^3)$ isomorphic to $K^G_n(\R^3)$, we find that $d^3_{3,0}$ must remove the $\Z_2$ subgroup of $E^3_{0,0}$, otherwise the $K$-group $K^G_n(\R^3) \cong K^{0}_{D_2}(T^3) (= \Z^{\times 13}) $ has a torsion.  
We have the $E^{\infty} = E^4$-page 
\begin{align}
\begin{array}{ll|cccc}
{\rm AIII} & n=1 & 0 & 0 & 0 & 0 \\ 
{\rm A} & n=0 & \Z^{\times 13} & 0 & 0 & 2\Z \\
\hline
& E^4_{p,n} & p=0 & p=1 & p=2 & p=3 \\
\end{array}
\end{align}
Here, $E^4_{3,0} = 2 \Z$ means that a class AIII insulator putting in $3$-cells compatible with the $P222$ space group symmetry must have an even $3d$ winding number. 
Also, based on the terminology introduced in Sec.~\ref{sec:lsm_spt}, the nontrivial third differential $d^3_{3,-3}$ leads to the LSM-type theorem to enforce a nontrivial topological insulator: 
If a class AIII system composed of anomalous zero energy degrees of freedom living in $\Z_2 \subset E^2_{0,0}$ forms a fully gapped state, this state must have an odd $3d$ winding number. 
From the $E^4$-page, the $K$-groups are determined as 
\begin{align}
K^G_0(\R^3) \cong \Z^{13}, \qquad 
K^G_1(\R^3) \cong 2 \Z.
\end{align}
It should be also noticed that this result gives us the correct group extension of the $E_{\infty}$-page of the $K$-cohomology $K^{-n}_{D_2}(T^3)$ for class AIII 
\begin{align}
0 \to \underbrace{\Z}_{E_{\infty}^{3,0}} \to \underbrace{K^{-1}_{D_2}(T^3)}_{\cong \Z} \to \underbrace{\Z_2}_{E_{\infty}^{1,0}} \to 0. 
\end{align}
The $E^{\infty}$-page of the $K$-homology and the $E_{\infty}$-page for the dual $K$-cohomology are quite complementary.

\subsection{3d complex AZ classes with space group $P2_12_12_1$}
\begin{figure}[!]
\begin{center}
\includegraphics[width=0.6\linewidth, trim=0cm 11cm 0cm 0cm]{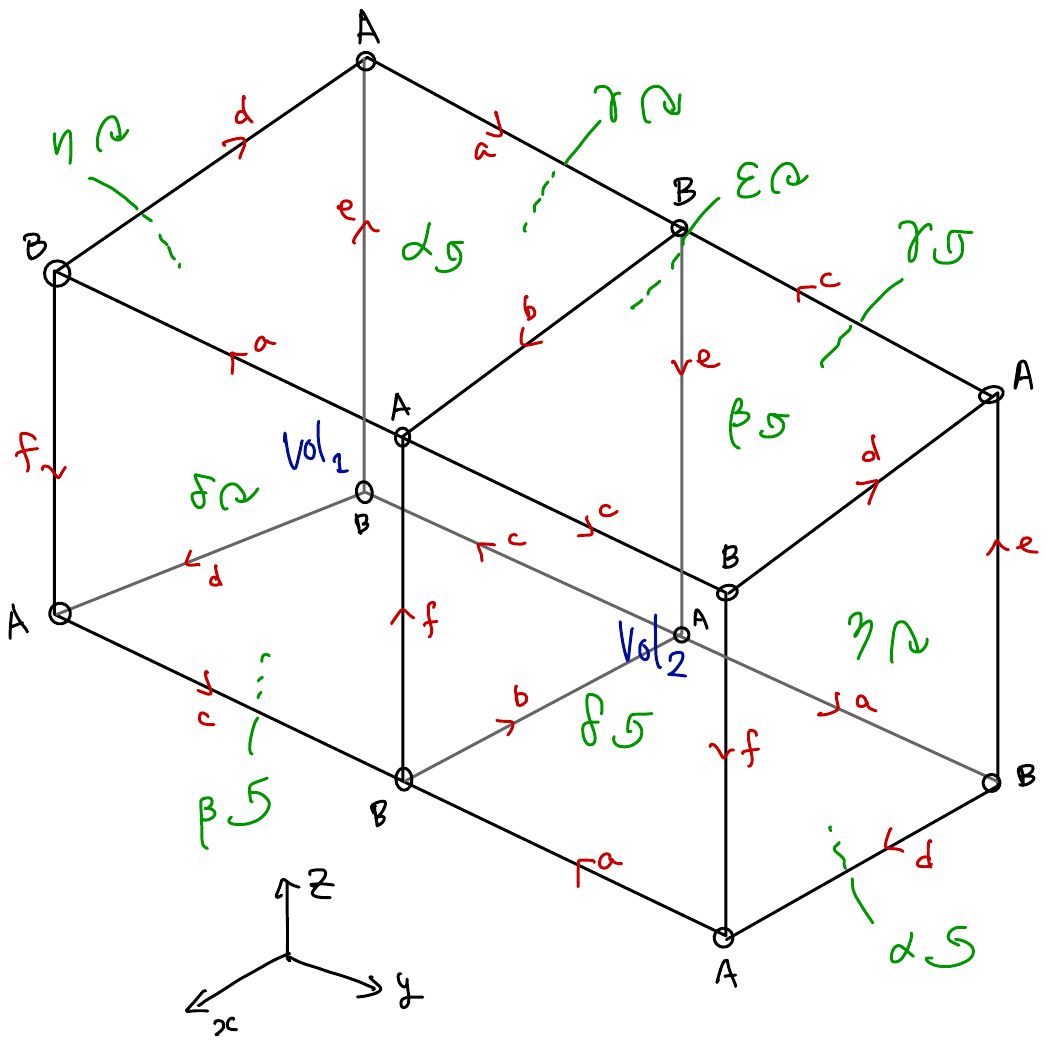}
\end{center}
\caption{
A $P2_12_12_1$-equivariant cell decomposition of the infinite space $\R^3$. 
The figure shows the independent region which is one-quarter of the unit cell. 
}
\label{fig:P212121}
\end{figure}

The next example of space group is $P2_12_12_1$ (No.\ 19) that is composed of the primitive $3d$ lattice translations and $D_2$ group action generated by 
$2_{100}:(x,y,z) \mapsto (x+1/2,-y+1/2,-z)$ and $2_{010}:(x,y,z) \mapsto (-x,y+1/2,-z+1/2)$. 
A $P2_12_12_1$-symmetric cell decomposition is in Fig.~\ref{fig:P212121} where we have shown one-quarter of the unit cell. 
$p$-cells are 
\begin{equation}\begin{split}
&\mbox{0-cells} = \{A,B\}, \\
&\mbox{1-cells} = \{a,b,c,d,e,f\}, \\
&\mbox{2-cells} = \{\alpha,\beta,\gamma,\delta,\epsilon,\eta\}, \\
&\mbox{3-cells} = \{vol_1,vol_2\}.
\end{split}\end{equation}
Because the space group $P2_12_12_1$ acts on the real space $\R^3$ freely, the irrep.\ at each $p$-cell is unique. 
We have the $E^1$-page 
\begin{align}
\begin{array}{ll|cccc}
{\rm AIII} & n=1 & 0 & 0 & 0 & 0 \\ 
{\rm A} & n=0 & \Z^{\times 2} & \Z^{\times 6} & \Z^{\times 6} & \Z^{\times 2} \\
\hline
& E^1_{p,n} & p=0 & p=1 & p=2 & p=3 \\
\end{array}
\end{align}
From the induced representations, the first differentials are given as 
\begin{align}
d^1_{1,0}
= 
\begin{array}{rrrrrr|r}
a&b&c&d&e&f&\\
\hline
1&-1&1&-1&-1&-1&A\\
-1&1&-1&1&1&1&B\\
\end{array}, 
\end{align}
\begin{align}
d^1_{2,0}
= 
\begin{array}{rrrrrr|r}
\alpha&\beta&\gamma&\delta&\epsilon&\eta&\\
\hline
-2&&1&-1&&&a\\
-1&1&&&-2&&b\\
&2&1&-1&&&c\\
-1&1&&&2&d\\
&&2&&1&-1&e\\
&&&-2&1&-1&f
\end{array}, 
\qquad 
d^1_{3,0}
= 
\begin{array}{rr|r}
vol_1&vol_2&\\
\hline
1&-1&\alpha\\
-1&1&\beta\\
1&-1&\gamma\\
-1&1&\delta\\
-1&1&\epsilon\\
1&-1&\eta
\end{array}, 
\end{align}
which satisfy $d^1 \circ d^1=0$. 
The homology of $d^1$ gives the $E^2$-page 
\begin{align}
\begin{array}{ll|cccc}
{\rm AIII} & n=1 & 0 & 0 & 0 & 0 \\ 
{\rm A} & n=0 & \Z & \Z_4^{\times 2} & 0 & \Z \\
\hline
& E^2_{p,n} & p=0 & p=1 & p=2 & p=3 \\
\end{array}
\end{align}
The 3rd differential $d^3_{3,0}:\Z\to \Z$ can be nontrivial. 
Comparing this with the $E_{\infty}$-page~\cite{KSAHSS}
\begin{align}
\begin{array}{ll|cccc}
{\rm A} & n=0 & \Z & \Z_4^{\times 3} & 0 & \Z \\
{\rm AIII} & n=1 & 0 & 0 & 0 & 0 \\ 
\hline
& E_{\infty}^{p,-n} & p=0 & p=1 & p=2 & p=3 \\
\end{array}
\end{align}
of the dual $K$-cohomology $K^{\tau-n}_{D_2}(T^3)$, we find that $d^3_{3,0}=0$. 
Therefore, $E^2 = E^{\infty}$. 
The $K$-homology group is fixed as 
\begin{align}
K^G_0(\R^3) \cong \Z, \qquad 
K^G_1(\R^3) \cong \Z \times \Z_4^{\times 2}.
\end{align}
It should be noticed again that the AHSSs for the $K$-cohomology and homology are complimentary. 
The $E^{\infty}$- and $E_{\infty}$-page implies that the $K$-cohomology group for class AIII obeys the nontrivial extension 
\begin{align}
0 \to \Z \xrightarrow{4} K^{\tau-1}_{D_2}(T^3) \xrightarrow{\rm mod\ 4} \Z_4^{\times 3} \to 0. 
\end{align}

\subsection{3d complex AZ classes with space group $F222$}
\begin{figure}[!]
\begin{center}
\includegraphics[width=0.6\linewidth, trim=0cm 13cm 0cm 0cm]{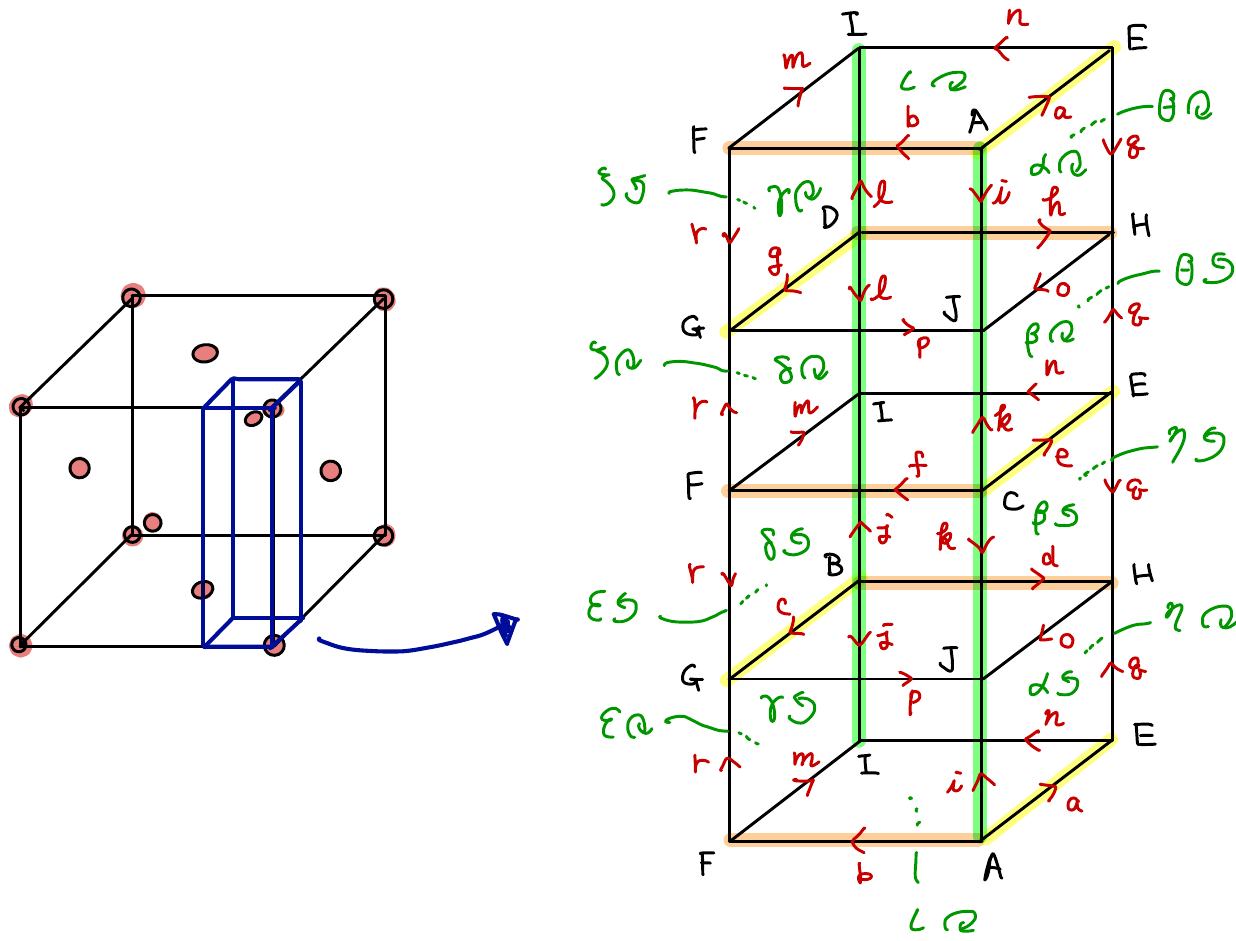}
\end{center}
\caption{
A $F222$-symmetric cell decomposition of the infinite space $\R^3$. 
The left figure shows the face-centered cubic lattice. 
The right figure shows an independent region. 
}
\label{fig:F222}
\end{figure}

The final example is the $3d$ complex AZ classes with the space group $F222$ (No.\ 22) that is generated by the $3d$ lattice translations by $(0,1,1)$, $(0,1,0)$, and $(0,1,1)$ of F222 and the $D_2$ point group. 
A $F222$-symmetric cell decomposition is in Fig.~\ref{fig:F222}. 
$p$-cells are composed of 
\begin{equation}\begin{split}
&\mbox{0-cells} = \{A,B,C,D,E,F,G,H,I,J\}, \\
&\mbox{1-cells} = \{a,b,c,d,e,f,g,h,i,j,k,l,m,n,o,p,q,r\}, \\
&\mbox{2-cells} = \{\alpha,\beta,\gamma,\delta,\epsilon,\zeta,\eta,\theta,\iota\}, \\
&\mbox{3-cells} = \{vol\}.
\end{split}\end{equation}
Let us consider spinful fermions, meaning that at the high-symmetric points $\{A,B,C,D\}$ of the $D_2$ group obeys the nontrivial projective representation of $D_2$. 
We have the $E^1$-page 
\begin{align}
\begin{array}{ll|cccc}
{\rm AIII} & n=1 & 0 & 0 & 0 & 0 \\ 
{\rm A} & n=0 & \Z^{\times 16} & \Z^{\times 30} & \Z^{\times 9} & \Z \\
\hline
& E^1_{p,n} & p=0 & p=1 & p=2 & p=3 \\
\end{array}
\end{align}
The first differentials are given as $d^1_{3,0}=0$, 
{\scriptsize
\arraycolsep=0.6pt
\begin{align}
d^1_{1,0}= 
\begin{array}{rr|rr|rr|rr|rr|rr|rr|rr|rr|rr|rr|rr|r|r|r|r|r|r||ccccccccc}
a&&b&&c&&d&&e&&f&&g&&h&&i&&j&&k&&l&&m&n&o&p&q&r&&\\
1&t&1&t&1&t&1&t&1&t&1&t&1&t&1&t&1&t&1&t&1&t&1&t&&&&&&&&&&&&&\\
\hline
-1&-1&-1&-1&&&&&&&&&&&&&-1&-1&&&&&&&&&&&&&&A\\
\hline
&&&&-1&-1&-1&-1&&&&&&&&&&&-1&-1&&&&&&&&&&&&B\\
\hline
&&&&&&&&-1&-1&-1&-1&&&&&&&&&-1&-1&&&&&&&&&&C\\
\hline
&&&&&&&&&&&&-1&-1&-1&-1&&&&&&&-1&-1&&&&&&&&D\\
\hline
1&0&&&&&&&1&0&&&&&&&&&&&&&&&&-1&&&-1&&1&E\\
0&1&&&&&&&0&1&&&&&&&&&&&&&&&&-1&&&-1&&t&\\
\hline
&&1&0&&&&&&&1&0&&&&&&&&&&&&&-1&&&&&-1&1&F\\
&&0&1&&&&&&&0&1&&&&&&&&&&&&&-1&&&&&-1&t&\\
\hline
&&&&1&0&&&&&&&1&0&&&&&&&&&&&&&&-1&&1&1&G\\
&&&&0&1&&&&&&&0&1&&&&&&&&&&&&&&-1&&1&t&\\
\hline
&&&&&&1&0&&&&&&&1&0&&&&&&&&&&&-1&&1&&1&H\\
&&&&&&0&1&&&&&&&0&1&&&&&&&&&&&-1&&1&&t&\\
\hline
&&&&&&&&&&&&&&&&&&1&0&&&1&0&1&1&&&&&1&I\\
&&&&&&&&&&&&&&&&&&0&1&&&0&1&1&1&&&&&t&\\
\hline
&&&&&&&&&&&&&&&&1&0&&&1&0&&&&&1&1&&&1&J\\
&&&&&&&&&&&&&&&&0&1&&&0&1&&&&&1&1&&&t&\\
\hline
\end{array}, 
\end{align}
\begin{align}
d^1_{2,0}
= 
\begin{array}{r|r|r|r|r|r|r|r|r||cc}
\alpha&\beta&\gamma&\delta&\epsilon&\zeta&\eta&\theta&\iota&&\\
\hline
1&&&&&&&&-1&1&a\\
1&&&&&&&&-1&t&\\
\hline
&&-1&&&&&&1&1&b\\
&&-1&&&&&&1&t&\\
\hline
&&&&-1&&&&&1&c\\
&&&&-1&&&&&t&\\
\hline
&&&&&&1&&&1&d\\
&&&&&&1&&&t&\\
\hline
&-1&&&&&&&&1&e\\
&-1&&&&&&&&t&\\
\hline
&&&1&&&&&&1&f\\
&&&1&&&&&&t&\\
\hline
&&&&&-1&&&&1&g\\
&&&&&-1&&&&t&\\
\hline
&&&&&&&-1&&1&h\\
&&&&&&&-1&&t&\\
\hline
-1&&1&&&&&&&1&i\\
-1&&1&&&&&&&t&\\
\hline
&&&&1&&-1&&&1&j\\
&&&&1&&-1&&&t&\\
\hline
&1&&-1&&&&&&1&k\\
&1&&-1&&&&&&t&\\
\hline
&&&&&1&&1&&1&l\\
&&&&&1&&1&&t&\\
\hline
&&&&-1&-1&&&1&&m\\
\hline
&&&&&&1&-1&-1&&n\\
\hline
1&-1&&&&&&&&&o\\
\hline
&&-1&1&&&&&&&p\\
\hline
1&-1&&&&&-1&1&&&q\\
\hline
&&-1&1&1&1&&&&&r\\
\hline
\end{array}.
\end{align}
}
These satisfy $d^1_{1,0} \circ d^1_{2,0}=0$. 
The homology of $d^1$ gives the $E^2=E^3$-page 
\begin{align}
\begin{array}{ll|cccc}
{\rm AIII} & n=1 & 0 & 0 & 0 & 0 \\ 
{\rm A} & n=0 & \Z \times \Z_2 & \Z^{\times 6} & 0 & \Z \\
\hline
& E^2_{p,n} & p=0 & p=1 & p=2 & p=3 \\
\end{array}
\end{align}
The 3rd differential $d^3_{3,0}$ can be nontrivial. 
Comparing the $E^3$-page with the $E_{\infty}$-page~\cite{KSAHSS}
\begin{align}
\begin{array}{ll|cccc}
{\rm A} & n=0 & \Z & \Z^{\times 6} & \Z_2 & \Z \\
{\rm AIII} & n=1 & 0 & 0 & 0 & 0 \\ 
\hline
& E_{\infty}^{p,-n} & p=0 & p=1 & p=2 & p=3 \\
\end{array}
\end{align}
of the $K$-cohomology group $K^{\tau-n}_{D_2}(T^3)$, we find that $d^3_{3,0}$ is trivial.
Therefore, $E^2 = E^{\infty}$. 
The $K$-homology group is fixed as 
\begin{align}
K^G_0(\R^3) \cong \Z \times \Z_2, \qquad 
K^G_1(\R^3) \cong \Z \times \Z^{\times 6}. 
\end{align}
$E^{\infty}_{0,0}=\Z \times \Z_2$ tells us that the $\Z_2$ nontrivial model in class A insulators is an atomic insulator, even if the $\Z_2$ invariant in the momentum space defined on the $2$-skeleton include the integral of the Berry curvature~\cite{KSAHSS}. 
The origin of the $\Z_2$ in atomic insulators is found in the structure of the 1-skeleton in the real space $\R^3$. 
See Fig.~\ref{fig:F222}. 
The $1$-skeleton $X_1$ is the disjoint union of the two sub $1$-skeletons, one of which includes the 0-cell $A$ and the other includes $B$. 
The $K$-group $K^G_0(\R^3) \cong \Z \times \Z_2$ is generated by two inequivalent atomic insulators $\ket{A}$ and $\ket{B}$ defined by the projective irrep.\ of $D_2$ located at $A$ and $B$, respectively. 
The $\Z_2$ structure is from that the direct sum $\ket{A} \oplus \ket{A}$ of irrep.\ at $A$ is deformable to $\ket{B} \oplus \ket{B}$, the two irreps.\ at $B$.

\section{Conclusion}
\label{sec:conc}
In the present paper, we studied a mathematical structure behind SPT phases and LSM-type theorems protected by crystalline symmetry. 
Our approach is based on the same spirit of Kitaev's proposal that the family of invertible states forms an $\Omega$-spectrum.~\cite{Kit13,Kit15}
Once an $\Omega$-spectrum is given, one can define the generalized (co)homology theory. 
In this paper, we proposed and demonstrated that the classification of a crystalline SPT phase is a generalized homology over the real space manifold on which a physical system is defined. 
This approach divides the problems of crystalline SPT phases into two aspects: SPT phases protected by onsite symmetry as building blocks and the role of crystalline symmetry. 
In the generalized homology description, onsite symmetry is inherited in the $\Omega$-spectrum, and the topological nature compatible with the crystalline symmetry is described by the mathematical structure of the equivariant generalized homology. 
In this sense, the generalized homology approach applies to any SPT phases, including fermionic systems. 
The underlying physical picture of our approach is that regarding topological phenomena, the correlation length of a bulk invertible state can be considered much smaller than the spatial length of crystalline symmetry, such as lattice translation~\cite{TE18}.

We have shown that the AHSS, the spectral sequence associated with a crystalline symmetric cell decomposition of the real space, is the perfect generalization of the prior developed machinery to classify SPT phases and LSM-type theorems in the presence of crystalline symmetry~\cite{Hermele_torsor,PWJZ17,HSHH17}. 
The AHSS successfully unifies various notions in crystalline SPT phases such as the layer construction, higher-order SPT phases~\cite{SchindlerHigher}, LSM theorems as the boundary of an SPT phase~\cite{PWJZ17}, and LSM-type theorem to enforce an SPT phase~\cite{YJVR17,Lu17}.

For free fermions, the generalized homology for free fermion SPT phases is attributed to the $K$-homology. 
It turns out that the AHSS for real space $K$-homology is quite complementary to that for momentum space $K$-cohomology~\cite{KSAHSS}.
As seen in Sec.~\ref{sec:free}, these AHSSs present different limiting pages that converge at the same $K$-group, which helps us to determine the $K$-group without explicitly solving the exact sequences among the limiting page.

Let us close by mentioning a number of future directions.
\begin{itemize}
\item
In this paper, we mainly focused on SPT phases on the infinite real space manifold. However, the generalized homology $h^G_n(X)$ is well-defined for any pairs of real space manifolds $(X,Y)$. 
For example, SPT phases on a sphere, Klein bottle, M\"obius strip, etc. 
It is interesting to explore the topological nature of SPT phases defined on topologically nontrivial real space manifolds which can be engineered. 

\item 
The physics of SPT phases gives us a practical definition to compute the differentials of the AHSS in generalized homology. 
It is also interesting to reinterpret known (co)homological definitions of the differentials in the AHSS for the $K$-theory and some cobordisms in the viewpoint of SPT physics. 

\item
The $\Omega$-spectrum structure of invertible states relates an adiabatic cycle or a kink texture which begins and ends at the trivial $(d+1)$-dimensional invertible state to a $d$-dimensional invertible state. 
Therefore, the $\Omega$-spectrum structure does not tell about the quantum number localized at a topological texture (kink, skyrmion, etc.) {\it within a nontrivial SPT phase} (See, for example, \cite{AbanovWiegmann}).  
In addition to the based loop space $\Omega F_{d+1}$, it should be of importance to understand the generic structure of the free loop space ${\cal L} F_{d+1} = \{\ell: S^1 \to F_{d+1}\}$ of $(d+1)$-dimensional invertible states, the space of adiabatic cycles which begin and end at an arbitrary $(d+1)$-dimensional invertible state. 

\item
For crystalline bosonic SPT phases without the $E_8$ phase as a building block, the AHSS of the corresponding generalized homology $h_0^G(\R^3)$ with a $3d$ space group $G$ is attributed to the strategy by Huang et al.~\cite{HSHH17}.
On the one hand, the classification result of Ref.~\cite{HSHH17} completely matches the cohomology theory $H^{4}_G(\R^3,U(1)^{\rm ori}) \cong H^{4}(BG,U(1)^{\rm ori})$ by Thorngren and Else~\cite{TE18}, where the equivariant cohomology $H^{4}_G(\R^3,U(1)^{\rm ori})$ is regarded as the classification of $G$ symmetric topological response theories over the real space manifold $\R^3$. 
This agreement suggests the existence of a sort of (twisted) generalized cohomology formulation of crystalline SPT phases over the real space $X$, which should be the Poincar\'e dual to the homological formulation. 
A possible route would be the homological AHSS based on the dual cell decomposition of the real space $X$ and reinterpreting the AHSS as a cohomological one. 

\end{itemize}

\medskip

\medskip

\noindent
{\it Acknowledgement---}
K.S.\ thanks 
Yohei Fuji, 
Yosuke Kubota
and
Yuji Tachiawa
for helpful discussions. 
K.G.\ is supported by JSPS Grant-in-Aid for Scientific Research on Innovative Areas "Discrete Geometric Analysis for Materials Design": Grant Number JP17H06462.

\medskip
\noindent
{\it Note added---}
In preparation for this work, we became aware of the following independent works. 
Refs.~\cite{DubinkinHughes18,YouHigher18} studied higher-order SPT phases in strongly interacting systems. 
Ref.~\cite{RasmussenLu18} studied the classification and construction of higher-order crystalline bosonic SPT phases.

\appendix

\section{The cohomology of $K(\Z_N, 2)$ in low degree}
\label{app:k(zn,2)}

The Eilenberg-MacLane space $K(\Z_N, 2)$ is a topological space which is, up to homotopy, uniquely characterized by its homotopy groups $\pi_2(K(\Z_N, 2)) \cong \Z_N$ and $\pi_i(K(\Z_N, 2)) = 0$ for $i \neq 2$. We here give some details of the computation of its integral cohomology group $H^n(K(\Z_N, 2); \Z)$ in low degree.


\subsection{Up to degree $3$}

The main strategy is to apply the Leray-Serre spectral sequence \cite{B-T} to the so-called path fibration $\Omega K(\Z_N, 2) \to PK(\Z_N, 2) \to K(\Z_N, 2)$. Here $P K(\Z_N, 2)$ is the space consisting of paths that start at a base point. Because the path space $P K(\Z_N, 2)$ is contractible, the based loop space $\Omega K(\Z_N, 2)$ is homotopy equivalent to the classifying space of principal $\Z_N$-bundles $K(\Z_N, 1) \simeq B\Z_N$. Its integral cohomology groups in low degree are $H^0(B\Z_N; \Z) = \Z$, $H^1(B\Z_N; \Z) = 0$, and $H^2(B\Z_N; \Z) = \Z_N$ and $H^3(B\Z_N; \Z) \cong 0$.

\begin{thm}
For positive integer $N$, we have
$$
\begin{array}{|c|c|c|c|c|c|}
\hline
& n = 0 & n = 1 & n = 2 & n = 3  \\
\hline
H^n(K(\Z_N, 2); \Z) & \Z & 0 & 0 & \Z_N \\
\hline
\end{array}
$$
\end{thm}

\begin{proof}
By the Hurewicz theorem \cite{Hat}, we have $H_0(K(\Z_N, 2); \Z) = \Z$, $H_1(K(\Z_N, 2); \Z) = 0$, and $H_2(K(\Z_N, 2); \Z) = \Z_N$. Applying the universal coefficient theorem \cite{Hat}, we conclude $H^0(K(\Z_N, 2); \Z) = \Z$, $H^1(K(\Z_N, 2); \Z) = 0$, $H^2(K(\Z_N, 2); \Z) = 0$. To determine $H^3(K(\Z_N, 2); \Z)$, we use the Lerray-Serre spectral sequence to the path fibration. Its $E_2$-term
$$
E_2^{p, q} = H^p(K(\Z_N, 2); H^q(B\Z_N; \Z))
$$
is summarized as follows.
$$
\begin{array}{c|c|c|c|c|}
\hline
q = 2 & \Z_N & 0 & & \\
\hline
q = 1 & 0 & 0 & 0 & 0 \\
\hline
q = 0 & \Z & 0 & 0 & \\
\hline
E_2^{p, q} & p = 0 & p = 1 & p = 2 & p = 3 \\
\end{array}
$$
From this, we find that $E_2^{0, 2} = E_3^{0, 2}$ and $E_2^{3, 0} = E_3^{3, 0}$. We also find that $E_4^{0, 2} = E_\infty^{0, 2}$ and $E_4^{3, 0} = E_\infty^{3, 0}$, both of which must be trivial, since $H^n(PK(\Z_N, 2); \Z) = 0$ for $n = 2, 3$. For this to be true, the differential $d_3 : E_3^{0, 2} \to E_3^{3, 0}$ must be an isomorphism, so that
$$
H^3(K(\Z_N, 2); \Z)
= E_2^{3, 0} = E_3^{3, 0} \cong E_3^{0, 2} \cong E_2^{0, 2} \cong \Z_N,
$$
and the proof is completed.
\end{proof}


\subsection{Degree $4$}

\begin{thm}
$H^4(K(\Z_N, 2); \Z) = 0$ for any $N > 0$.
\end{thm}

\begin{proof}
It is known that the cohomology ring of the classifying space $B\Z_N$ is the following quotient of the polynomial ring:
$$
H^*(B\Z_N; \Z) = \Z[u]/(Nu),
$$
where $u \in H^2(B\Z_N; \Z) = \Z_N$ is a generator. Then the $E_2$-term of the Leray-Serre spectral sequence for the path fibration
$$
E_2^{p, q} = H^p(K(\Z_N, 2); H^q(B\Z_N; \Z))
$$
can be summarized as follows:
$$
\begin{array}{c|c|c|c|c|c|}
\hline
q = 3 & 0 & 0 & 0 & 0 & 0 \\
\hline
q = 2 & \Z_N & 0 &  & & \\
\hline
q = 1 & 0 & 0 & 0 & 0 & 0 \\
\hline
q = 0 & \Z & 0 & 0 & \Z_N & \\
\hline
E_2^{p, q} & p = 0 & p = 1 & p = 2 & p = 3 & p = 4 \\
\end{array}
$$
It then turns out that $H^4(K(\Z_N, 2); \Z) = E_2^{4, 0} = E_\infty^{4, 0}$. To keep the consistency with the fact that $H^4(PK(\Z_N, 2); \Z) = 0$, we must have $E_\infty^{4, 0} = 0$. 
\end{proof}


\subsection{Degree $5$ and $6$}

\begin{lem}
For $N > 0$, the following holds true for $H^5(K(\Z_N, 2); \Z)$.
\begin{itemize}
\item
If $N$ is odd, then $H^5(K(\Z_N, 2); \Z) \cong \Z_N$.

\item
If $N$ is even, then $H^5(K(\Z_N, 2); \Z)$ is either $\Z_N \oplus \Z_2$ or $\Z_{2N}$.

\end{itemize}
\end{lem}

\begin{proof}
The $E_2$-term of the Leray-Serre spectral sequence for the path fibration reads:
$$
\begin{array}{c|c|c|c|c|c|c|}
q = 5 & 0 & 0 & 0 & 0 & 0 & 0  \\
\hline
q = 4 & \Z_N & 0 & \Z_N & \Z_N & & \\
\hline
q = 3 & 0 & 0 & 0 & 0 & 0 & 0  \\
\hline
q = 2 & \Z_N & 0 & \Z_N & \Z_N & & \\
\hline
q = 1 & 0 & 0 & 0 & 0 & 0 & 0  \\
\hline
q = 0 & \Z & 0 & 0 & \Z_N & 0 & \\
\hline
E_2^{p, q} & p = 0 & p = 1 & p = 2 & p = 3 & p = 4 & p = 5  \\
\end{array}
$$
For $p + q = 4$, the differential $d_2 : E_2^{p, q} \to E_2^{p+2, q-1}$ is trivial, so that $E_2^{p, q} = E_3^{p, q}$ in this case. Then we have possibly non-trivial differentials $d_3 : E_3^{0, 4} \to E_3^{3, 2}$ and $d_3 : E_3^{2, 2} \to E_3^{5, 0}$. It is straight to see that $d_3 : E_3^{2, 2} \to E_3^{5, 0}$ is injective. Therefore 
$$
E_5^{5, 0} = E_4^{5, 0} = \mathrm{Coker}[d_3 : E_3^{2, 2} \to E_3^{5, 0}].
$$
To compute $d_3 : E_3^{0, 4} \to E_3^{3, 2}$, recall that $\Z$ is a ring. Then, as a matter of fact about the Leray-Serre spectral sequence, the differential $d_3$ acts as a derivation. If $u \in H^2(B\Z_N; \Z) = \Z_N$ is a generator, then so is $u^2 \in H^4(B\Z_N; \Z) = \Z_N$. We can regard $u \in E_2^{0, 2} = E_2^{3, 0}$, so that $u^2 \in E_2^{0, 4} = E_3^{0, 4}$. We know that $d_3 : E_3^{0, 2} \to E_3^{3, 0}$ is an isomorphism. This implies that $ud_3(u) \in E_3^{3, 0}$ is a generator. We have $d_3(u^2) = d_3(u) u + ud_3(u) = 2 u d_3(u)$. Thus, if $N$ is odd, then $d_3 : E_3^{0, 4} \to E_3^{3, 2}$ is an isomorphism $\Z_N \cong \Z_N$, so that $0 = E_4^{0, 4} = E_5^{0, 4}$ and $E_\infty^{5, 0} = E_5^{5, 0} = E_4^{5, 0}$. Then $E_\infty^{5, 0} = 0$ implies that $d_3 : E_3^{2, 2} \to E_3^{5, 0}$ is isomorphic, and 
$$
H^5(K(\Z_N, 2); \Z) = E_2^{5, 0} = E_3^{5, 0} 
\cong E_3^{2, 2} = E_2^{2,2} = \Z_N.
$$
To the contrary, if $N$ is even, then $E_5^{0, 4} = E_4^{0, 4} = \Z_2$. To have $E_\infty^{5, 0} = E_6^{5, 0} = 0$, the map $d_5 : E_5^{0, 4} \to E_5^{5, 0}$ must be bijective. As a result, we have an exact sequence 
$$
0 \to \Z_N \to E_2^{5, 0} \to \Z_2 \to 0.
$$
The isomorphism classes of extensions of $\Z_2$ by $\Z_N$ are classified by the $\mathrm{Ext}$ group \cite{B-T} $\mathrm{Ext}(\Z_2, \Z_N) = \Z_2$. Hence $E_2^{5, 0} = H^5(K(\Z_N, 2); \Z)$ is $\Z_2 \oplus \Z_N$ or $\Z_{2N}$.
\end{proof}

To complete the computation of $H^5(K(\Z_N, 2); \Z)$, we appeal to the following fact \cite{Serre}.

\begin{prop} \label{prop:Z2_cohomology_ring}
Let $f \ge 1$. The cohomology ring of $K(\Z_{2^f}, 2)$ with coefficients in $\Z_2$ is the polynomial ring
$$
H^*(K(\Z_{2^f}, 2); \Z_2) 
\cong \Z_2[w_2, w_3, w_5, w_9, w_{17}, \ldots]
$$
generated by elements $w_d \in H^d(K(\Z_{2^f}, 2); \Z_2)$ of degree $d = 2 + \sum_{i = 0}^r2^i$ with $r = 0, 1, 2, \cdots$.
\end{prop}

\begin{thm}
The following holds true for $N > 0$
$$
H^5(K(\Z_N, 2); \Z)
= 
\left\{
\begin{array}{ll}
\Z_N, & (\mbox{$N$ odd}) \\
\Z_{2N}. & (\mbox{$N$ even})
\end{array}
\right.
$$
\end{thm}

\begin{proof}
The universal coefficient theorem gives
$$
H^4(K(\Z_N, 2); \Z_2)
\cong \mathrm{Tor}(H^5(K(\Z_N, 2); \Z), \Z_2).
$$
Suppose that $N$ is even, so that $N = 2^fq$ with $f$ positive and $q$ odd. Applying the universal coefficient theorem to the results of $H^n(K(\Z_q, 2); \Z)$ so far, we find:
$$
\begin{array}{|c|c|c|c|c|c|}
\hline
& n = 0 & n = 1 & n = 2 & n = 3 &  n = 4 \\
\hline
H^n(K(\Z_q, 2); \Z_2) & \Z_2 & 0 & 0 & 0 & 0 \\
\hline
\end{array}
$$
Applying the K\"{u}nneth formula to $K(\Z_N, 2) = K(\Z_{2^f}, 2) \times K(\Z_q, 2)$, we have
$$
H^4(K(\Z_N, 2); \Z_2) 
=
H^4(K(\Z_{2^f}, 2), \Z_2),
$$
which is $\Z_2$ by Proposition \ref{prop:Z2_cohomology_ring}. Thus, we conclude $H^5(K(\Z_N, 2); \Z) = \Z_{2N}$.
\end{proof}

We remark that $H^5(K(\Z_N, 2); \Z) \cong H^4(K(\Z_N, 2); \R/\Z)$ can be identified with the group of quadratic functions \cite{KT13}.

\begin{cor}
The following holds true for $N > 0$
$$
H^6(K(\Z_N, 2); \Z)
= 
\left\{
\begin{array}{ll}
0, & (\mbox{$N$ odd}) \\
\Z_2. & (\mbox{$N$ even})
\end{array}
\right.
$$
\end{cor}

\begin{proof}
Let us see the spectral sequence:
$$
\begin{array}{c|c|c|c|c|c|c|c|}
q = 5 & 0 & 0 & 0 & 0 & 0 & 0 & 0 \\
\hline
q = 4 & \Z_N & 0 & \Z_N & \Z_N & & & \\
\hline
q = 3 & 0 & 0 & 0 & 0 & 0 & 0 & 0 \\
\hline
q = 2 & \Z_N & 0 & \Z_N & \Z_N & & & \\
\hline
q = 1 & 0 & 0 & 0 & 0 & 0 & 0 & 0 \\
\hline
q = 0 & \Z & 0 & 0 & \Z_N & 0 & & \\
\hline
E_2^{p, q} & p = 0 & p = 1 & p = 2 & p = 3 & p = 4 & p = 5 & p = 6 \\
\end{array}
$$
For $E_\infty^{6, 0}$ to be killed, the homomorphism $d_3 : E_3^{3, 2} \to E_3^{6, 0}$ must be surjective. Also, since $0 = E_\infty^{3, 2} = E_4^{3,2}$, we have an exact sequence
$$
E_3^{0, 4} \overset{d_3}{\to} E_3^{3, 2} \overset{d_3}{\to} E_3^{6, 0}.
$$
Putting these results together, we have
\begin{align*}
E_2^{6, 0}
&=
E_3^{6, 0} \cong E_3^{3, 2}/\mathrm{Ker}[d_3 : E_3^{3, 2} \to E_3^{6, 0}]\\
&= E_3^{3, 2}/\mathrm{Im}[d_3 : E_3^{0, 4} \to E_3^{3, 2}].
\end{align*}
As is seen, if $N$ is odd, then $d_3 : E_3^{0, 4} \to E_3^{3, 2}$ is bijective, so that $E_2^{6, 0} = 0$. If $N$ is even, then the image of $d_3$ is $2\Z_N \subset \Z_N$, so that $E_2^{6, 0} = \Z_N/2\Z_N = \Z_2$.
\end{proof}

\bibliography{ref,bib_June2017,glide_spt_bib_01}
\bibliographystyle{unsrt} 

\end{document}